\documentclass[12pt]{article}
\usepackage[margin=1.0in]{geometry}
\usepackage[utf8]{inputenc}
\usepackage[intlimits]{amsmath}
\usepackage{amssymb,amsfonts,amsthm,dsfont}
\usepackage{physics}
\usepackage[english]{babel}
\usepackage[bottom]{footmisc}
\usepackage{csquotes}

\newtheorem{theorem}{Theorem}[section]
\newtheorem{lemma}[theorem]{Lemma}

\renewcommand{\d}{\mbox{d}}
\renewcommand{\bar}[1]{\overline{#1}}

\renewcommand{\dd}{\partial}

\numberwithin{equation}{section}

\begin{document}

\vspace{1.7 cm}

\begin{flushright}
{\small FIAN/TD/5-2026}
\end{flushright}
\vspace{1.7 cm}

\begin{center}
{\large\bf The $\sigma_-$ Cohomology Analysis for Coxeter higher spin $B_2$ model}

\vspace{1 cm}

{\bf A.A.~Tarusov and K.A.~Ushakov }\\
\vspace{0.5 cm}
\textbf{}\textbf{}\\
\vspace{0.5cm}
\textit{ I.E. Tamm Department of Theoretical Physics,
    Lebedev Physical Institute,}\\
\textit{ Leninsky prospect 53, 119991, Moscow, Russia}\\
\end{center}

\vspace{0.4 cm}

\begin{abstract}
\noindent
The dynamical content of equations resulting from rank-two covariant derivatives in $B_2$ Coxeter theory in $AdS_4$ are analyzed in terms of $\sigma_-$-complexes. Primary fields and gauge-invariant differential operators on primary fields are classified for $(adj \otimes adj)$ one-form fields $\omega$ and $(tw\otimes adj)$ zero-form fields $C$. It is shown that one-forms $\omega$ in the $(adj \otimes adj)$ sector encode symmetric massless fields and partially massless fields of all spins and depth of masslessness. Gluing of the one-form module to the zero-form modules at the linear vertices is studied.
\end{abstract}

\newpage
\vspace{-1cm}
\tableofcontents
\newpage

\section{Introduction}

Higher-spin (HS) gauge theories are consistent interacting theories containing massless fields of all spins. However, for the fields with spin above $2$ to have interactions, they must occur on a curved background as the restrictions of no-go theorems \cite{Coleman:1967ad, Haag:1974qh} are circumvented by abandonment of a flat space-time. The most symmetric curved background is $AdS$ \cite{Fradkin:1987ks, Fradkin:1986qy, Vasiliev:1988sa}. A particular interest to us is the 4-dimensional case, the lowest dimensional case with propagating fields.

As the spectrum of HS theories involves infinite towers of massless fields of arbitrary spins \cite{Konshtein:1988yg, Vasiliev:2004cm}, these theories are somewhat analogous to String theory whose spectrum also includes particles of all spins. The relation between the HS theory and the String theory has been discussed in literature for some time. Some calculations in limit cases indicate that String theory possesses HS symmetry at high energy \cite{Gross:1987ar,Gross:1988ue, Lindstrom:2003mg, Bonelli:2003kh, Sagnotti:2003qa}, but a direct relation between the two theories has not yet been established. It has been conjectured, however, that the multiparticle $B_2$ Coxeter extended HS (CHS) theory can provide a realization of the most symmetric phase of String theory \cite{Vasiliev:2018zer}. After a symmetry breaking the originally massless fields are expected to gain mass, thus potentially matching a String theory spectrum. The research into CHS has naturally required the analysis of dynamical content of such systems.
    
Nonlinear HS theory is described as a consistent nonlinear system of equations on zero-forms $C$ and one-forms $\omega$ in the so-called unfolded approach \cite{Vasiliev:1988sa, Vasiliev:1988xc} (see \cite{Joung:2021bhf, Iazeolla:2025btr, Misuna:2024ccj, Misuna:2024dlx} for recent applications of unfolding technique). Due to this compact formalism, not all fields in the system are independent primary fields, as some are their derivatives. The set of primary fields, independent field equations and differential gauge symmetries can be found as cohomology of a specific operator $\sigma_-$ \cite{Shaynkman:2000ts}. This $\sigma_-$ cohomological analysis has been carried out in a wide variety of different HS setups \cite{Shaynkman:2000ts, Vasiliev:2001zy, Gelfond:2003vh, Bekaert:2004qos, Boulanger:2008up, Boulanger:2008kw, Vasiliev:2009ck, Spirin:2024zgy, Tatarenko:2025krq}. However, in the context of $B_2$ CHS, the relevant results are those concerning $H(\sigma_-)$ in the standard $4d$ theory of HS \cite{Bychkov:2021zvd}, partially massless systems \cite{Skvortsov:2009nv}, as well as $\sigma_-$ analysis of $\mathfrak{sp}(2M)$ invariant equations in spaces with matrix coordinates \cite{Gelfond:2013lba}, which include usual $4d$ Minkowski space as a particular case. In \cite{Tarusov:2025sre} it was conjectured that $(adj \otimes adj)$ one-forms encode partially massless fields \cite{Deser:2001pe,Deser:2001us,Zinoviev:2001dt,Zinoviev:2002ye,Skvortsov:2006at} (see \cite{Grigoriev:2020lzu, Basile:2022mif, Basile:2024dcs, Zinoviev:2024xta, Zinoviev:2025nlp, Zinoviev:2025tpx, Zinoviev:2025yqa, Zinoviev:2025jff} for recent advancements in studying partially massless HS theories), as well as the usual massless fields \cite{Fronsdal:1978rb}, i.e., 2-particle field $\omega$ contains an infinite number of copies of 1-particle fields described by at most two-row Young diagrams of $\mathfrak{so}(3,2)$. If the conjecture is valid, then the $\sigma_-$ analysis of $(adj \otimes adj)$ sector should lead to the cohomology of \cite{Bychkov:2021zvd, Skvortsov:2009nv}. In \cite{Gelfond:2013lba} dynamical content of covariant constancy equations that correspond to the rank-$r$ tensor products of a standard twisted-adjoint modules has been studied. Therefore, there is no need to examine the content of $(tw\otimes tw)$ and the novel entangled (which are not isomorphic to a tensor product of modules from standard theory) one-form equations, since the entangled sector has been shown \cite{Tarusov:2025sre} to be related to $(tw \otimes tw)$. It was found that $H^2(\sigma_-)$ is zero, so $(tw\otimes tw)$ and entangled one-forms are topological. Therefore it should be possible for linear vertices presented in \cite{Tarusov:2025sre} to be eliminated by a suitable field redefinitions in both cases.

In this paper we study the dynamical content of $B_2$ CHS theory encoded in $(adj \otimes adj)$, $(tw \otimes adj)$ and $(adj \otimes tw)$ modules. We find primary fields and gauge-invariant differential operator on non-trivial primary fields as cohomology of the $\sigma_-$ operator. As the correct definition of $\sigma_-$ requires it to act on $\mathfrak{so}(3,2)$ indecomposable (irreducible if possible) modules, we consider highest(lowest) weight vectors of a commuting $\mathfrak{sl}(2)\oplus \mathfrak{sl}(2)$ algebra (which in the $(adj \otimes adj)$ case forms a Howe-dual pair with $\mathfrak{so}(3,2)$ \cite{Howe:1989}). This construction reveals a highly intricate structure of zero-form modules $(tw \otimes adj)$ and $(adj \otimes tw)$, which can, however, be greatly simplified by the truncation to unitary submodules. A $\sigma_-$ operator defined on $\mathfrak{so}(3,2)$ indecomposable representations allows a direct computation of cohomology on the highest(lowest) weight vectors of the dual algebra.

The cohomological analysis of the $(adj\otimes adj)$ sector shows that 1-cocycles describe symmetric massless and partially massless fields, and 2-cocycles are generalizations of the Weyl cocycle (Chevalley-Eilenberg cocycle that allows mixing of one-form and zero-form fields) and cocycles corresponding to the Fronsdal equations and the equations of a partially massless field. Thus, the conjecture that 2-particle $(adj\otimes adj)$ field encodes copies of all $AdS_4$ 1-particle fields is proven. For zero-form $C$ associated with $(tw \otimes adj)$ and $(adj \otimes tw)$ the lowest cohomology groups $H^{0,1}(\sigma_-)$ are obtained. A unitary truncation results in a standard generalized Weyl tensors and Bianchi 1-cocycles for symmetric massless fields coupled with a scalar 0-cocycle and Klein-Gordon 1-cocycle, as expected. The list of $0-$ and $1-$cocycles in a non-truncated case shows a presence of cocycles that can be associated with generalized Weyl tensors and Bianchi 1-cocycles for symmetric massless and partially massless fields, as well as some additional cocycles whose meaning is yet to be explored. Moreover, in the course of the cohomology analysis, it has been found that with respect to the algebra $\mathfrak{sl}(2)\oplus \mathfrak{sl}(2)$, commuting with the action of $\mathfrak{so}(3,2)$, the modules of zero-forms decompose into the sum of irreducible and reducible indecomposable representations. We identify the reducible indecomposable $\mathfrak{sl}(2)\oplus \mathfrak{sl}(2)$ modules with the non-split extension modules \cite{Humphreys2008} (similar structure was discussed in \cite{Shaynkman:2004vu} with respect to subsingular and subsubsingular modules of conformal algebra $\mathfrak{o}(M,2)$). As $\mathfrak{sl}(2)\oplus \mathfrak{sl}(2)$ and $\mathfrak{so}(3,2)$ are commuting algebras, presence of reducible indecomposable representations of one algebra may induce the same structures for the other. However, this issue requires the application of advanced representation theory and will be discussed in more detail elsewhere.

Gluing of one- and zero-forms performed by the linear vertices in the $(adj\otimes adj)$ sector is discussed. Comparison of the vertices with 2-cocycles of $\sigma_-$ in the $(adj\otimes adj)$ sector shows that the Weyl 2-cocycles are present in all one-form equations while 2-cocycles responsible for the field equations do not appear only if $\omega$ is built out of monomials with equal number of $Y_1$ and $Y_2$ oscillators. We determine that the standard homotopy with a zero shift used in \cite{Tarusov:2025sre} produce vertices with additional non-cohomological terms. The resulting vertices allow a straightforward unitary truncation, but interpretation in a non-truncated case is complicated. A field redefinition can be carried out so that the linear vertices consist of only 2-cocycles, as shown in the simplest examples. After the unitary truncation all equation 2-cocycles and most of the Weyl cocycles vanish, the system is found to be on-shell. Truncated $B_2$ CHS theory encodes several copies of dynamical Fronsdal fields and an infinite number of copies of topological massless and partially massless fields in accordance with the conclusions of \cite{Tarusov:2025sre}.

The paper is organized as follows. In section \ref{sec: Coxeter HS model} we recall the main concepts for a general CHS following \cite{Vasiliev:2018zer}, provide a succinct summary of the analysis of $B_2$ CHS in $AdS_4$ performed in \cite{Tarusov:2025sre}, give a definition of a $\sigma_-$-complex in a general unfolded setup and explain (following \cite{Shaynkman:2000ts}) the physical interpretation of cohomology $H^\bullet(\sigma_-)$ of an unfolded systems. Then in sections \ref{sec: one-form adj-adj}, \ref{sec: zero-form tw-adj} we carry out the cohomological analysis for $(adj \otimes adj)$ and $(tw \otimes adj)$ modules. In section \ref{sec: linear vertext analysis} the gluing of one-forms $(adj \otimes adj)$ with zero-forms $(tw \otimes adj)$ and $(adj \otimes tw)$ is examined. Obtained results are
summarized in section \ref{sec: conclusion}. We discuss a non-split extension modules of $\mathfrak{sl}(2)$ and $\mathfrak{sl}(2)\oplus \mathfrak{sl}(2)$ in appendix \ref{sec: nonsplit extension}. In appendix \ref{sec: momentum action} we show how $\mathfrak{so}(3,2)$ momentum operator acts on the lowest-weight $\mathfrak{sl}(2)\oplus \mathfrak{sl}(2)$ vectors in $(tw \otimes adj)$ module. Equations for components of zero-forms $C$ in a rank-one like decomposition are presented in appendix \ref{sec: rank-one zero-form eq} for $(adj \otimes tw)$ and $(tw\otimes adj)$ modules.

\section{Coxeter HS model and $\sigma_-$ cohomology}\label{sec: Coxeter HS model}

In this section we provide a brief introduction into the general CHS framework \cite{Vasiliev:2018zer}, recollect the analysis of $B_2$ CHS \cite{Tarusov:2025sre} and outline the concept of $\sigma_-$ cohomology \cite{Shaynkman:2000ts, Bekaert:2004qos}. 

\subsection{Coxeter groups and framed Cherednik algebra}
 
A rank-$p$ Coxeter group $\mathcal{C}$ is generated by reflections with respect to a system of root vectors $v_a$ in a $p$-dimensional Euclidean vector space $V$ with the scalar product $(x, y) \in \mathds{R}$, $x, y \in V$. An elementary reflection associated with the root vector $v_a$ acts on $x \in V$ as follows
\begin{equation}
    R_{v_a}x^i = x^i - 2 \frac{(v_a,x)}{(v_a,v_a)}v_a^i\,, \quad R_{v_a}^2 = Id\,.
\end{equation}

The root system of $B_2$ consists of two conjugacy classes under the action of $B_2$
\begin{equation}\label{e: B2 roots}
    \mathcal{R}_1 = \{\pm e^i, 1 \leq i \leq 2\}\,, \quad \mathcal{R}_2 = \{\pm e^1 \pm e^2\}\,.
\end{equation}
Group $B_2$ contains reflections of any basis axis in $V = \mathds{R}^2$ generated by $v^i_\pm = \pm e^i$ and permutations generated by $v^{12} = e^1-e^2$.

CHS theories are based on a framed Cherednik algebra generated via a set of idempotents $I_n$, a set of oscillators $q_\alpha^n$ and Klein operators $\hat{K}_v$ for each root vector $v$ (here $\alpha \in \{1,2\}, n\in\{1,\dots,p\}$), that obey
\begin{equation}\label{e: qIK1}
     I_n I_m = I_m I_n\,, \quad I_n I_n = I_n\,, \quad I_n q_\alpha^n = q_\alpha^n I_n = q_\alpha^n\,, \quad I_m q_\alpha^n = q_\alpha^n I_m\,
\end{equation}
with no summation over repeated Latin indices and
\begin{alignat}{2}\label{e: qIK2}
    \hat{K}_v q_\alpha^n &= R_v{}^n{}_m q_\alpha^m \hat{K}_v\,, &\hat{K}_v \hat{K}_u &= \hat{K}_u \hat{K}_{R_u(v)} = \hat{K}_{R_v(u)}\hat{K}_v\,,
    \\ 
    \label{e: qIK3}
    \hat{K}_v \hat{K}_v &= \prod I_{i_1(v)}\dots I_{i_k(v)}\,, &\hat{K}_v &= \hat{K}_{-v}\,,
    \\
\label{e:cherednik comm}
    [q_\alpha^n,q_\beta^m] &= -i \varepsilon_{\alpha\beta}\left(2\delta^{nm}I_n + \sum_{v \in \mathcal{R}} \nu(v) \frac{v^n v^m}{(v,v)} \hat{K}_v\right)\,,
\end{alignat}
where $\mathcal{R}$ is a set of conjugacy classes of root vectors under the action of  $\mathcal{C}$, $\nu(v)$ is a function of conjugacy classes and the labels $i_1(v), . . . , i_k(v)$ enumerate those idempotents $I_n$ that carry labels affected by the reflection $R_v$. In case of $B_2$ there are two types of Klein operators: $\hat K_{12}$ corresponding to the root vector $v^{12}$ and $\hat K_i$ corresponding to the vector $v_+^i$.
Klein operators $\hat{K}_v$ obey
\begin{align}\label{e:KI comm}
    I_n \hat{K}_v &= \hat{K}_v I_n\,, \forall n  \in \{1,\dots,p\}\,,
    \\
\label{e:KI absorb}
    I_{n} \hat{K}_v &= \hat{K}_v I_n= \hat{K}_v \,, \forall n  \in \{i_1(v), . . . , i_k(v)\}\,.
\end{align}
The double commutator of $q_\alpha^n$ satisfies Jacobi identity (see \cite{Vasiliev:2018zer} for details).

For any Coxeter root system the generators
\begin{equation}
    t_{\alpha\beta} = \frac{i}{4}\sum_{n=1}^{p}\{q^n_\alpha, q^n_\beta\}I_n\,
\end{equation}
obey the $\mathfrak{sp}(2)$ commutation relations
\begin{equation}
\left[t_{\alpha \beta}, t_{\gamma \delta}\right]=\epsilon_{\beta \gamma} t_{\alpha \delta}+\epsilon_{\beta \delta} t_{\alpha \gamma}+\epsilon_{\alpha \gamma} t_{\beta \delta}+\epsilon_{\alpha \delta} t_{\beta \gamma}\,,
\end{equation}
properly rotating all Greek indices,
\begin{equation}\label{e: sl2 invariance}
    [t_{\alpha\beta}, q^n_\gamma] = \epsilon_{\beta\gamma} q^n_\alpha + \epsilon_{\alpha\gamma} q^n_\beta\,.
\end{equation}
These properties hold for any $\nu(v)$ constants chosen in the commutation relations (\ref{e:cherednik comm}). The fact that the framed Cherednik algebra possesses the $\mathfrak{sp}(2)$ automorphism algebra implies that nonlinear CHS equations are also $\mathfrak{sp}(2)$ covariant and, therefore, implies usual Lorentz covariance \cite{Vasiliev:1999ba}.

\subsection{General CHS theory}\label{sec: general chs}

Consider the fields $W$, $S$ and $B$ depending on the space-time variable $x$ as well as on $p$ sets of variables enumerated by the label $n \in \{1,\dots,p\}$, that include $Y^n_\mathcal{A} , Z^n_\mathcal{A}$ $(\mathcal{A} \in \{1,\dots,4\})$, idempotents $I_n$, anticommuting differentials $dZ^\mathcal{A}_n$ and Klein operators $\hat{K}_v$ associated with all root vectors of a chosen Coxeter group $\mathcal{C}$ (at the convention $\hat K_{-v} = \hat K_v)$. The field $W(Y,Z,I;\hat{K}|x)$ is a $dx$ one-form,  $S(Y,Z,I;\hat{K}|x)$ is a $dZ$ one-form and  $B(Y,Z,I;\hat{K}|x)$ is a zero-form. The field equations associated with the framed Cherednik algebra (\ref{e:cherednik comm}) are formulated in terms of the star product

\begin{align}
\label{e:star product}
    (f*g)(Y,Z,I) &= \frac{1}{(2\pi)^{4 p}}\int d^{4p}S d^{4p}T \exp\bigg(iS^\mathcal{A}_n T^\mathcal{B}_m C_{\mathcal{A}\mathcal{B}}\delta^{nm}\bigg) \times \nonumber
    \\
    &\times f(Y_i+I_i S_i,Z_i+I_i S_i,I)g(Y+T,Z-T,I)\,,
\end{align}

where
\begin{equation}
    C_{AB} = \begin{pmatrix}
\varepsilon_{\alpha\beta} & 0 \\
0 & \bar\varepsilon_{\dot \alpha \dot \beta}
\end{pmatrix}\,.
\end{equation}
The spinor indices are raised and lowered by the Lorentz invariant antisymmetric tensors $\varepsilon^{\alpha\beta}$ and $\bar\varepsilon^{\dot\alpha\dot\beta}$ according to the rules
\begin{equation}
    A^\alpha = \varepsilon^{\alpha\beta}A_\beta\,, \quad A_\beta = \varepsilon_{\alpha\beta}A^\alpha\,, \quad A^{\dot\alpha} = \bar\varepsilon^{\dot\alpha\dot\beta}A_{\dot\beta}\,, \quad A_{\dot\beta} = \bar\varepsilon_{\dot\alpha\dot\beta}A^{\dot\alpha}\,.
\end{equation}

Central elements $I_n$ obey (no summation over repeated indices)
\begin{alignat}{3}
    Y^m_\mathcal{A} * I_n &= I_n * Y^m_\mathcal{A}\,, &\quad Y^n_\mathcal{A} * I_n &= I_n * Y^n_\mathcal{A} = Y^n_\mathcal{A}\,, &\quad Z^m_\mathcal{A} * I_n &= I_n * Z^m_\mathcal{A}\,, \\
    Z^n_\mathcal{A} * I_n &= I_n * Z^n_\mathcal{A} = Z^n_\mathcal{A}\,, &\quad I_n*I_n &= I_n\,, &\quad I_n * I_m &= I_m * I_n\,.
\end{alignat}

Analogously to the standard HS construction, CHS star product admits inner Klein operators $\varkappa_v$, $\bar{\varkappa}_v$ associated with
the root vectors $v$
\begin{equation}
    \varkappa_v = \exp\bigg(i \frac{v^n v^m}{(v,v)} z_{\alpha n}y^\alpha_m \bigg)\,, \quad \bar\varkappa_v = \exp\bigg(i \frac{v^n v^m}{(v,v)} \bar z_{\dot\alpha n}\bar y^{\dot\alpha}_m \bigg)\,.
\end{equation}

The inner Klein operators $\varkappa_v$ generate the star product realization of $\mathcal{C}$ via
\begin{equation}
    \varkappa_v * q_\alpha^n = R_v{}^n{}_m q_\alpha^m * \varkappa_v\,, \quad q_\alpha^n = y_\alpha^n, z_\alpha^n\,,
\end{equation}
and analogously for $\bar{q}_{\dot\alpha}$.

Nonlinear equations of general CHS theory are \cite{Vasiliev:2018zer}
\begin{align}
    \d_x W + W*W &= 0\,, \label{e:nonlinear system 1}
    \\
    \d_x B + W*B - B*W &= 0\,, \label{e:nonlinear system 2}
    \\
    \d_x S + W*S + W*S &= 0\,, \label{e:nonlinear system 3}
    \\
    S*B - B*S &= 0\,,  \label{e:nonlinear system 4}
    \\
    S*S &= i \bigg(dZ^{\mathcal{A}n}dZ_{\mathcal{A}n} + \sum_i \sum_{v \in \mathcal{R}_i} \bigg[\eta_{i} B \frac{v^n v^m}{(v,v)}dz^\alpha_n dz_{\alpha m} * \varkappa_v \hat{k}_v  + \nonumber\\
    &+\bar{\eta}_{i}B \frac{v^n v^m}{(v,v)}d\bar z^{\dot\alpha}_n d\bar z_{\dot\alpha m} * \bar\varkappa_v \hat{\bar k}_v\bigg]\bigg)\,,  \label{e:nonlinear system 5}
\end{align}

where $\varkappa_v \hat{k}_v$ acts on $dz^\alpha_n$ as
\begin{equation}
    \varkappa_v\hat{k}_v * dz^\alpha_n = R_v{}_n{}^m dz^\alpha_m *  \varkappa_v\hat{k}_v\,.
\end{equation}
It is important to note that 
\begin{equation}
    \hat{\gamma}_i = \sum_{v \in \mathcal{R}_i} \frac{v^n v^m}{(v,v)}dz^\alpha_n dz_{\alpha m} * \varkappa_v \hat{k}_v
\end{equation}
and its conjugated counterpart $\hat{\bar{\gamma}}_i$ are central with respect to the star product (\ref{e:star product}).

As shown in \cite{Tarusov:2025sre}, the vacuum solution of the nonlinear system (\ref{e:nonlinear system 1})--(\ref{e:nonlinear system 5}) for any $\mathcal{C}$ that describes $AdS_4$ is
\begin{align}
    B_0 &= 0\,, \quad S_0 = dZ^{\mathcal{A}n}Z_{\mathcal{A}n}\,, \\
    W_0 &= \Omega_{AdS}(Y|x) = -\frac{i}{4}\delta^{nm}\bigg(\omega_{\alpha\beta}(x)y^\alpha_n y^\beta_m + \bar \omega_{\dot \alpha \dot \beta}(x) \bar y^{\dot \alpha}_n \bar y^{\dot \beta}_m + 2 e_{\alpha \dot \alpha}(x) y^\alpha_n \bar y^{\dot \alpha}_m\bigg)\,,
\end{align}
where the $AdS_4$ spin-connections $\omega_{\alpha\beta}, \bar\omega_{\dot\alpha\dot\beta}$ and vierbein $e_{\alpha \dot \alpha}$ obey
\begin{align}
    \d_x \omega_{\alpha\beta} + \omega_{\alpha\gamma}\wedge \omega^{\gamma}{}_{\beta}  +  e_{\alpha\dot \gamma}\wedge e_{\beta}{}^{\dot \gamma} &= 0\,, \\
    \d_x \bar\omega_{\dot\alpha\dot\beta} + \bar\omega_{\dot\alpha\dot\gamma}\wedge\bar{\omega}^{\dot\gamma}{}_{\dot\beta} + e_{\gamma\dot \alpha}\wedge e^{\gamma}{}_{\dot \beta} &= 0\,, \\
    \d_x e_{\alpha \dot \alpha} +   \omega_{\alpha\gamma} \wedge e^{\gamma}{}_{\dot \alpha} + \bar\omega_{\dot\alpha\dot\gamma} \wedge e_{\alpha}{}^{\dot \gamma}  &= 0\,.
\end{align}

Decomposition of the nonlinear CHS equations over this vacuum solution leads to the following set of linear equations
\begin{align}\label{e: zero-form equation}
    \mathcal{D}C(Y,I;\hat{K}|x) &= 0\,,
    \\
    \label{e: one-form equation}
    \mathcal{D} \omega(Y,I;\hat{K}|x) &= \Upsilon^{\eta}(\Omega_{AdS},\Omega_{AdS},C) + \Upsilon^{\bar{\eta}}(\Omega_{AdS},\Omega_{AdS},C)\,.
\end{align}

Here covariant derivative
\begin{equation}\label{e: covariant derivative}
    \mathcal{D} = D_L + e^{\alpha \dot\alpha}[P_+^{kl}(\hat{k},\hat{\bar k})(y_{\alpha k} \hat{\bar\dd}_{\dot\alpha l} + \bar y_{\dot\alpha l} \hat{\dd}_{\alpha k}) - i P_-^{kl}(\hat{k},\hat{\bar k})(y_{\alpha k}\bar y_{\dot\alpha l} -  \hat\dd_{\alpha k}\hat{\bar\dd}_{\dot\alpha l})]
\end{equation}
is composed of Lorenz covariant derivative
\begin{equation}\label{e: DL}
    D_L  =  \d_x  + \delta^{nm}\bigg(\omega^{\alpha \beta} y_{\alpha n}\dd_{\beta m} + \bar \omega^{\dot\alpha \dot\beta} \bar y_{\dot\alpha n}\bar\dd_{\dot\beta m}\bigg)\,
\end{equation}
and the vierbein-dependent terms. The latter involve matrix operator, which acts differently on the component fields specified by $\hat{k}$ and $\hat{\bar{k}}$, 
\begin{equation}\label{e: projector formula}
    P_{\pm}^{kl}(\hat{k},\hat{\bar k})=\frac{1}{2}\delta^{nm}\bigg(\mathds{1}^k_n \bar{\mathds{1}}^l_m \pm R(k){}^k_n \bar{R}(\bar{k}){}^l_m \bigg)\,,
\end{equation}
where $\mathds{1}^k_n$ and $\bar{\mathds{1}}^l_m$ are identity matrices, $\hat{k}$ and $\hat{\bar k}$ are products of elementary Klein operators $\hat{k}_v$ and $\hat{\bar k}_v$, matrices $R(k){}^k_n$ and $\bar{R}(\bar{k}){}^l_m$ are reflections in the root space which correspond to the products of elementary Klein operators $\hat{k}$ and $\hat{\bar k}$ (identity matrix in case of $\hat{k} = 1$) that are the arguments of the fields $C(Y,I;\hat{K}|x)$ and $\omega(Y,I;\hat{K}|x)$ on which the derivative $\mathcal{D}$ acts. In this paper, only the $\hat{K}$ implies a power series of elementary Klein operators $\hat{k}$ and $\hat{\bar{k}}$, while explicitly written $\hat{k}$ or $\hat{\bar{k}}$ are used to denote Klein dependence in component fields. We also introduce a notation for derivatives accompanied by idempotent $I_n$
\begin{equation}
    \hat{\dd}_{\alpha n} := I_n \dd_{\alpha n}\,.
\end{equation}

For the explicit form of vertices $\Upsilon^{\eta}(\Omega_{AdS},\Omega_{AdS},C)$ we refer the reader to \cite{Tarusov:2025sre}.

\subsection{$B_2$ CHS}\label{B_2CHS}

The root system of $B_2$ is represented in (\ref{e: B2 roots}). A generating set of $B_2$ is $\{R_{e^i}, R_{e^1-e^2}\}$ reflections of two-dimensional root space. Holomorphic Klein operators associated with the generating reflections are $\hat{k}_i$ and $\hat{k}_{12}$. After identification of unity $1$ with idempotents $I_n$ the holomorphic group $B_2$ is generated by
\begin{equation}
   \Bigl\{ \hat{k}_i, \hat{k}_{12}| \hat{k}_i^2 = I_i, \hat{k}_{12}^2 = I_1 I_2, \hat{k}_1 \hat{k}_{12} = \hat{k}_{12} \hat{k}_2, \hat{k}_2 \hat{k}_{12} = \hat{k}_{12} \hat{k}_1, i \in \{1,2\} \Bigl\} \,.
\end{equation}

It is useful to use notation
\begin{equation}
\hat{k}^+_{12} := \hat{k}_1 \hat{k}_2 \hat{k}_{12}\,
\end{equation}
and view $\hat{k}^+_{12}$ as an additional redundant generator that corresponds to the root vector $e^1+e^2 \in \mathcal{R}_2$ (reflection with respect to $e^1+e^2$ is equivalent to the composition of reflections with respect to $e^1-e^2$ and basis vectors $e^i$).
The reflection matrices $R(k)$ in (\ref{e: projector formula}) are
\begin{alignat}{3}
R(1) &= \begin{pmatrix}
1 & 0 \\
0 & 1
\end{pmatrix}\,, &\quad
R(k_1) &= \begin{pmatrix}
-1 & 0 \\
0 & 1
\end{pmatrix}\,, &
R(k_2) &= \begin{pmatrix}
1 & 0 \\
0 & -1
\end{pmatrix}\,, \label{e: B2 R first}\\
R(k_{12}) &= \begin{pmatrix}
0 & 1 \\
1 & 0
\end{pmatrix}\,, &
R(k_1 k_2) &= \begin{pmatrix}
-1 & 0 \\
0 & -1
\end{pmatrix}\,, &\quad
R(k_1 k_{12}) &= \begin{pmatrix}
0 & -1 \\
1 & 0
\end{pmatrix}\,,\\
R(k_2 k_{12}) &= \begin{pmatrix}
0 & 1 \\
-1 & 0
\end{pmatrix}\,, &
R(k^+_{12}) &= \begin{pmatrix}
0 & -1 \\
-1 & 0
\end{pmatrix}\,.
\label{e: B2 R last}
\end{alignat}

Analogous matching of the reflection matrices $\bar R(\bar k)$ takes place for the anti-holomorphic Klein operators $\bar k_i, \bar k_{12}$.

As was explained in \cite{Tarusov:2025sre}, idempotents $I_n$ induce filtration and decompose the CHS system (\ref{e:nonlinear system 1})--(\ref{e:nonlinear system 5}) into the sectors $I_{1,2}$ and $I_1 I_2$ in a triangular way. Sectors $I_{i}$ reproduce a standard HS theory and contribute to the mixed $I_1 I_2$ sector, where new fields and corresponding interaction vertices appear. As idempotents have no inverse elements, these new objects do not reappear in the $I_i$ sectors, leaving them identical to the standard theory. 

The expansion in powers of Klein operators for a general field in the mixed idempotent sector $C(Y_1, Y_2; \hat{K}\vert x)*I_1 I_2$ contains $2^6=64$ component fields
\begin{equation}
    C(Y_1, Y_2; \hat{K}\vert x)*I_1 I_2 = \sum_{a,b,c,\bar{a},\bar{b},\bar{c} = 0}^1 C_{abc\bar{a}\bar{b}\bar{c}} (Y_1, Y_2\vert x) *I_1 I_2 * \hat k_1^a \hat k_2^b \hat k_{12}^c \hat{\bar k}{}_1^{\bar a}\hat{\bar k}{}_2^{\bar b}\hat{\bar k}{}_{12}^{\bar c}\,.
\end{equation}

However, these $64$ component fields are divided into $8$ groups of $8$ fields as the form of covariant derivative (\ref{e: covariant derivative}) depends not on Klein operators directly, but on matrix operators (\ref{e: projector formula}). We list all possible products of reflection matrices $R(k)$ and $\bar{R}(\bar{k})$. 
\begin{gather}
   \delta^{nm}R(k){}^k_n \bar{R}(\bar{k}){}^l_m = \Biggl\{ \begin{pmatrix}
1 & 0 \\
0 & 1
\end{pmatrix}\,,
\begin{pmatrix}
-1 & 0 \\
0 & 1
\end{pmatrix}\,,
\begin{pmatrix}
1 & 0 \\
0 & -1
\end{pmatrix}\,,
\begin{pmatrix}
-1 & 0 \\
0 & -1
\end{pmatrix}\,,
\begin{pmatrix}
0 & 1 \\
1 & 0
\end{pmatrix}\,, \label{e: RbarR1} \\
\begin{pmatrix}
0 & -1 \\
-1 & 0
\end{pmatrix}\,,
\begin{pmatrix}
0 & -1 \\
1 & 0
\end{pmatrix}\,,
\begin{pmatrix}
0 & 1 \\
-1 & 0
\end{pmatrix}\, \label{e: RbarR2}\Biggr\}.
\end{gather}

According to (\ref{e: covariant derivative}), this leaves us with 8 (note that $\text{ord} (B_2) = 8$) different covariant constancy equations whose solution spaces comprise the field modules for the $B_2$ CHS theory
\begin{align}
    \bigg(D_L + e^{\alpha \dot\alpha}\sum_{i = 1}^2 (y_{\alpha i}\bar\dd_{\dot \alpha i} + \bar y_{\dot\alpha i}\dd_{\alpha i})\bigg)C(Y_1,Y_2,I;\hat{k},\hat{\bar k}|x) &= 0\,, \label{e:cov1}
\end{align}

\medmuskip=3mu

\begin{align}
    \bigg(D_L - i e^{\alpha \dot\alpha}(y_{\alpha 1}\bar y_{\dot \alpha 1} - \dd_{\alpha 1} \bar \dd_{\dot\alpha 1}) + e^{\alpha \dot\alpha}(y_{\alpha 2}\bar\dd_{\dot \alpha 2} + \bar y_{\dot\alpha 2}\dd_{\alpha 2})\bigg)C(Y_1,Y_2,I;\hat{k},\hat{\bar k}|x) &= 0\,, \label{e:cov2}\\
    \bigg(D_L + e^{\alpha \dot\alpha}(y_{\alpha 1}\bar\dd_{\dot \alpha 1} + \bar y_{\dot\alpha 1}\dd_{\alpha 1}) - i e^{\alpha \dot\alpha}(y_{\alpha 2}\bar y_{\dot \alpha 2} - \dd_{\alpha 2} \bar \dd_{\dot\alpha 2}) \bigg)C(Y_1,Y_2,I;\hat{k},\hat{\bar k}|x) &= 0\,, \label{e:cov3}\\
    \bigg(D_L - i e^{\alpha \dot\alpha}\sum_{i = 1}^2 (y_{\alpha i}\bar y_{\dot \alpha i} - \dd_{\alpha i} \bar \dd_{\dot\alpha i})\bigg)C(Y_1,Y_2,I;\hat{k},\hat{\bar k}|x) &= 0\,, \label{e:cov4}\\
    \bigg(D_L + \frac{1}{2}e^{\alpha \dot \alpha}\bigg[(y_{\alpha 1} + y_{\alpha 2})(\bar\dd_{\dot \alpha 1} + \bar\dd_{\dot \alpha 2}) + (\bar y_{\dot \alpha 1} + \bar y_{\dot \alpha 2})(\dd_{\alpha 1} + \dd_{\alpha 2})\bigg] -\label{e:cov5} \\
    - \frac{i}{2}e^{\alpha \dot \alpha}\bigg[(y_{\alpha 1} - y_{\alpha 2})(\bar y_{\dot\alpha 1} - \bar y_{\dot\alpha 2}) - (\dd_{\alpha 1} - \dd_{\alpha 2})(\bar\dd_{\dot \alpha 1} - \bar\dd_{\dot \alpha 2})\bigg]\bigg)C(Y_1,Y_2,I;\hat{k},\hat{\bar k}|x) &= 0\,, \nonumber
\\
    \bigg(D_L + \frac{1}{2}e^{\alpha \dot \alpha}\bigg[(y_{\alpha 1} - y_{\alpha 2})(\bar\dd_{\dot \alpha 1} - \bar\dd_{\dot \alpha 2}) + (\bar y_{\dot \alpha 1} - \bar y_{\dot \alpha 2})(\dd_{\alpha 1} - \dd_{\alpha 2})\bigg] - \label{e:cov6}\\
    - \frac{i}{2}e^{\alpha \dot \alpha}\bigg[(y_{\alpha 1} + y_{\alpha 2})(\bar y_{\dot\alpha 1} + \bar y_{\dot\alpha 2}) - (\dd_{\alpha 1} + \dd_{\alpha 2})(\bar\dd_{\dot \alpha 1} + \bar\dd_{\dot \alpha 2})\bigg]\bigg)C(Y_1,Y_2,I;\hat{k},\hat{\bar k}|x) &= 0\,,  \nonumber
\\
    \bigg(D_L + \frac{1}{2}e^{\alpha \dot \alpha}\bigg[y_{\alpha 1}(\bar\dd_{\dot\alpha 1} - \bar\dd_{\dot\alpha 2}) + y_{\alpha 2}(\bar\dd_{\dot\alpha 1} + \bar\dd_{\dot\alpha 2}) + \bar y_{\dot\alpha 1}(\dd_{\alpha 1} + \dd_{\alpha 2}) - \bar y_{\dot \alpha 2}(\dd_{\alpha 1} - \dd_{\alpha 2})\bigg] - \label{e:cov7}\\
    - \frac{i}{2}e^{\alpha \dot \alpha}\bigg[y_{\alpha 1}(\bar y_{\dot\alpha 1} + \bar y_{\dot\alpha 2}) - y_{\alpha 2}(\bar y_{\dot\alpha 1} - \bar y_{\dot\alpha 2}) - \dd_{\alpha 1}(\bar\dd_{\dot\alpha 1} + \bar\dd_{\dot\alpha 2}) + \dd_{\alpha 2}(\bar\dd_{\dot\alpha 1} - \bar\dd_{\dot\alpha 2})\bigg]\bigg)\times\nonumber\\
    \times C(Y_1,Y_2,I;\hat{k},\hat{\bar k}|x) &= 0\,, \nonumber
\\
    \bigg(D_L + \frac{1}{2}e^{\alpha \dot \alpha}\bigg[y_{\alpha 1}(\bar\dd_{\dot\alpha 1} + \bar\dd_{\dot\alpha 2}) - y_{\alpha 2}(\bar\dd_{\dot\alpha 1} - \bar\dd_{\dot\alpha 2}) + \bar y_{\dot\alpha 1}(\dd_{\alpha 1} - \dd_{\alpha 2}) + \bar y_{\dot \alpha 2}(\dd_{\alpha 1} + \dd_{\alpha 2})\bigg] - \label{e:cov8}\\
    - \frac{i}{2}e^{\alpha \dot \alpha}\bigg[y_{\alpha 1}(\bar y_{\dot\alpha 1} - \bar y_{\dot\alpha 2}) + y_{\alpha 2}(\bar y_{\dot\alpha 1} + \bar y_{\dot\alpha 2}) - \dd_{\alpha 1}(\bar\dd_{\dot\alpha 1} - \bar\dd_{\dot\alpha 2}) - \dd_{\alpha 2}(\bar\dd_{\dot\alpha 1} + \bar\dd_{\dot\alpha 2})\bigg]\bigg) \times\nonumber\\
    \times C(Y_1,Y_2,I;\hat{k},\hat{\bar k}|x) &= 0\,, \nonumber
\end{align}

\medmuskip=4mu

corresponding to the matrices (\ref{e: RbarR1}), (\ref{e: RbarR2}) reading from left to right, from top to bottom. Each of the equations above has 8 component fields corresponding to it. It is worth noting that unhatted derivatives $\dd_{\alpha i}$ appear in the covariant constancy equations due to the properties (\ref{e: qIK1}) and (\ref{e:KI absorb}).

As discussed in \cite{Tarusov:2025sre}, the covariant derivatives (\ref{e:cov1})--(\ref{e:cov8}) give rise to four families of modules. The first three are tensor products of two adjoint modules, products of twisted-adjoint and adjoint modules, and products of two twisted-adjoint modules from the standard $4d$ HS theory that will be denoted as $(adj\otimes adj), (tw\otimes adj), (adj\otimes tw)$ and $(tw\otimes tw)$, respectively. These families are associated with the equations (\ref{e:cov1})--(\ref{e:cov6}). The fourth set of modules (\ref{e:cov7}) and (\ref{e:cov8}) cannot be represented as a tensor product of standard HS modules, and therefore we refer to them as entangled ones. The entangled modules admit an exponential ansatzes 
\begin{align}
    C(Y_1,Y_2,I;\hat{k},\hat{\bar k}|x) &= \exp\bigg(-i y_{1\alpha}y_2^{\alpha} + i \bar{y}_{1 \dot\alpha}\bar{y}_2^{\dot\alpha}\bigg) \tilde{C}(Y_1,Y_2,I;\hat{k},\hat{\bar k}|x)\,, \label{def:expAnsatz} \\
    C(Y_1,Y_2,I;\hat{k},\hat{\bar k}|x) &= \exp\bigg(i y_{1\alpha}y_2^{\alpha} - i \bar{y}_{1 \dot\alpha}\bar{y}_2^{\dot\alpha}\bigg) \tilde{C}(Y_1,Y_2,I;\hat{k},\hat{\bar k}|x)\,
\end{align}
which transform entangled covariant constancy equations for the fields $C$ into a covariant constancy equation for the fields $\tilde{C}$ that resemble the equation for the product of two twisted-adjoint modules. However, entangled modules are not isomorphic to $(tw\otimes tw)$ (see \cite{Tarusov:2025sre}).

The different characteristics of various modules, in particular their unitarizability, lead to the necessity of truncating the $B_2$ system. A consistent way to achieve this truncation in the full nonlinear CHS system is to leave only the invariant space of an automorphism $\hat{K}_v \rightarrow - \hat{K}_v$. This truncation restricts zero-form fields $C$ to $(adj\otimes tw)$ and $(tw\otimes adj)$ and one-form fields $\omega$ to $(adj\otimes adj)$, $(tw\otimes tw)$ and the entangled modules. As will be explained in section \ref{sec: Interpretation of cohomology}, covariant derivative (\ref{e: covariant derivative}) uniquely defines field theoretical pattern of the corresponding unfolded system, i.e., the spectrum of (primary) fields, differential gauge parameters, differential field equations and Bianchi identities. However, the question of whether primary fields that reside in one-forms $\omega$ are on-shell or off-shell depends not only on the content of corresponding one-form modules given by the covariant constancy equations, but also on the structure of vertices $\Upsilon(\Omega_{AdS},\Omega_{AdS},C)$ that serve as a gluing between one-form modules $\omega$ and zero-form modules $C$.

In \cite{Tarusov:2025sre} we have considered a physical content of $B_2$ CHS in case of zero-forms $C$ valued in a quotient by non-unitarizable adjoint factor of $(adj\otimes tw)$ and $(tw\otimes adj)$. Such a restriction leaves us with unitarizable zero-form modules that correspond to one-particle states described by the generalized Weyl tensors for symmetric massless fields. Brief analysis of $\Upsilon(\Omega_{AdS},\Omega_{AdS},C)$ has shown that several copies of Fronsdal fields are present in a spectrum of reduced system. The complete spectrum of fields was not known  and will be obtained in this paper via a $H^\bullet(\sigma_-)$ cohomology analysis. Moreover, knowledge of all encoded fields, differential gauge parameters and differential field equations will allow us to study the physical interpretation of the $B_2$ model with zero-forms valued in a complete, not unitary reduced $(adj\otimes tw)$ and $(tw\otimes adj)$ modules. 

\subsection{Interpretation of unfolded systems through $\sigma_-$ cohomology} \label{sec: Interpretation of cohomology}

In this section we recall an analysis of a general unfolded equation of the form
\begin{equation}\label{e: general unfolded}
    \mathcal{D} \mathcal{W} = 0
\end{equation}
via a cohomology of nilpotent operator included into the definition of $\mathcal{D}$.

We start with the introduction of an $\mathbb{N}$-graded vector space $V = \oplus_{q}V^q$, where $V^q$ is a subspace of grade $q$ (note that grading of $V$ is bounded from below). Suppose $\mathcal{W}$ is an element of $\Lambda^p(\mathcal{M}^d)\otimes V$ over some smooth $d$-dimensional manifold $\mathcal{M}^d$. Let $\sigma_{\pm}$ be operators that act \enquote{vertically}, i.e., do not affect the space-time coordinates,
and shift grading by $\pm 1$, $D_L$ be the Grassmann-odd operator that has zero grading and is allowed to act non-trivially on the space-time coordinates.
Consider the covariant constancy condition of a general form $\mathcal{W}$ along with
the zero-curvature condition
\begin{equation}\label{e:covariantconsteq}
\mathcal{D} \mathcal{W}  = (D_L + \sigma_- + \sigma_+) \mathcal{W} = 0, \quad \mathcal{D}^2 = 0.
\end{equation}
Notice that eq. \eqref{e:covariantconsteq} is invariant under the gauge transformations
\begin{equation}\label{e:gaugetransformcov}
\delta \mathcal{W} = \mathcal{D} \varepsilon ,
\end{equation}
where  $\varepsilon \in \Lambda^{p-1}(\mathcal{M}^d) \otimes V$. These gauge transformations contain both differential gauge transformations (e.g., linearized diffeomorphisms) and Stueckelberg gauge symmetries (e.g. linearized local Lorentz transformations).

Zero-curvature condition combined with the definite grading of each term in the definition of $\mathcal{D}$ results in
\begin{gather}
    (\sigma_\pm)^2 = 0\,, \quad \{D_L\,,  \sigma_\pm\} = 0\,, \quad    D_L^2 + \{\sigma_-\,,\sigma_+\}= 0\,.
\end{gather}
Therefore, one can see that $\sigma_-$ is a nilpotent differential on the $\Lambda^\bullet(\mathcal{M}^d)\otimes V$ complex.

We will use the following terminology customary in the HS literature \cite{Shaynkman:2000ts, Vasiliev:2001zy}. By dynamical field, we mean a field that is not expressed via derivatives of something else by field equations (e.g. the frame field in gravity or a frame-like HS one-form field $\omega_{\alpha(s)\,,\dot{\alpha}(s)}$ in the standard HS theory). The fields that are expressed by virtue of the field equations as derivatives of the dynamical fields modulo Stueckelberg gauge symmetries, i.e., that are not annihilated by $\sigma_-$, are referred to as auxiliary fields (e.g., the Lorentz connection in gravity or its HS analogues $\omega_{\alpha(s+k)\,,\dot{\alpha}(s-k)}\,, \, k\in \{-s,\dots,-1,1,\dots,s\}$, where $\alpha(n)$ denotes $n$ symmetrized indices). By Stueckelberg fields we mean $\sigma_-$-exact fields, i.e., fields of the form $\sigma_-\chi$, as they can be eliminated by an appropriate $\sigma_-$-exact term in the gauge transformation. A field that is neither auxiliary nor pure gauge by Stueckelberg gauge symmetries is said to be a nontrivial dynamical (primary) field.

The classification for the gauge parameters is analogous. The parameters, that are not
annihilated by $\sigma_-$, describe algebraic Stueckelberg shifts. The leftover symmetries are described by the parameters  in $\ker (\sigma_-)$. The $\sigma_-$-exact parameters correspond to the gauge for gauge transformations. Parameters, which are $\sigma_-$-closed and not $\sigma_-$-exact, are referred to as differential gauge parameters. Note that since in the CHS theory the gauge parameters are zero-forms there is no room for gauge for gauge symmetries.

Introducing generalized curvatures as
\begin{equation}\label{e:def of R}
    R = \mathcal{D} \mathcal{W} = 0\,,
\end{equation}
we can interpret a zero-curvature condition $\mathcal{D}^2=0$ as Bianchi identities
\begin{equation}\label{e:covariantconsteq for R}
    \mathcal{D} R = (D_L + \sigma_- + \sigma_+) R=0\,.
\end{equation}
As we can decompose equation (\ref{e:def of R}) with respect to the $\mathbb{N}$-grading and study the constraints imposed by the Bianchi identities starting with the lowest grading, one can see that components of $R$ belong to $\ker\sigma_-$. The $\sigma$-exact components of the field strength $R$ at the level $n$ can be set to zero by an algebraic shift of the field $\mathcal{W}$ at the level $n+1$, and the cohomological components of $R$ can be expressed via (\ref{e:def of R}) in terms of the derivatives of primary fields. As the r.h.s. of (\ref{e:covariantconsteq}) is set to zero ($R$ = 0), the differential equations are imposed on primary fields. Therefore, the representatives of $H^{p+1}(\sigma_-)$ correspond to independent equations for the nontrivial dynamical fields. The Bianchi identities themselves may also not be independent, so the higher cohomology groups may turn out to be non-empty.

To sum up all that has been said above, one can prove the following proposition \cite{Shaynkman:2000ts} (see also \cite{Gelfond:2003vh, Bekaert:2004qos, Vasiliev:2009ck, Bychkov:2021zvd}):
\begin{theorem}\label{Theorem_cohomology}
The following is true for a $p$-form $\mathcal{W}$:

1) Differential gauge symmetry parameters $\varepsilon$ span  $H^{p-1}(\sigma_-)$

2) Nontrivial dynamical fields $\mathcal{W}$ span $H^p(\sigma_-)$

3) Gauge-invariant differential operators on the nontrivial dynamical fields,
contained in $\mathcal{D} \mathcal{W}  = 0$, span  $H^{p+1}(\sigma_-)$
\end{theorem}

More generally, higher cohomology $H^{k}(\sigma_-)$ with $k>p+1$ describes Bianchi identities for dynamical equations at $k=p+2$ and Bianchi for Bianchi identities at $k>p+2$ \cite{Vasiliev:2009ck}. Similarly, the lower cohomology $H^{k}(\sigma_-)$ with $k<p-1$ describes gauge for gauge differential symmetries.

Note that if $H^{p+1}(\sigma_-) = 0$, the equation (\ref{e: general unfolded}) contains only constraints which express auxiliary fields via derivatives of the dynamical fields, imposing no restrictions on the latter. For non-empty $H^{p+1}(\sigma_-)$, the equation (\ref{e: general unfolded}) implies that the cohomology vanish, which imposes some differential equations on the primary fields. The addition of some cohomology representative on the r.h.s. of (\ref{e: general unfolded}) lifts an equation and relaxes the dynamics of the system. If $D_L$ is a first order differential operator (which is true in HS applications) then, if $H^{p+1}(\sigma_-)$ is nonzero in the grade $k$ sector, the associated differential equations on a grade $l$ dynamical field are of order $k + 1 - l$. The same rule determines the number of derivatives in a genuine gauge transformation.

In the context of a HS theory $\sigma_-$-cohomology analysis reduces to the analysis of the one-form $\omega$ sector and zero-form $C$ sector. In the first case $H^{2}(\sigma_-)$ consists of two types of 2-cocycles: equation cocycles and the Weyl cocycle. Gauge-invariant differential operators corresponding to the equation cocycles are genuine Fronsdal or partially massless differential field equations. To derive these equations explicitly one has to eliminate all auxiliary fields expressing them as derivatives of primary fields. The Weyl cocycle is a Chevalley-Eilenberg cohomology that glues a one-form and zero-form modules. The gluing is performed via a Weyl-like tensor component of $C$ (for spin 2 it is indeed the linearized Weyl-tensor). If the Chevalley-Eilenberg cocycle is set to zero, the corresponding HS gauge one-forms are pure gauge (it is equivalent to setting the generalized Weyl tensor to zero, preventing nontrivial dynamic). For a deeper discussion of Chevalley-Eilenberg cohomology in the HS context see \cite{Vasiliev:2007yc}.

As a vector space $V$ is usually a module of some Lie algebra $\mathfrak{g}$, it is important to note that a correct physical interpretation of an unfolded system (\ref{e: general unfolded}) can be achieved only if $\sigma_-$ is defined on $\mathfrak{g}$ indecomposable modules. If that is not the case, while the analysis above still holds true, a mixing of independent fields occurs that significantly complicates their identification and physical interpretation of the system. Therefore, in the rest of the paper we will have to decompose all modules under consideration into $\mathfrak{so}(3,2)$ indecomposable (irreducible if possible) modules with a help of dual algebras.

\section{One-form module  $(adj \otimes adj)$ }
\label{sec: one-form adj-adj}

We start the analysis of $B_2$ CHS theory spectrum with the case of one-forms $\omega$ along $I_1 I_2$ valued in a $(adj \otimes adj)$ as this sector contains copies of symmetric massless fields. Of the component fields contained in $\omega(Y_1,Y_2,I_1 I_2;\hat{K}|x)$ only 8 correspond to the $(adj\otimes adj)$-sector due to the multiplicity in the product of reflection matrices discussed in section \ref{B_2CHS}. We will proceed with the equation on $\omega(Y_1,Y_2|x)*I_1 I_2$ as the equations on the other 7 $(adj \otimes adj)$ one-forms have the same form:
\medmuskip=2mu
\begin{align} \label{e: adj-adj omega}
    &\mathcal{D}\omega(y_1, y_2, \bar{y}_1, \bar{y}_2,I_1 I_2 | x) \equiv\left[D_L + e^{\alpha \dot\alpha} \sum_{i=1}^{2}(\bar{y}_{\dot\alpha i} \dd_{\alpha i} + y_{\alpha i} \bar{\dd}_{\dot\alpha i})\right] \omega(y_1, y_2, \bar{y}_1, \bar{y}_2,I_1 I_2 | x) = \nonumber\\
    =&-\frac{i \eta_1}{2}\bar{H}_{\dot\alpha \dot\beta}\bar{\dd}^{\dot\alpha}_1\bar{\dd}^{\dot\beta}_1 C(0,y_2,\bar{y}_1,\bar{y}_2,I_1 I_2;\hat{k}_1 | x) * \hat{k}_1 -\frac{i \eta_1}{2}\bar{H}_{\dot\alpha \dot\beta}\bar{\dd}^{\dot\alpha}_2\bar{\dd}^{\dot\beta}_2 C(y_1,0,\bar{y}_1,\bar{y}_2,I_1 I_2; \hat{k}_2 |x)*\hat{k}_2 - \nonumber\\ 
    &-\frac{i \eta_2}{2}\bar{H}_{\dot\alpha \dot\beta} \bar{\dd}_-^{\dot\alpha} \bar{\dd}_-^{\dot\beta} C(y_+,0,\bar{y}_+,\bar{y}_-,I_1 I_2;\hat{k}_{12}|x)*\hat{k}_{12} -\nonumber\\
    &- \frac{i \eta_2}{2}\bar{H}_{\dot\alpha \dot\beta} \bar{\dd}_+^{\dot\alpha} \bar{\dd}_+^{\dot\beta} C(0,y_-,\bar{y}_+,\bar{y}_-,I_1 I_2;\hat{k}^+_{12}|x)*\hat{k}^+_{12} + \text{c.c.}\,,
\end{align}
\medmuskip=4mu
where zero-forms $C$ take values in one of the tensor products of the standard twisted-adjoint and adjoint modules, and variables $Y_{\pm} = \frac{1}{\sqrt{2}}(Y_1 \pm Y_2)$.

As discussed in previous sections, extraction of the dynamical content of (\ref{e: adj-adj omega}) amounts to the analysis of $H^\bullet(\sigma_-)$. Consider the ring of differential forms $\mathfrak{R} = \Lambda^\bullet(M)\otimes \mathbb{C}[[y_1,y_2,\bar y_1, \bar y_2]]$ valued in $Y_i$ series. We leave aside the question of the function class that preserves the star-product algebra and locality in this linear analysis, though it is obviously a highly important topic for moving to higher orders. The homogeneous elements of this ring are
\begin{equation}\label{e: hom one-form}
    \omega_{n,\bar{n}; \,m,\bar{m}}(Y_1,Y_2\,|x,dx) = \omega^{\alpha(n),\dot\alpha(\bar n)| \beta(m), \dot\beta(\bar m)}(x,dx)\, y_{1\alpha(n)}\,\bar y_{1\dot\alpha(\bar{n})}y_{2\beta(m)}\,\bar y_{2\dot\beta(\bar{m})}\,,
\end{equation}
where indices $(\alpha,\dot{\alpha})$ and $(\beta,\dot{\beta})$ are not related by any symmetry, i.e., coefficients of the series belong to $\operatorname{Irrep}_{\mathfrak{so}(3,1)}\otimes \operatorname{Irrep}_{\mathfrak{so}(3,1)}$. From the covariant derivative in (\ref{e: adj-adj omega}) it is obvious that $\mathfrak{so}(3,2)$ acts on $\mathfrak{R}$ as
\begin{align}
P_{\alpha \dot{\alpha}} &= \delta^{ij}(y_{\alpha i}\bar{\dd}_{\dot{\alpha}j} + \dd_{\alpha i}\bar{y}_{\dot{\alpha}j})\,,\label{e: so32 adj-adj 1}\\
L_{\alpha\alpha} &= \delta^{ij}y_{i\alpha}\dd_{j\alpha}\,, \quad\bar{L}_{\dot{\alpha}\dot{\alpha}} = \delta^{ij}\bar{y}_{i\dot{\alpha}}\bar{\dd}_{j\dot{\alpha}}\,. \label{e: so32 adj-adj 2}
\end{align}
This action preserves the total degree of any monomial, thus as a $\mathfrak{so}(3,2)$ module the ring of $Y_i$ series decomposes into the direct sum of differential forms valued in homogeneous polynomials $\operatorname{PolyHom}(Y_1,Y_2; n)$ of total degree $n$
\begin{equation}\label{RingPolyHom}
    \mathfrak{R} \simeq \oplus_{n = 0}^\infty\,  \Lambda^\bullet(M)\otimes \operatorname{PolyHom}(Y_1,Y_2; n)\,.
\end{equation}

The $(adj\otimes adj)$ $\mathfrak{so}(3,2)$-module is reducible and, therefore, each $\operatorname{PolyHom}(Y_1,Y_2; n)$ should be decomposed into the sum of finite dimensional $\mathfrak{so}(3,2)$-irreps. There is an isomorphism $\mathfrak{so}(3,2) \simeq \mathfrak{sp}(4)$ that can be observed directly in our case if  spinor two-component indices ($\alpha\,, \dot{\alpha}$) are unified into a single $\mathfrak{sp}(4)$ index $\mathcal{A}\in \{1\,,\dots\,,4\}$. All finite-dimensional $\mathfrak{sp}(4)$-irreps are exhausted by Young diagrams of at most two rows. Decomposition of $\operatorname{PolyHom}(Y_1,Y_2; n)$ into the direct sum of $\mathfrak{sp}(4)$ Young diagrams can be performed with the help of dual algebra.

One can check that the following operators bilinear in $(Y^\mathcal{A}_i\,, \dd_{i\mathcal{A}})$ commute with $\mathfrak{so}(3,2)\simeq \mathfrak{sp}(4)$ generated by (\ref{e: so32 adj-adj 1})--(\ref{e: so32 adj-adj 2})
\begin{align}
E_h &= y_1^\alpha \dd_{2\alpha} + \bar{y}_1^{\dot{\alpha}} \bar{\dd}_{2\dot{\alpha}} = Y_1^\mathcal{A} \dd_{2\mathcal{A}}\,, \quad F_h = y_2^\alpha \dd_{1\alpha} + \bar{y}_2^{\dot{\alpha}} \bar{\dd}_{1\dot{\alpha}} = Y_2^\mathcal{A} \dd_{1\mathcal{A}}\,, \label{e: sl2h raise-lower}\\
H_h &= N_1 - N_2 + \bar{N}_1 - \bar{N}_2 = Y_1^\mathcal{A} \dd_{1\mathcal{A}}- Y_2^\mathcal{A} \dd_{2\mathcal{A}}\,,\\
N_i &= y_i^\alpha \dd_{i\alpha}\,,\quad \bar{N}_i = \bar{y}_i^{\dot{\alpha}} \bar{\dd}_{i\dot{\alpha}}
\end{align}
and
\begin{align}
E_v &= y_{1\alpha} y_2^\alpha + \bar{y}_{1\dot{\alpha}} \bar{y}_2^{\dot{\alpha}} = Y_{1\mathcal{A}} Y_{2}^{\mathcal{A}}\,, \quad F_v = \dd_{2\alpha} \dd_1^\alpha + \bar{\dd}_{2\dot{\alpha}} \bar{\dd}_1^{\dot{\alpha}} = \dd_{2\mathcal{A}} \dd_1^\mathcal{A}\,, \\
H_v &= N_1 + N_2 + \bar{N_1} + \bar{N_2} + 4 = Y_i^\mathcal{A} \dd_{j\mathcal{A}} \delta^{ij} + 4\,.
\end{align}
Operators $\{E_h\,, F_h\,, H_h\}$ and $\{E_v\,, F_v\,, H_v\}$ commute with each other and form $\mathfrak{sl}^h(2)\oplus \mathfrak{sl}^v(2)$ with following non-trivial commutation relations:
\begin{equation}
\begin{aligned}[c]
[E_h, F_h] &= H_h\,, \\
[H_h, E_h] &= 2E_h\,, \\
[H_h, F_h] &= -2F_h\,,
\end{aligned}
\qquad 
\begin{aligned}[c]
[E_v, F_v] &= H_v\,, \\
[H_v, E_v] &= 2E_v\,, \\
[H_v, F_v] &= -2F_v\,.
\end{aligned}
\end{equation}
  
The correspondence between (highest-weight, lowest-weight) vectors (from here on out HW, LW vectors) of $\mathfrak{sl}^h(2)\oplus \mathfrak{sl}^v(2)$ and $\mathfrak{sp}(4)$ two-row Young diagram is quite transparent. Consider homogeneous polynomials
\begin{equation}
    \psi_{n+m}(Y_1\,,Y_2) = \psi_{\mathcal{A}(n)|\mathcal{B}(m)} Y^{\mathcal{A}(n)}_1 Y^{\mathcal{B}(m)}_2\,.
\end{equation}
Then
\begin{align}
    E_h \psi_{n+m}(Y_1,Y_2) &= Y_1^\mathcal{A} \dd_{2\mathcal{A}}\psi_{n+m}(Y_1,Y_2) = 0 \text{: Young condition on $\psi_{\mathcal{A}(n)|\mathcal{B}(m)}$\,,}\\
    F_v \psi_{n+m}(Y_1,Y_2) &= \dd_{2\mathcal{A}} \dd_1^\mathcal{A} \psi_{n+m}(Y_1,Y_2) = 0  \text{: $\mathfrak{sp}(4)$ tracelessness of $\psi_{\mathcal{A}(n)|\mathcal{B}(m)}$\,,}\\
    H_h \psi_{n+m}(Y_1,Y_2) &= (n - m)\psi_{n+m}(Y_1,Y_2) = \lambda_h \psi_{n+m}(Y_1,Y_2)  \text{: row length}\,,\\
    H_v \psi_{n+m}(Y_1,Y_2) &= (n + m + 4)\psi_{n+m}(Y_1,Y_2) = \lambda_v \psi_{n+m}(Y_1,Y_2)   \text{: number of boxes}+4\,,
\end{align}
\begin{equation}
\psi_{n+m}(Y_1,Y_2) \Longleftrightarrow
\begin{picture}(10,18)(0,7)
{
\put(05,00){\line(1,0){35}}%
\put(05,10){\line(1,0){35}}%
\put(05,0){\line(0,1){10}}%
\put(40,0){\line(0,1){10}}%
\put(17,3){\scriptsize  ${m}$}%
}
\end{picture}\begin{picture}(60,18)(10,7)
{\put(35,13){\scriptsize  ${n}$}
 \put(05,20){\line(1,0){60}}%
\put(05,10){\line(1,0){60}}%
\put(05,10){\line(0,1){10}}%
\put(65,10){\line(0,1){10}}%
}
\end{picture}
\bigg\vert_{\mathfrak{sp}(4)} \Longleftrightarrow 
\begin{picture}(10,18)(0,7)
{
\put(05,00){\line(1,0){65}}%
\put(05,10){\line(1,0){65}}%
\put(05,0){\line(0,1){10}}%
\put(70,0){\line(0,1){10}}%
\put(17,3){\scriptsize  ${(n-m)/2}$}%
}
\end{picture}\begin{picture}(75,18)(10,7)
{\put(35,13){\scriptsize  ${(n+m)/2}$}
 \put(05,20){\line(1,0){90}}%
\put(05,10){\line(1,0){90}}%
\put(05,10){\line(0,1){10}}%
\put(95,10){\line(0,1){10}}%
}
\end{picture}\quad \bigg\vert_{\mathfrak{so}(3,2)}\,.
\end{equation}

Therefore, there is a one-to-one correspondence between the set of (HW, LW) vectors of $\mathfrak{sl}^h(2)\oplus \mathfrak{sl}^v(2)$ in $\mathfrak{R}$ and $\mathfrak{sp}(4)$-irreps. Any element of $\mathfrak{R}$ can be represented as a linear combination of $\mathfrak{sl}^h(2)\oplus \mathfrak{sl}^v(2)$ (HW, LW) vectors and their descendants as per the Howe duality \cite{Howe:1989}, i.e., $\operatorname{Irrep}^{\lambda_h, \lambda_v}_{\mathfrak{sl}(2)\oplus \mathfrak{sl}(2)}$ constitute a multiplicity space of $\operatorname{Irrep}^{\lambda_h, \lambda_v}_{\mathfrak{sp}(4)}$. In other words, the ring of differential forms $\mathfrak{R}$ decomposes into the sum of tensor products of $\mathfrak{sl}^h(2)\oplus \mathfrak{sl}^v(2)$ and $\mathfrak{sp}(4)$ irreps defined by two weights $(\lambda_h,\lambda_v)$
\begin{equation}\label{e: Howe duality adj-adj}
    \mathfrak{R} \simeq \bigoplus_{\substack{\lambda_h,\lambda_v}} \operatorname{Irrep}^{\lambda_h, \lambda_v}_{\mathfrak{sl}(2)\oplus \mathfrak{sl}(2)} \otimes \operatorname{Irrep}^{\lambda_h, \lambda_v}_{\mathfrak{sp}(4)}\,.
\end{equation}

One can solve the $\mathfrak{sl}^h(2)\oplus \mathfrak{sl}^v(2)$ (HW, LW) conditions explicitly. The Young condition is satisfied by polynomials of the form
\begin{equation}\label{e: e1 HW}
    \psi(Y_1,Y_2) = \sum_{\lambda_h,r,p,q}\sum_{n+m=\lambda_h}\psi^{\lambda_h}{}_{\alpha(n+r),\dot{\alpha}(m+r)}y_1^{\alpha(n)}\bar{y}_1^{\dot{\alpha}(m)}X^{\alpha\dot{\alpha}(r)}\zeta^p \bar{\zeta}{}^q,
\end{equation}
where
\begin{gather}
    X^{\alpha\dot{\alpha}} = (y_1 \bar{y}_2 - y_2 \bar{y}_1)^{\alpha\dot{\alpha}}\,, \label{e: X bispinor}\\
    \zeta = y_{1\alpha} y_2^\alpha\,, \quad \bar{\zeta}= \bar{y}_{1\dot{\alpha}} \bar{y}_2^{\dot{\alpha}}\,. \label{e: traces}
\end{gather}

The choice of (HW, LW) makes the solutions biased towards the $Y_1$ variable as opposed to (LW, LW), which would be skewed towards $Y_2$. The whole $\mathfrak{sl}^h(2)\oplus \mathfrak{sl}^v(2)$ module describes all possible polynomials of $Y_{1,2}$, as should be obvious from the form of raising/lowering operators (\ref{e: sl2h raise-lower}).

Let us note here that 
\medmuskip=3mu
\begin{align}
    F_h \bigg(y_1^{\alpha(n)}\bar{y}_1^{\dot{\alpha}(m)}X^{\alpha\dot{\alpha}(r)}\zeta^p \bar{\zeta}{}^q\bigg) &= n y_1^{\alpha(n-1)}y_2^\alpha\bar{y}_1^{\dot{\alpha}(m)}X^{\alpha\dot{\alpha}(r)}\zeta^p \bar{\zeta}{}^q + m y_1^{\alpha(n)}\bar{y}_1^{\dot{\alpha}(m-1)}\bar{y}_2^{\dot{\alpha}}X^{\alpha\dot{\alpha}(r)}\zeta^p \bar{\zeta}{}^q \,,\\
    H_h \bigg(y_1^{\alpha(n)}\bar{y}_1^{\dot{\alpha}(m)}X^{\alpha\dot{\alpha}(r)}\zeta^p \bar{\zeta}{}^q\bigg) &= (n+m) y_1^{\alpha(n)}\bar{y}_1^{\dot{\alpha}(m)}X^{\alpha\dot{\alpha}(r)}\zeta^p \bar{\zeta}{}^q\,,\\
    H_v\bigg(y_1^{\alpha(n)}\bar{y}_1^{\dot{\alpha}(m)}X^{\alpha\dot{\alpha}(r)}\zeta^p \bar{\zeta}{}^q \bigg) &=(n+m+2r+2p+2q+4)y_1^{\alpha(n)}\bar{y}_1^{\dot{\alpha}(m)}X^{\alpha\dot{\alpha}(r)}\zeta^p \bar{\zeta}{}^q\,,
\end{align}
\medmuskip=4mu
i.e., $\mathfrak{sl}^h(2)$ do not act on variables $\{X, \zeta, \bar{\zeta}\}$.
Since $\lambda_h = n+m$ and the constant term $4$ are explicitly present in $\lambda_v$ let us redefine $\lambda_v$ such that $\lambda_v = 2(p+q+r)$ for further convenience. It does not affect the subsequent analysis as redefined $(\lambda_h,\lambda_v)$ encode the same information as genuine $\mathfrak{sl}^h(2)\oplus \mathfrak{sl}^v(2)$ weights.

The tracelessness condition is then solved within the (\ref{e: e1 HW}) class:
\begin{equation}\label{e: e1 f2 HW LW}
    \Psi^{\vec{\lambda}}{}_{n,m,r}^{l}(Y_1,Y_2) = \left( \sum_{p+q=l}\Gamma^{p,q}_{n,m,r}\zeta^p \bar{\zeta}{}^q\right)\psi^{\vec{\lambda}}{}_{\alpha(n+r),\dot{\alpha}(m+r)}y_1^{\alpha(n)}\bar{y}_1^{\dot{\alpha}(m)}X^{\alpha\dot{\alpha}(r)}\,,
\end{equation}
where
\begin{align}
    \Gamma^{p,q}_{n,m,r} &= (-1)^q \frac{l!}{p! q! (n+r+p+1)!(m+r+q+1)!}\,,\\
    \vec{\lambda} &= (\lambda_h, \lambda_v): \quad \lambda_h = n+m\,, \quad \lambda_v = 2(l+r)\,.
\end{align}

Polynomials (\ref{e: e1 f2 HW LW}) exhaust the full set of $\mathfrak{sp}(4)$-irreps present in the decomposition of $\mathfrak{R}$. Now we can restrict one-forms $\omega$ to this class of polynomials and consider an action of $\mathcal{D}$ on each $\mathfrak{sp}(4)$-irrep to determine the physical content of the $B_2$ system in the $(adj \otimes adj)$ sector following section \ref{sec: Interpretation of cohomology}.

A naive approach of defining a grading and thus the $\sigma_-$ operator using Euler operators $y_i^\mathcal{A} \dd_{i\mathcal{A}}$ proves to be fruitless, as $\sigma_-$ defined in such a way does not commute with $\mathfrak{sl}^h(2)\oplus \mathfrak{sl}^v(2)$ operators, and thus does not act within a single irreducible module. This can be understood by examining the diagonal operators of $\mathfrak{sl}^h(2)\oplus \mathfrak{sl}^v(2)$, which depend on either the sum or the difference in the numbers of $Y_i$ oscillators. Still, we can faithfully determine the action of $\sigma_-$ on the two-row Young diagrams as they are completely determined by a set of spin-tensors. Indeed, given $\mathfrak{sp}(4)$-irrep can be decomposed into the direct sum of $\mathfrak{so}(3,1)$-irreps as $\mathfrak{so}(3,1)\simeq \mathfrak{sl}(2, \mathbb{C})$ is a subgroup of $\mathfrak{so}(3,2)\simeq \mathfrak{sp}(4)$
\begin{align}\label{e: restriction to subgroup}
    \operatorname{Irrep}^{\vec{\lambda}}_{\mathfrak{so}(3,2)} &= \bigoplus \operatorname{Irrep}^{\vec{\lambda}}_{\mathfrak{so}(3,2)\downarrow \mathfrak{so}(3,1)}\,,
\\
\begin{picture}(10,18)(0,7)
{
\put(05,00){\line(1,0){35}}%
\put(05,10){\line(1,0){35}}%
\put(05,0){\line(0,1){10}}%
\put(40,0){\line(0,1){10}}%
\put(17,3){\scriptsize  ${M}$}%
}
\end{picture}\begin{picture}(50,18)(10,7)
{\put(35,13){\scriptsize  ${N}$}
 \put(05,20){\line(1,0){60}}%
\put(05,10){\line(1,0){60}}%
\put(05,10){\line(0,1){10}}%
\put(65,10){\line(0,1){10}}%
}
\end{picture}\quad \bigg\vert_{\mathfrak{so}(3,2)} &= \bigoplus_{\substack{i = 0}}^{N-M}\bigoplus_{\substack{j = 0}}^{M}
\begin{picture}(10,18)(0,7)
{
\put(05,00){\line(1,0){35}}%
\put(05,10){\line(1,0){35}}%
\put(05,0){\line(0,1){10}}%
\put(40,0){\line(0,1){10}}%
\put(20,3){\scriptsize  ${j}$}%
}
\end{picture}\begin{picture}(50,18)(10,7)
{\put(25,13){\scriptsize  ${M + i}$}
 \put(05,20){\line(1,0){60}}%
\put(05,10){\line(1,0){60}}%
\put(05,10){\line(0,1){10}}%
\put(65,10){\line(0,1){10}}%
}
\end{picture}\quad \bigg\vert_{\mathfrak{so}(3,1)}
\,,
\end{align}
where $N = (\lambda_h + \lambda_v)/2$ and $M = \lambda_h/2$.

To determine how covariant derivative $\mathcal{D}$ acts on spin-tensors $\omega^{\vec{\lambda}}{}_{\alpha(n+r),\dot{\alpha}(m+r)}$, let us introduce two auxiliary spinors $u^\alpha, \bar u^{\dot \alpha}$ with corresponding derivatives $\dd_\alpha\,, \bar{\dd}_{\dot{\alpha}}$ and use a generating function
\begin{equation}\label{e: rank-one adjadj}
\omega^{\vec{\lambda}}{}_{n,m,r}^{l} = \frac{1}{(n+r)!(m+r)!}\left( \sum_{p+q=l}\Gamma^{p,q}_{n,m,r}\zeta^p \bar{\zeta}{}^q\right)\left(\dd_\alpha \bar{\dd}_{\dot{\alpha}}X^{\alpha\dot{\alpha}}\right)^r \left(\dd_\alpha y_1^\alpha\right)^n \left( \bar{\dd}_{\dot{\alpha}}\bar{y}_1^{\dot{\alpha}}\right)^m \omega^{\vec{\lambda}}{}^l_{n,m,r}(u,\bar{u})
\end{equation}
with a spin-tensor $\omega^{\vec{\lambda}}{}_{\alpha(n+r),\dot{\alpha}(m+r)}$ encoded in a polynomial of $(u,\bar{u})$
\begin{equation}\label{e: u spin tensor}
    \omega^{\vec{\lambda}}{}^l_{n,m,r}(u,\bar{u}) = \omega^{\vec{\lambda}}{}_{\alpha(n+r),\dot{\alpha}(m+r)}u^{\alpha(n+r)}\bar{u}^{\dot{\alpha}(m+r)}\,.
\end{equation}
Since $\omega^{\vec{\lambda}}{}^l_{n,m,r}(u,\bar{u})$ is a polynomial of degrees $(n+r)$ with respect to $u$ and $(m+r)$ with respect to $\bar{u}$, there is no need to set $u=\bar{u} = 0$ in (\ref{e: rank-one adjadj}).
This trick allows us to transform an equation on rank-two field (\ref{e: adj-adj omega}) into the equation on rank-one field with a much simpler definition of the $\sigma_-$ operator.

While the Lorentz derivative $D_L$ (\ref{e: DL}) \enquote{commutes} with the differential operator prefactor in (\ref{e: rank-one adjadj}) and acts on $\omega^{\vec{\lambda}}{}^l_{n,m,r}(u,\bar{u})$ in a standard way
\begin{align} 
    D_L \omega^{\vec{\lambda}}{}^l_{n,m,r}(u,\bar{u}) &= \left(\d_x  + \omega(u\,,\dd) + \bar \omega(\bar{u}\,,\bar{\dd})\right)\omega^{\vec{\lambda}}{}^l_{n,m,r}(u,\bar{u})\,,
\\
    \omega (u\,,\dd) &= \omega^{\alpha \beta} u_{\alpha }\dd_{\beta}\,, \quad \bar{\omega}(\bar{u}\,,\bar{\dd}) = \bar \omega^{\dot\alpha \dot\beta} \bar u_{\dot\alpha }\bar\dd_{\dot\beta }\,,
\end{align}
that is not so for the momentum operator $P_{\alpha\dot{\alpha}}$
\thickmuskip=1mu 
\medmuskip=1mu
\thinmuskip=1mu
\begin{align}
    &e^{\alpha\dot{\alpha}}P_{\alpha\dot{\alpha}}\omega^{\vec{\lambda}}{}_{n,m,r}^{l}(Y_1, Y_2) =\nonumber
    \\
    &=\left[\left( \sum_{p+q=l}\Gamma^{p,q}_{n-1,m+1,r}\zeta^p \bar{\zeta}{}^q\right)\left(\dd_\alpha \bar{\dd}_{\dot{\alpha}}X^{\alpha\dot{\alpha}}\right)^r \left(\dd_\alpha y_1^\alpha\right)^{n-1} \left( \bar{\dd}_{\dot{\alpha}}\bar{y}_1^{\dot{\alpha}}\right)^{m+1}n(m+r+l+2)e^{\alpha\dot{\alpha}}\bar{u}_{\dot{\alpha}}\dd_\alpha + \nonumber\right.
    \\
    &\left.+\left( \sum_{p+q=l}\Gamma^{p,q}_{n+1,m-1,r}\zeta^p \bar{\zeta}{}^q\right)\left(\dd_\alpha \bar{\dd}_{\dot{\alpha}}X^{\alpha\dot{\alpha}}\right)^r \left(\dd_\alpha y_1^\alpha\right)^{n+1} \left( \bar{\dd}_{\dot{\alpha}}\bar{y}_1^{\dot{\alpha}}\right)^{m-1}m(n+r+l+2)e^{\alpha\dot{\alpha}}\bar{\dd}_{\dot{\alpha}}u_\alpha -\nonumber\right.
    \\
    &\left.-\left( \sum_{p+q=l+1}\Gamma^{p,q}_{n,m,r-1}\zeta^p \bar{\zeta}{}^q\right)\left(\dd_\alpha \bar{\dd}_{\dot{\alpha}}X^{\alpha\dot{\alpha}}\right)^{r-1} \left(\dd_\alpha y_1^\alpha\right)^{n} \left( \bar{\dd}_{\dot{\alpha}}\bar{y}_1^{\dot{\alpha}}\right)^{m}r(n+m+r+1)e^{\alpha\dot{\alpha}}\bar{\dd}_{\dot{\alpha}}\dd_\alpha + \nonumber\right.
\end{align}
\begin{align}\label{e: P adjadj}
    &\left.+\left( \sum_{p+q=l-1}\Gamma^{p,q}_{n,m,r+1}\zeta^p \bar{\zeta}{}^q\right)\left(\dd_\alpha \bar{\dd}_{\dot{\alpha}}X^{\alpha\dot{\alpha}}\right)^{r+1} \left(\dd_\alpha y_1^\alpha\right)^{n} \left( \bar{\dd}_{\dot{\alpha}}\bar{y}_1^{\dot{\alpha}}\right)^{m}l(n+m+2r+l+3)e^{\alpha\dot{\alpha}}\bar{u}_{\dot{\alpha}}u_\alpha\right]\times\nonumber
    \\
    &\times \frac{1}{(n+r+1)(m+r+1)} \omega^{\vec{\lambda}}{}^l_{n,m,r}(u,\bar{u})\,,
\end{align}
\thickmuskip=5mu 
\medmuskip=4mu
\thinmuskip=3mu
where the factorial prefactor in (\ref{e: rank-one adjadj}) has been absorbed into the $\omega^{\vec{\lambda}}{}_{\alpha(n+r),\dot{\alpha}(m+r)}$ for the future convenience of the analysis of the gluing between one-forms $\omega$ and zero-forms $C$ in the vertices, but is irrelevant to the $H^\bullet(\sigma_-)$ cohomology computation.

One can see that the momentum operator $P$ preserves the function class (\ref{e: e1 f2 HW LW}) and properly acts on the rank-one fields $\omega^{\vec{\lambda}}{}^l_{n,m,r}(u,\bar{u})$ (\ref{e: u spin tensor}). Therefore, $\mathcal{D}\omega(Y_1, Y_2)$ in (\ref{e: adj-adj omega}) is expressed via the rank-one fields in the form of
\begin{align}
&\sum_{\lambda_h\,,\lambda_v}\sum_{\substack{n+m=\lambda_h\\2(r+l)=\lambda_v}}\left(\sum_{p+q=l} \Gamma_{n, m,r}^{p, q} \zeta^p \bar{\zeta}{}^q\right)(\dd \bar{\dd} X)^r\left(\dd y_1\right)^n\left(\bar{\dd} \bar{y}_1\right)^m\Bigg\{D_L \omega^{\vec{\lambda}}{}^l_{n, m, r}(u, \bar{u})+\nonumber\\
&+\frac{(n+1)(m+r+l+1)}{(n+r+2)(m+r)} e^{\alpha \dot{\alpha}} \bar{u}_{\dot{\alpha}} \dd_\alpha \omega^{\vec{\lambda}}{}^l_{n+1, m-1, r}(u, \bar{u})+\nonumber\\
&+\frac{(m+1)(n+r+l+1)}{(n+r)(m+r+2)} e^{\alpha \dot{\alpha}} u_\alpha \bar{\dd}_{\dot{\alpha}} \omega^{\vec{\lambda}}{}^l_{n-1, m+1, r}(u, \bar{u})- \nonumber\\
&-\frac{(r+1)(n+m+r+2)}{(n+r+2)(m+r+2)} e^{\alpha \dot{\alpha}} \dd_\alpha \bar{\dd}_{\dot{\alpha}} \omega^{\vec{\lambda}}{}^{l-1}_{n, m, r+1}(u, \bar{u})+\nonumber\\
&+\frac{(l+1)(n+m+2 r+l+2)}{(n+r)(m+r)} e^{\alpha \dot{\alpha}} u_\alpha \bar{u}_{\dot{\alpha}} \omega^{\vec{\lambda}}{}^{l+1}_{n, m, r-1}(u, \bar{u})\Bigg\}\,,\label{e: lhs adj-adj pre-strip}
\end{align}
where useful notations are introduced
\begin{equation}\label{e: u notation}
    \dd\, \bar{\dd} \,X := \dd_\alpha \bar{\dd}_{\dot\alpha} X^{\alpha \dot{\alpha}}\,, \quad \dd\, y_1 := \dd_\alpha y^\alpha_1\,, \quad \bar{\dd}\, \bar{y}_1 := \bar{\dd}_{\dot\alpha} \bar{y}^{\dot\alpha}_1\,.
\end{equation}

Expression (\ref{e: lhs adj-adj pre-strip}) can be made even more concise by introducing a generating function with additional auxiliary variables $\chi,t$ used to encode the degree of $X^{\alpha \dot \alpha}$ and the total degree of traces $\zeta^p\bar{\zeta}{}^q$
\begin{equation}
    \omega(u,\bar{u},\chi,t) = \sum_{\lambda_h, \lambda_v}\sum_{\substack{n+m=\lambda_h\\ 2(r+l)=\lambda_v}}\omega^{\vec{\lambda}}{}^l_{n,m,r}(u\,, \bar{u}) \chi^r t^l\,.
\end{equation}

The action of the covariant derivative on this generating function is
\begin{align}
&D_L \omega(u, \bar{u}, \chi, t)+\left[\frac{\left(\hat{N}_u-\hat{N}_\chi+1\right)\left(\hat{\bar{N}}_u+\hat{N}_t+1\right)}{\left(\hat{N}_u+2\right) \hat{\bar{N}}_u} e^{\alpha \dot{\alpha}} \bar{u}_{\dot{\alpha}} \dd_\alpha+\nonumber\right.\\
&\left.+\frac{\left(\hat{\bar{N}}_u-\hat{N}_\chi+1\right)\left(\hat{N}_u+\hat{N}_t+1\right)}{\left(\hat{\bar{N}}_u+2\right) \hat{N}_u} e^{\alpha \dot{\alpha}} u_\alpha \bar{\dd}_{\dot{\alpha}} -\frac{\left(\hat{N}_u+\hat{\bar{N}}_u-\hat{N}_\chi+2\right)}{\left(\hat{N}_u+2\right)\left(\hat{\bar{N}}_u+2\right)} e^{\alpha \dot{\alpha}} \dd_\alpha\bar{\dd}_{\dot{\alpha}}\left(t \dd_\chi\right)+\nonumber\right.\\
&\left.+\frac{\left(\hat{N}_u+\hat{\bar{N}}_u+\hat{N}_t+2\right)}{\hat{N}_u \hat{\bar{N}}_u} e^{\alpha \dot{\alpha}} u_\alpha \bar{u}_{\dot{\alpha}}\left(\chi \dd_t\right)\right]\omega(u, \bar{u}, \chi, t)\,,
\end{align}
where
\begin{alignat}{2}
& \hat{N}_u=u^\alpha \frac{\dd}{\dd u^\alpha}\,, &\quad \hat{\bar{N}}_u&=\bar{u}^{\dot{\alpha}} \frac{\dd}{\dd \bar{u}^{\dot{\alpha}}}\,, \\
& \hat{N}_\chi=\chi \dd_\chi\,, &\quad \hat{N}_t&=t \dd_t\,.
\end{alignat}

Defining grading as $G=|\hat{N}_u-\hat{\bar{N}}_u|+2\hat{N}_\chi$ we obtain correct nilpotent $\sigma_-$ operator, that acts on $\mathfrak{sp}(4)$-irreps:
\begin{align}
&\sigma_-=\frac{\left(\hat{N}_u-\hat{N}_\chi+1\right)\left(\hat{\bar{N}}_u+\hat{N}_t+1\right)}{\left(\hat{N}_u+2\right) \hat{\bar{N}}_u} e^{\alpha \dot{\alpha}} \bar{u}_{\dot{\alpha}}\dd_\alpha\, \theta\left(\hat{N}_u-\hat{\bar{N}}_u\right)+\nonumber\\
&+\frac{\left(\hat{\bar{N}}_u-\hat{N}_\chi+1\right)\left(\hat{N}_u+\hat{N}_t+1\right)}{\left(\hat{\bar{N}}_u+2\right) \hat{N}_u} e^{\alpha \dot{\alpha}} u_\alpha \bar{\dd}_{\dot{\alpha}}\, \theta\left(\hat{\bar{N}}_u-\hat{N}_u\right)-\nonumber\\
&-\frac{\left(\hat{N}_u+\hat{\bar{N}}_u-\hat{N}_\chi+2\right)}{\left(\hat{N}_u+2\right)\left(\hat{\bar{N}}_u+2\right)} e^{\alpha \dot{\alpha}} \dd_\alpha \bar{\dd}_{\dot{\alpha}}\left(t \dd_\chi\right)\,,\label{e: sigma adj adj}
\end{align}

where Heaviside step function is defined as $\theta(0)=0$.

Note that first two terms reduce the difference in the number of dotted and undotted indices and the last term convert a pair of $(\alpha,\dot{\alpha})$ and $\chi$ into the \enquote{trace} $t$, i.e., the underlying spin-tensor shortens by one dotted and undotted indices.

After transferring the action of the covariant derivative $\mathcal{D}$ to the spin-tensors, this definition of $\sigma_-$ is an expected one, as it acts on the r.h.s. of (\ref{e: restriction to subgroup}) by cutting off boxes from $\mathfrak{so}(3,1)$ Young diagrams.

\subsection{Computation of $H^\bullet(\sigma_-)$}
\label{sec: cohomology computation adj-adj}

In this section we derive lower cohomology groups $H^{0,1,2}(\sigma_-)$ in the $(adj \otimes adj)$ sector, with $\sigma_-$ (\ref{e: sigma adj adj}). Instead of a generalized Hodge like approach to the cohomology calculation often used in a HS context \cite{Vasiliev:2009ck, Bychkov:2021zvd, Gelfond:2013lba}, we perform a straightforward extraction of $\ker \sigma_-$ and discard $\sigma_-$-exact forms. This is possible due to special properties of two-component spinors like the Schouten identity.

We stick to the following procedure:
\begin{itemize}
    \item Decompose a general $p$-form into Lorentz-irreducible $p$-forms;
    \item Act by $\sigma_-$ on a general $p$-form, decompose the result into Lorentz-irreducible $(p+1)$-forms and equate to zero;
    \item Extract a system of linear equations on the $p$-form coefficients. Solution provides a $\ker \sigma_-$;
    \item Go through each $\sigma_-$-closed form and either prove its non-exactness or represent it as a $\sigma_-$-exact form. 
\end{itemize}

Introduce a notation that will significantly simplify future formulas
\begin{equation}
    f_{a,b}(u,\bar{u}) := f_{\alpha(a),\dot{\alpha}(b)}u^{\alpha(a)}\bar{u}^{\dot{\alpha}(b)}\,.
\end{equation}

\subsubsection{$H^0(\sigma_-)$}
A general zero-form can be written down via auxiliary spinor variables $u^\alpha, \bar{u}^{\dot{\alpha}}$ and variables $t,\chi$ that preserve information about traces and bispinors $X^{\alpha \dot\alpha}$
\begin{equation}
    \mathcal{E}^{\vec{\lambda}} (u,\bar u, \chi,t) = \sum_{2(r+l) = \lambda_v} \;\sum_{n+m=\lambda_h} \mathcal{E}^{\vec{\lambda}}_{n+r, m+r} (u,\bar{u}) \chi^r t^l \,.
\end{equation}

Direct computation yields
\begin{align}
&\sigma_-\mathcal{E}^{\vec{\lambda}} (u,\bar u, \chi,t) = \sum_{\substack{2(r+l) = \lambda_v\\n+m=\lambda_h}}\Bigg\{\theta(n-m) \frac{n(m+r+l+2)}{(n+r+1)(m+r+1)}e(\dd,\bar{u}) + \nonumber\\
&+\theta(m-n) \frac{m(n+r+l+2)}{(n+r+1)(m+r+1)}e(u,\bar{\dd}) - \nonumber\\
&-\frac{(n+m+r+1)}{(n+r+1)(m+r+1)}e(\dd,\bar{\dd})\left(t \dd_\chi\right)\Bigg\}\mathcal{E}^{\vec{\lambda}}_{n+r, m+r} (u,\bar{u}) \chi^r t^l\,,\label{e: sigma zero-form}
\end{align}

where Lorenz-irreducible one-forms have been introduced
\begin{equation} \label{e: irreducible one-forms}
    e(u,\bar\dd) = e^{\beta \dot\beta} u_\beta \bar{\dd}_{\dot\beta}, \quad e(\dd,\bar u) = e^{\beta \dot\beta} \dd_\beta \bar{u}_{\dot\beta}, \quad e(u,\bar u) = e^{\beta \dot\beta} u_\beta \bar{u}_{\dot\beta}, \quad e(\dd,\bar\dd) = e^{\beta \dot\beta} \dd_\beta \bar{\dd}_{\dot\beta}\,.
\end{equation}
Since irreducible one-forms are linearly independent, the kernel can be easily found to be on the diagonal $n=m$, where both step functions vanish, supplemented by the condition $r = 0$
\begin{equation}
    \ker(\sigma_-) \vert_{\deg \Lambda = 0} = \left\{\mathcal{E}^{\vec{\lambda}}_{\alpha\left(\frac{\lambda_h}{2}\right), \dot{\alpha}\left(\frac{\lambda_h}{2}\right)} u^{\alpha\left(\frac{\lambda_h}{2}\right)} \bar{u}^{\dot{\alpha}\left(\frac{\lambda_h}{2}\right)} t^{\frac{\lambda_v}{2}}\right\}\,.
\end{equation}

As $\text{Im}(\sigma_-) \vert_{\deg \Lambda = 0} = \{\emptyset\}$ this yields
\begin{equation}
H^0(\sigma_-)=\left\{\mathcal{E}^{\vec{\lambda}}_{\frac{\lambda_h}{2}\,, \frac{\lambda_h}{2}}(u\,, \bar{u}) t^{\frac{\lambda_v}{2}} ; \quad \lambda_h, \lambda_v \in 2 \mathbb{Z}_{\geqslant 0}\right\}\,.
\end{equation}

\subsubsection{$H^1(\sigma_-)$}

Decomposition of a general one-form with respect to Lorenz-irreducible one-forms is
\begin{align}
&\omega^{\vec{\lambda}}(u,\bar{u},\chi,t)=\sum_{\substack{n+m=\lambda_h \\ 2(r+l)=\lambda_v}}\bigg\{\frac{1}{(n+r+1)(m+r+1)} e(\dd, \bar{\dd}) \omega_{A}^{\vec{\lambda}}{}_{n+r+1, m+r+1}(u, \bar{u}) \chi^r t^l + 
\nonumber\\
& + e(u, \bar{u}) \omega_{B}^{\vec{\lambda}}{}_{n+r-1,m+r-1}(u, \bar{u}) \chi^r t^l+\frac{1}{n+r+1} e(\dd, \bar{u}) \omega^{\vec{\lambda}}_{C}{}_{n+r+1, m+r-1}(u, \bar{u}) \chi^r t^l + \nonumber\\
& + \frac{1}{m+r+1} e(u, \bar{\dd}) \omega^{\vec{\lambda}}_{D}{}_{n+r-1, m+r+1}(u, \bar{u}) \chi^r t^l\bigg\}\,.
\end{align}

Here the Latin labels $(A,B,C,D)$ refer to different Lorenz-irreducible one-form components. Consequently, they have different numbers of dotted and undotted indices and act as independent terms for the purpose of finding a null solution.

Introduce basis two-forms
\begin{equation}
e^{\nu \dot{\nu}} \wedge e^{\lambda \dot{\lambda}}=\frac{1}{2} H^{\nu \lambda} \bar{\varepsilon}^{\dot{\nu} \dot{\lambda}}+\frac{1}{2} \bar{H}^{\dot{\nu} \dot{\lambda}} \varepsilon^{\nu \lambda}\,,
\end{equation}
where 
\begin{equation}
H^{\nu \lambda}=H^{(\nu \lambda)}:=e^\nu{}_{\dot{\gamma}} \wedge e^{\lambda \dot{\gamma}}, \quad \bar{H}^{\dot{\nu} \dot{\lambda}}=\bar{H}^{(\dot{\nu} \dot{\lambda})}:=e_\gamma{}^{\dot{\nu}} \wedge e^{\gamma \dot{\lambda}}\,.
\end{equation}

With the help of basis two-forms $\sigma_- \omega^{\vec{\lambda}}$ turns into
\begin{align}
&\sigma_-\omega^{\vec{\lambda}} (u,\bar u, \chi,t) = \sum_{\substack{2(r+l) = \lambda_v\\ n+m = \lambda_h}} \Bigg\{ \nonumber\\
& -\theta(n-m) \frac{n(m+r+l+2)}{2(n+r+1)^2(m+r+1)} H(\dd, \dd) \,
   \omega_{A}^{\vec{\lambda}}{}_{n+r+1, m+r+1}(u, \bar{u}) \chi^r t^l -\nonumber\\
& -\theta(m-n) \frac{m(n+r+l+2)}{2(n+r+1)(m+r+1)^2} \bar{H}(\bar{\dd}, \bar{\dd}) \,
   \omega_{A}^{\vec{\lambda}}{}_{n+r+1, m+r+1}(u, \bar{u}) \chi^r t^l +\nonumber\\
& +\theta(n-m) \frac{n(m+r+l+2)}{2(m+r+1)} \bar{H}(\bar{u}, \bar{u}) \,
   \omega_{B}^{\vec{\lambda}}{}_{n+r-1, m+r-1}(u, \bar{u}) \chi^r t^l +\nonumber \\
& +\theta(m-n) \frac{m(n+r+l+2)}{2(n+r+1)} H(u, u) \,
   \omega_{B}^{\vec{\lambda}}{}_{n+r-1, m+r-1}(u, \bar{u}) \chi^r t^l -\nonumber \\
& -\left[
     \frac{r(n+m+r+1)}{2(m+r+1)} \bar{H}(\bar{u}, \bar{\dd}) +
     \frac{r(n+m+r+1)}{2(n+r+1)} H(u, \dd)
   \right] 
   \omega_{B}^{\vec{\lambda}}{}_{n+r-1, m+r-1}(u, \bar{u}) \chi^{r-1} t^{l+1}+ \nonumber\\
& +\theta(m-n) \left[
     \frac{m(n+r+l+2)}{2(n+r+1)^2} H(u, \dd) - \right. \nonumber\\
&-\left. \frac{m(n+r+l+2)}{2(n+r+1)(m+r+1)} \bar{H}(\bar{u}, \bar{\dd})\right] 
   \omega_{C}^{\vec{\lambda}}{}_{n+r+1, m+r-1}(u, \bar{u}) \chi^r t^l -\nonumber\\
& -\frac{r(n+m+r+1)}{2(n+r+1)^2} H(\dd, \dd) \,
   \omega_{C}^{\vec{\lambda}}{}_{n+r+1, m+r-1}(u, \bar{u}) \chi^{r-1} t^{l+1} + \nonumber\\
& +\theta(n-m) \left[\frac{n(m+r+l+2)}{2(m+r+1)^2} \bar{H}(\bar{u}, \bar{\dd}) -\right. \nonumber&&\\
&-\left.\frac{n(m+r+l+2)}{2(n+r+1)(m+r+1)} H(u, \dd)\right] 
   \omega_{D}^{\vec{\lambda}}{}_{n+r-1, m+r+1}(u, \bar{u}) \chi^r t^l - \nonumber\\
& -\frac{r(n+m+r+1)}{2(m+r+1)^2} \bar{H}(\bar{\dd}, \bar{\dd}) \,
   \omega_{D}^{\vec{\lambda}}{}_{n+r-1, m+r+1}(u, \bar{u}) \chi^{r-1} t^{l+1}
\Bigg\}\,,\label{e: sigma one-form}&&
\end{align}

where the linearly independent Lorenz-irreducible two-forms are
\begin{align}
    &H(\dd, \dd) = H^{\alpha \beta}\dd_\alpha\dd_\beta\,, \quad H(u, \dd) = H^{\alpha \beta}u_\alpha\dd_\beta\,, \quad H(u, u) = H^{\alpha \beta}u_\alpha u_\beta\,,\nonumber\\
    &\bar{H}(\bar{\dd}, \bar{\dd}) = \bar{H}^{\dot{\alpha}\dot{\beta}}\bar{\dd}_{\dot{\alpha}} \bar{\dd}_{\dot{\beta}}\,, \quad \bar{H}(\bar{u}, \bar{\dd}) = \bar{H}^{\dot{\alpha}\dot{\beta}}\bar{u}_{\dot{\alpha}} \bar{\dd}_{\dot{\beta}}\,, \quad \bar{H}(\bar{u}, \bar{u}) = \bar{H}^{\dot{\alpha}\dot{\beta}}\bar{u}_{\dot{\alpha}} \bar{u}_{\dot{\beta}}\,.\label{basis 2 forms}
\end{align}

Linear independency of the irreducible two-forms allows us to split the equation $\sigma_- \omega^{\vec{\lambda}} = 0$ (\ref{e: sigma one-form}) into a system of equations that intertwine spin-tensors $\{\omega^{\vec{\lambda}}_A\,, \omega^{\vec{\lambda}}_B\,, \omega^{\vec{\lambda}}_C\,, \omega^{\vec{\lambda}}_D\}$ with different number of dotted and undotted indices. The presence of $\theta$-functions results in non-overlapping support with respect to parameters $n,m$ for most of their values, which significantly limits the pool of non-zero spin-tensors. Thus, the intertwining equations can be solved directly and the kernel of $\sigma_-$ can be fully determined.
\begin{align}
&\ker(\sigma_-)\vert_{\deg \Lambda = 1} = \Bigg\{ \lambda_h, \lambda_v \in 2 \mathbb{Z}_{\geqslant 0}\,, \quad n+m=\lambda_h\,, \quad 2(r+l)=\lambda_v:\nonumber\\
&\bullet \quad e(\dd, \bar{\dd}) \omega^{\vec{\lambda}}_{\frac{\lambda_h}{2}+r+1\,, \frac{\lambda_h}{2}+r+1}(u\,, \bar{u}) \chi^{r} t^l\,; \quad \bullet \quad e(u, \bar{u}) \omega^{\vec{\lambda}}_{\frac{\lambda_h}{2}-1\,, \frac{\lambda_h}{2}-1}(u\,, \bar{u}) t^{\frac{\lambda v}{2}}\,; \nonumber\\
&\bullet \quad \Bigg(e(u, \bar{u}) \chi t^{\frac{\lambda_v}{2}-1}+\frac{\lambda_h}{\left(\lambda_h+\lambda_v+2\right)}\{e(\dd, \bar{u})+e(u, \bar{\dd})\}t^{\frac{\lambda v}{2}}\Bigg) \omega^{\vec{\lambda}}_{\frac{\lambda_h}{2}\,, \frac{\lambda_h}{2}}(u\,,\bar{u})\,,\quad \lambda_v \geqslant 2\,;\nonumber\\
&\bullet \quad e(\dd,\bar{u})\omega^{\vec{\lambda}}_{n+1\,, m-1}(u\,,\bar{u}) t^{\frac{\lambda_v}{2}} = \frac{2m(n+2)}{(n+1)(2m+2+\lambda_v)}\sigma_-\left(\omega^{\vec{\lambda}}_{n+1\,, m-1}(u\,,\bar{u}) t^{\frac{\lambda_v}{2}}\right)\,, \; n \geq m\,; \nonumber\\
&\bullet \quad e(u,\bar{\dd})\omega^{\vec{\lambda}}_{n-1\,, m+1}(u\,,\bar{u}) t^{\frac{\lambda_v}{2}} = \frac{2n(m+2)}{(m+1)(2n+2+\lambda_v)} \sigma_-\left(\omega^{\vec{\lambda}}_{n-1\,, m+1}(u\,,\bar{u}) t^{\frac{\lambda_v}{2}}\right)\,, \;  m \geq n\,;\nonumber\\
& \bullet \quad \sigma_-\left(\omega^{\vec{\lambda}}_{\frac{\lambda_h}{2}+k+r+2,\frac{\lambda_h}{2}-k+r}\chi^{r+1}t^{l-1}\right) \,, \quad k \in \left\{1,\dots, \frac{\lambda_h}{2} - 1\right\}\,;\nonumber\\
& \bullet \quad \sigma_-\left(\omega^{\vec{\lambda}}_{\frac{\lambda_h}{2}-k+r,\frac{\lambda_h}{2}+k+r+2}\chi^{r+1}t^{l-1}\right) \,, \quad k \in \left\{1,\dots, \frac{\lambda_h}{2} - 1\right\}
\Bigg\}\,.
\end{align}

As the one-form $e(u,\bar{u})$ does not appear in the $\text{Im}(\sigma_-)\vert_{\deg\Lambda = 1}$ (\ref{e: sigma zero-form}), any element of $\ker(\sigma_-)\vert_{\deg \Lambda = 1}$ that includes $e(u,\bar{u})$ is not $\sigma_-$-exact. Notice that $r \in \{0,\dots, \lambda_v/2\}$ for any $\mathfrak{sp}(4)$ irreducible module and if the spin-tensor has equal amount of dotted and undotted indices then first two terms of $\sigma_-$ (\ref{e: sigma adj adj}) are absent and the third one reduces the $r$ parameter. These observations lead us to the conclusion that symmetric spin-tensors with $r = \lambda_v/2$ are not $\sigma_-$-exact. For the rest of the kernel elements we can show a $\sigma_-$-exact representation. Thus, $H^1(\sigma_-)$ is spanned by three 1-cocycles:

\begin{flalign}
    &H^1\left(\sigma_{-}\right)=\Bigg\{\bullet \quad e(\dd, \bar{\dd}) \omega^{\vec{\lambda}}_{\frac{\lambda_h+\lambda_v}{2}+1\,, \frac{\lambda_h+\lambda_v}{2}+1}(u\,, \bar{u}) \chi^{\frac{\lambda v}{2}}\,; \quad\bullet \quad e(u, \bar{u}) \omega^{\vec{\lambda}}_{\frac{\lambda_h}{2}-1\,, \frac{\lambda_h}{2}-1}(u\,, \bar{u}) t^{\frac{\lambda v}{2}}\,; \nonumber&&\\
    &\bullet \; \Bigg(e(u, \bar{u}) \chi t^{\frac{\lambda_v}{2}-1}+\frac{\lambda_h}{\left(\lambda_h+\lambda_v+2\right)}\{e(\dd, \bar{u})+e(u, \bar{\dd})\}t^{\frac{\lambda v}{2}}\Bigg) \omega^{\vec{\lambda}}_{\frac{\lambda_h}{2}\,, \frac{\lambda_h}{2}}(u\,,\bar{u})\,,\; \lambda_v \geqslant 2\,; \nonumber&&\\
    &\lambda_h, \lambda_v \in 2 \mathbb{Z}_{\geqslant 0}
\Bigg\}\,.&&
\end{flalign}

\subsubsection{$H^2(\sigma_-)$}

As shown above, there are six Lorenz-irreducible two-forms (\ref{basis 2 forms}), so the decomposition of a general two-form is
\begin{align}
& R^{\vec{\lambda}}(u,\bar{u},\chi,t)=\sum_{\substack{n+m=\lambda_h \\ 2(r+l)=\lambda_v}}\bigg\{H(\dd, \dd) R^{\vec{\lambda}}_{A}{}_{n+r+2, m+r}(u, \bar{u}) \chi^r t^l + \bar{H}(\bar{\dd}, \bar{\dd}) R^{\vec{\lambda}}_{\bar{A}}{}_{n+r, m+r+2}(u, \bar{u}) \chi^r t^l + \nonumber\\
& + H(u, u) R^{\vec{\lambda}}_{B}{}_{n+r-2, m+r}(u, \bar{u}) \chi^r t^l + \bar{H}(\bar{u}, \bar{u}) R^{\vec{\lambda}}_{\bar{B}}{}_{n+r, m+r-2}(u, \bar{u}) \chi^r t^l + \nonumber \\
&+H(u, \dd) R^{\vec{\lambda}}_{C}{}_{n+r, m+r}(u, \bar{u}) \chi^r t^l + \bar{H}(\bar{u}, \bar{\dd}) R^{\vec{\lambda}}_{\bar{C}}{}_{n+r, m+r}(u, \bar{u}) \chi^r t^l\bigg\} \,,
\end{align}
where once again Latin labels refer to different irreducible components of two-forms.

Since $\sigma_- R^{\vec{\lambda}}(u,\bar{u},\chi,t) \in \Lambda^3(M)\otimes \operatorname{Poly}(u,\bar{u},\chi,t)$ we need basis three-forms for the further decomposition, which are defined as
\begin{equation}
\mathcal{H}_{\alpha \dot{\alpha}}:=-e_{\alpha}{}^{\dot{\beta}} \bar{H}_{\dot{\beta} \dot{\alpha}}=-e_\alpha{}^{\dot{\beta}}\wedge e_{\gamma \dot{\beta}} \wedge e^\gamma{}_{\dot{\alpha}} = -e_\alpha{}^{\dot{\gamma}} \wedge e_{\beta \dot{\gamma}} \wedge e^\beta{}_{\dot{\alpha}} = e^\beta{}_{\dot{\alpha}} \wedge e_{\beta \dot{\gamma}} \wedge e_\alpha{}^{\dot{\gamma}}=e^\beta{}_{\dot{\alpha}} H_{\beta \alpha}\,.
\end{equation}
One can check that the following properties hold true
\begin{equation}
    e^{\alpha \dot{\alpha}}\wedge H^{\mu\nu} = -\frac{1}{3}\left(\varepsilon^{\alpha \mu} \mathcal{H}^{\nu \dot{\alpha}} + \varepsilon^{\alpha \nu} \mathcal{H}^{\mu \dot{\alpha}}\right)\,, \quad e^{\alpha \dot{\alpha}}\wedge \bar{H}^{\dot{\mu}\dot{\nu}} = \frac{1}{3}\left(\bar{\varepsilon}^{\dot{\alpha} \dot{\mu}} \mathcal{H}^{\alpha \dot{\nu}} + \bar{\varepsilon}^{\dot{\alpha} \dot{\nu}} \mathcal{H}^{\alpha \dot{\mu}}\right)\,.
\end{equation}

With these properties the image of $\sigma_-$ decomposes into 
\begin{flalign}
&\sigma_- R^{\vec{\lambda}}(u,\bar{u},\chi,t) = 
\sum_{\substack{n+m = \lambda_h \\ 2(r+l) = \lambda_v}} \Bigg\{ \nonumber \theta(m-n) \frac{2m(n+r+l+2)}{3(m+r+1)} \,
   \mathcal{H}(\dd, \bar{\dd}) \,
   R_{A}^{\vec{\lambda}}{}_{n+r+2,\, m+r} \, \chi^r t^l - \nonumber&&\\
&\quad -\theta(n-m) \frac{2n(m+r+l+2)}{3(n+r+1)} \,
   \mathcal{H}(\dd, \bar{\dd}) \,
   R_{\bar{A}}^{\vec{\lambda}}{}_{n+r,\, m+r+2} \, \chi^r t^l -\nonumber&&\\
&\quad -\theta(n-m) \frac{2n(m+r+l+2)}{3(m+r+1)} \,
   \mathcal{H}(u, \bar{u}) \,
   R_{B}^{\vec{\lambda}}{}_{n+r-2,\, m+r} \, \chi^r t^l +\nonumber&&\\
&\quad +\theta(m-n) \frac{2m(n+r+l+2)}{3(n+r+1)} \,
   \mathcal{H}(u, \bar{u}) \,
   R_{\bar{B}}^{\vec{\lambda}}{}_{n+r,\, m+r-2} \, \chi^r t^l +\nonumber&&\\
&\quad +\frac{2r(n+m+r+1)}{3(m+r+1)} \,
   \mathcal{H}(u, \bar{\dd}) \,
   R_{B}^{\vec{\lambda}}{}_{n+r-2,\, m+r} \, \chi^{r-1} t^{l+1} -\nonumber&&\\
&\quad -\frac{2r(n+m+r+1)}{3(n+r+1)} \,
   \mathcal{H}(\dd, \bar{u}) \,
   R_{\bar{B}}^{\vec{\lambda}}{}_{n+r,\, m+r-2} \, \chi^{r-1} t^{l+1} -\nonumber&&\\
   &\quad -\theta(n-m) \frac{n(m+r+l+2)(n+r+2)}{3(n+r+1)(m+r+1)} \,
   \mathcal{H}(\dd, \bar{u}) \,
   R_{C}^{\vec{\lambda}}{}_{n+r,\, m+r} \, \chi^r t^l +\nonumber&&\\
&\quad +\theta(m-n) \frac{m(n+r+l+2)(n+r)}{3(n+r+1)(m+r+1)} \,
   \mathcal{H}(u, \bar{\dd}) \,
   R_{C}^{\vec{\lambda}}{}_{n+r,\, m+r} \, \chi^r t^l +\nonumber&&
\end{flalign}
\begin{flalign}
&\quad +\frac{r(n+m+r+1)(n+r+2)}{3(n+r+1)(m+r+1)} \,
   \mathcal{H}(\dd, \bar{\dd}) \,
   R_{C}^{\vec{\lambda}}{}_{n+r,\, m+r} \, \chi^{r-1} t^{l+1} -\nonumber&&\\
&\quad -\theta(n-m) \frac{n(m+r+l+2)(m+r)}{3(n+r+1)(m+r+1)} \,
   \mathcal{H}(\dd, \bar{u}) \,
   R_{\bar{C}}^{\vec{\lambda}}{}_{n+r,\, m+r} \, \chi^r t^l +\nonumber&&\\
&\quad +\theta(m-n) \frac{m(n+r+l+2)(m+r+2)}{3(n+r+1)(m+r+1)} \,
   \mathcal{H}(u, \bar{\dd}) \,
   R_{\bar{C}}^{\vec{\lambda}}{}_{n+r,\, m+r} \, \chi^r t^l -\nonumber&&\\
&\quad -\frac{r(n+m+r+1)(m+r+2)}{3(n+r+1)(m+r+1)} \,
   \mathcal{H}(\dd, \bar{\dd}) \,
   R_{\bar{C}}^{\vec{\lambda}}{}_{n+r,\, m+r} \, \chi^{r-1} t^{l+1}
\Bigg\}\,.&&
\end{flalign}

There are four independent irreducible three-forms $\mathcal{H}(\cdot,\cdot)$ and, therefore, four equations on the spin-tensors $\{R_{A}^{\vec{\lambda}}\,,R_{B}^{\vec{\lambda}}\,,R_{C}^{\vec{\lambda}}\,,R_{\bar{A}}^{\vec{\lambda}}\,,R_{\bar{B}}^{\vec{\lambda}}\,,R_{\bar{C}}^{\vec{\lambda}}\}$. As in the case of one-forms, $\theta$-functions restrict the space of non-trivial solutions, and the system of equation on the spin-tensors can be solved directly. Thus, the kernel of $\sigma_-$ in the sector of two-forms is
\begin{flalign}
&\ker(\sigma_-)\vert_{\deg \Lambda = 2} = \Bigg\{ \lambda_h, \lambda_v \in 2 \mathbb{Z}_{\geqslant 0}\,, \quad n+m=\lambda_h\,, \quad 2(r+l)=\lambda_v: \nonumber\\
&\bullet \quad H(\dd, \dd) R^{\vec{\lambda}}_{n + r + 2\,, m + r}(u\,,\bar{u}) \chi^{r} t^{l}\,, \, n \geq m\, \quad \oplus \quad \bar{H}(\bar{\dd}, \bar{\dd}) R^{\vec{\lambda}}_{n + r\,, m + r + 2}(u\,,\bar{u}) \chi^{r} t^{l}\,,\, m \geq n\,;\nonumber\\
&\bullet \quad H(u, u) R^{\vec{\lambda}}_{n-2, m}(u\,,\bar{u}) t^{\frac{\lambda v}{2}}\,,\quad m \geqslant n \quad \oplus \quad \bar{H}(\bar{u}, \bar{u}) R^{\vec{\lambda}}_{n, m-2} (u\,,\bar{u})t^{\frac{\lambda v}{2}}\,,\quad n \geqslant m\,;\nonumber\\
&\bullet \quad H(u, \dd) R^{\vec{\lambda}}_{\frac{\lambda_h}{2}, \frac{\lambda_h}{2}} (u, \bar{u})t^{\frac{\lambda}{2}} \quad \oplus \quad \bar{H}(\bar{u}, \bar{\dd}) R^{\vec{\lambda}}_{\frac{\lambda_h}{2}, \frac{\lambda_h}{2}} (u, \bar{u})t^{\frac{\lambda}{2}}\,; \nonumber\\
&\bullet \quad \bigg(H(u, u)+\bar{H}(\bar{u},\bar{u})\bigg) R^{\vec{\lambda}}_{\frac{\lambda_h}{2}-1\,, \frac{\lambda_h}{2}-1} (u\,,\bar{u})t^{\frac{\lambda v}{2}}\,,\quad \lambda_h \geqslant 2\,;\nonumber\\
&\bullet \quad H(u, \dd) R_{C}^{\vec{\lambda}}{}_{n, m}(u,\bar{u}) t^{\frac{\lambda v}{2}}+\bar{H}(\bar{u}, \bar{\dd}) R_{\bar{C}}^{\vec{\lambda}}{}_{n, m}(u,\bar{u}) t^{\frac{\lambda v}{2}}-\frac{1}{4\left(\lambda_n+2\right)} \frac{m\left(\lambda_v+2 n+4\right)}{n+1}\times\nonumber\\
&\times \bigg[n H(u, u) R_{C}^{\vec{\lambda}}{}_{n,m}(u,\bar{u})+(m+2) H(u, u) R_{\bar{C}}^{\vec{\lambda}}{}_{n,m}(u,\bar{u}) \bigg]\chi t^{\frac{\lambda v}{2}-1}\,, \quad  m > n \quad \oplus \quad \text{c.c.}\,;\nonumber\\
&\bullet \quad \bigg(H(\dd, \dd)+\bar{H}(\bar{\dd},\bar{\dd})\bigg) R^{\vec{\lambda}}_{\frac{\lambda_h}{2}+r+1\,, \frac{\lambda_h}{2}+r+1}(u\,,\bar{u}) \chi^{r} t^l\,;\nonumber\\
&\bullet \quad \bigg(H(u, \dd)+\bar{H}(\bar{u},\bar{\dd})\bigg) R^{\vec{\lambda}}_{\frac{\lambda_h}{2}+r\,, \frac{\lambda_h}{2}+r} (u\,,\bar{u})\chi^{r} t^l \,;\nonumber\\
&\bullet \quad \left(H(u, \dd) \chi^{r+1} t^{l-1} -\frac{2(r+1)\left(\lambda_h+r+2\right)\left(\lambda_h+2r+6\right)}{\left(\lambda_h+2r+4\right)\left(\lambda_h+2\right)\left(\lambda_h+\lambda_v+2\right)} H(\dd, \dd) \chi^r t^l\right)\times\nonumber\\
&\times \left[R_{C}^{\vec{\lambda}}{}_{\frac{\lambda_h}{2} + r + 1, \frac{\lambda_h}{2} + r + 1}(u\,,\bar{u}) - R_{\bar{C}}^{\vec{\lambda}}{}_{\frac{\lambda_h}{2} + r + 1, \frac{\lambda_h}{2} + r + 1}(u\,,\bar{u}) \right]\,, \quad r \neq \frac{\lambda_v}{2} \quad \oplus \quad \text{c.c.}\,;\nonumber\\
&\bullet \quad H(u, \dd) R_{C}^{\vec{\lambda}}{}_{n+r, m+r}(u\,,\bar{u}) \chi^r t^l + \bar{H}(\bar{u}, \bar{\dd}) R_{\bar{C}}^{\vec{\lambda}}{}_{n+r, m+r}(u\,,\bar{u}) \chi^r t^l -\nonumber\\
&- \frac{m(n+r+l+2)}{2(r+1)(\lambda_h+r+2)(n+r+1)} \bigg[(n+r) H(u, u) R_{C}^{\vec{\lambda}}{}_{n+r,m+r}(u\,,\bar{u}) +\nonumber\\
&+ (m+r+2) H(u, u) R_{\bar{C}}^{\vec{\lambda}}{}_{n+r,m+r}(u\,,\bar{u})\bigg] \chi^{r+1} t^{l-1} +\nonumber
\end{flalign}
\begin{flalign}
& +\frac{r(\lambda_h+r+1)}{(m+1)\left(\lambda_v + 2n + 2\right)(n+r+1)}\bigg[(m+r+2) H(\dd, \dd) R_{\bar{C}}^{\vec{\lambda}}{}_{n+r,m+r}(u\,,\bar{u}) - \nonumber \\
& -(n+r+2) H(\dd, \dd) R_{C}^{\vec{\lambda}}{}_{n+r,m+r}\bigg] \chi^{r-1} t^{l+1} \,, \quad m > n\,, r \neq \frac{\lambda_v}{2} \quad \oplus \quad \text{c.c.}\,;\nonumber\\
&\bullet \quad \bigg(H(u, \dd)\chi^{\frac{\lambda_v}{2}}-\frac{\lambda_v+2n}{\lambda_v+2m+4} \bar{H}(\bar{u}, \bar{\dd}) \chi^{\frac{\lambda_v}{2}}-\frac{\lambda_v\left(\lambda_v+2\lambda_h+2\right)}{2(m+1)\left(\lambda_v+2 n+2\right)} H(\dd, \dd) \chi^{\frac{\lambda_v}{2}-1} t\bigg)\times \nonumber\\
&\times R_C^{\vec{\lambda}}{}_{n + \frac{\lambda_v}{2}, m + \frac{\lambda_v}{2}}(u\,,\bar{u})\,, \quad m > n \quad \oplus \quad \text{c.c.}\Bigg\}\,.\label{e: two-form kernel adj-adj}
\end{flalign}

Now we have to discard those two-forms that are $\sigma_-$-exact. Despite the apparent complexity of filtering out $\sigma_-$-exact forms, the process turns out to be no more difficult than that for one-forms.

Once again we should recall that for fixed weights $(\lambda_h, \lambda_v)$ parameters $(n,m,r,l)$ cover the finite set of values. Thus, the grading $G=|\hat{N}_u-\hat{\bar{N}}_u|+2\hat{N}_\chi$ is bounded from  above and below, so we expect cohomology on the boundaries of $G$ and potentially at the points where $|\hat{N}_u-\hat{\bar{N}}_u|$ is zero. It turns out that our expectations are valid as most of the elements in $\ker(\sigma_-)$ are $\sigma_-$-exact which can be shown directly with the help of (\ref{e: sigma one-form}). For example,
\thickmuskip=1.5mu 
\medmuskip=1mu
\thinmuskip=1.5mu
\begin{align}
&H(u, \dd) R_{C}^{\vec{\lambda}}{}_{n+r, m+r}(u\,,\bar{u}) \chi^r t^l + \bar{H}(\bar{u}, \bar{\dd}) R_{\bar{C}}^{\vec{\lambda}}{}_{n+r, m+r}(u\,,\bar{u}) \chi^r t^l - \frac{m(n+r+l+2)}{2(r+1)(\lambda_h+r+2)(n+r+1)}\times \nonumber\\
&\times\bigg[(n+r) H(u, u) R_{C}^{\vec{\lambda}}{}_{n+r,m+r}(u\,,\bar{u}) + (m+r+2) H(u, u) R_{\bar{C}}^{\vec{\lambda}}{}_{n+r,m+r}(u\,,\bar{u})\bigg] \chi^{r+1} t^{l-1} +\nonumber\\
& +\frac{r(\lambda_h+r+1)}{(m+1)\left(\lambda_v + 2n + 2\right)(n+r+1)}\bigg[(m+r+2) H(\dd, \dd) R_{\bar{C}}^{\vec{\lambda}}{}_{n+r,m+r}(u\,,\bar{u}) -\nonumber\\
& -(n+r+2) H(\dd, \dd) R_{C}^{\vec{\lambda}}{}_{n+r,m+r}\bigg] \chi^{r-1} t^{l+1} = \frac{2(n+r)(n+r+2)}{(m+1)(\lambda_v +2n+2)(n+r+1)} \times\nonumber\\
&\times\sigma_-\bigg(e(\dd, \bar{u}) R_{C}^{\vec{\lambda}}{}_{n+r,m+r}(u\,,\bar{u}) \chi^r t^l-\frac{(m+1)(\lambda_v+2n+2)}{2(r+1)(\lambda_h+r+2)} e(u, \bar{u}) R_{C}^{\vec{\lambda}}{}_{n+r,m+r}(u\,,\bar{u}) \chi^{r+1} t^{l-1}\bigg)-\nonumber\\
&-\frac{2(m+r+2)(n+r)}{(m+1)(\lambda_v+2n+2)(n+r+1)}\sigma_-\bigg(e(\dd, \bar{u}) R_{\bar{C}}^{\vec{\lambda}}{}_{n+r,m+r}(u\,,\bar{u}) \chi^r t^l + \nonumber\\
&+ \frac{(m+1)(\lambda_v+2 n+2)(n+r+2)}{2(r+1)(n+r)\left(\lambda_h + r +2\right)} e(u, \bar{u}) R_{\bar{C}}^{\vec{\lambda}}{}_{n+r, m+r} (u\,,\bar{u})\chi^{r+1} t^{l-1}\bigg)\,,
\end{align}
\begin{align}
&\bigg(H(u, \dd)\chi^{\frac{\lambda_v}{2}}-\frac{\lambda_v+2n}{\lambda_v+2m+4} \bar{H}(\bar{u}, \bar{\dd}) \chi^{\frac{\lambda_v}{2}}-\frac{\lambda_v\left(\lambda_v+2\lambda_h+2\right)}{2(m+1)\left(\lambda_v+2 n+2\right)} H(\dd, \dd) \chi^{\frac{\lambda_v}{2}-1} t\bigg)\times\nonumber \\
&\times R_C^{\vec{\lambda}}{}_{n + \frac{\lambda_v}{2}, m + \frac{\lambda_v}{2}}(u\,,\bar{u}) = \frac{2(\lambda_v+2n)}{(m+1)(\lambda_v+2n+2)}\sigma_-\left(e(\dd, \bar{u})R_C^{\vec{\lambda}}{}_{n + \frac{\lambda_v}{2}, m + \frac{\lambda_v}{2}}(u\,,\bar{u}) \chi^{\frac{\lambda_v}{2}}\right)\,.
\end{align}
\thickmuskip=5mu 
\medmuskip=4mu
\thinmuskip=3mu

It is important to note that some elements in (\ref{e: two-form kernel adj-adj}) are actually representatives of the same cohomology class. A direct check shows that
\thickmuskip=1mu 
\medmuskip=1mu
\thinmuskip=1mu
\begin{align}
&\bigg(H(\dd, \dd)+\bar{H}(\bar{\dd},\bar{\dd})\bigg) R^{\vec{\lambda}}_{\frac{\lambda_v+\lambda_h}{2}\,, \frac{\lambda_v+\lambda_h}{2}}(u\,,\bar{u}) \chi^{\frac{\lambda v}{2}-1}t \simeq \bigg(H(u, \dd)+\bar{H}(\bar{u},\bar{\dd})\bigg) R^{\vec{\lambda}}_{\frac{\lambda_v+\lambda_h}{2}\,, \frac{\lambda_v+\lambda_h}{2}} (u\,,\bar{u})\chi^{\frac{\lambda v}{2}}\nonumber \\
&\simeq\bigg(H(u, \dd)\chi^{\frac{\lambda_v}{2}} - \frac{\lambda_v(\lambda_h + \lambda_v + 4)(\lambda_v + 2\lambda_h + 2)}{2(\lambda_h+2)(\lambda_h+\lambda_v+2)^2}H(\dd, \dd)\chi^{\frac{\lambda_v}{2}-1}t\bigg) R^{\vec{\lambda}}_{\frac{\lambda_v+\lambda_h}{2}\,,\frac{\lambda_v+\lambda_h}{2}}(u\,, \bar{u})\,, \quad\lambda_v \geqslant 2\,.
\end{align}
\thickmuskip=5mu 
\medmuskip=4mu
\thinmuskip=3mu

Therefore, the second cohomology group is
\begin{equation}
H^2\left(\sigma_{-}\right)=\left\{\begin{array}{ll}
\bullet \quad H(\dd, \dd) R^{\vec{\lambda}}_{\frac{\lambda_v}{2}+\lambda_h+2\,, \frac{\lambda_v}{2}}(u\,,\bar{u}) \chi^{\frac{\lambda v}{2}} \quad \oplus \quad \bar{H}(\bar{\dd}, \bar{\dd}) R^{\vec{\lambda}}_{\frac{\lambda_v}{2}\,,\frac{\lambda_v}{2}+\lambda_h+2}(u\,,\bar{u}) \chi^{\frac{\lambda v}{2}}\,;\\
\bullet \quad \left(H(\dd, \dd)+\bar{H}(\bar{\dd},\bar{\dd})\right) R^{\vec{\lambda}}_{\frac{\lambda_v+\lambda_h}{2}+1\,, \frac{\lambda_v+\lambda_h}{2}+1}(u\,,\bar{u}) \chi^{\frac{\lambda v}{2}}\,;\\
\bullet \quad \left(H(u, u)+\bar{H}(\bar{u},\bar{u})\right) R^{\vec{\lambda}}_{\frac{\lambda_h}{2}-1\,, \frac{\lambda_h}{2}-1} (u\,,\bar{u})t^{\frac{\lambda v}{2}}\,, \quad \lambda_h \geqslant 2\,;\\
\bullet \quad \left(H(\dd, \dd)+\bar{H}(\bar{\dd},\bar{\dd})\right) R^{\vec{\lambda}}_{\frac{\lambda_v+\lambda_h}{2}\,, \frac{\lambda_v+\lambda_h}{2}}(u\,,\bar{u}) \chi^{\frac{\lambda v}{2}-1}t \overset{\text{Im}(\sigma_-)}{\simeq} \\
\overset{\text{Im}(\sigma_-)}{\simeq}\left(H(u, \dd)+\bar{H}(\bar{u},\bar{\dd})\right) R^{\vec{\lambda}}_{\frac{\lambda_v+\lambda_h}{2}\,, \frac{\lambda_v+\lambda_h}{2}} (u\,,\bar{u})\chi^{\frac{\lambda v}{2}} \overset{\text{Im}(\sigma_-)}{\simeq} \\
\overset{\text{Im}(\sigma_-)}{\simeq} \left(H(u, \dd)\chi^{\frac{\lambda_v}{2}} - \xi H(\dd, \dd)\chi^{\frac{\lambda_v}{2}-1}t\right) R^{\vec{\lambda}}_{\frac{\lambda_v+\lambda_h}{2}\,,\frac{\lambda_v+\lambda_h}{2}}(u\,, \bar{u})\,, \lambda_v \geqslant 2\,;\\
\lambda_h, \lambda_v \in 2 \mathbb{Z}_{\geqslant 0} \;; \quad \xi = \frac{\lambda_v(\lambda_h + \lambda_v + 4)(\lambda_v + 2\lambda_h + 2)}{2(\lambda_h+2)(\lambda_h+\lambda_v+2)^2}
\end{array}\right\}\,.
\end{equation}
The last cohomology class is absent in the case of $\lambda_v = 0$.

\subsection{Summary for $H^{0,1,2}(\sigma_-)$}

Here we collect the final results for the cocycles associated with the gauge parameters, fields and field equations in the $(adj\otimes adj)$ sector of $B_2$ CHS on the $AdS_4$ background.

Recall that fields $\omega^{\vec{\lambda}}$ are valued in $\mathfrak{so}(3,2)$-irreps parametrized by weights $(\lambda_h,\lambda_v)$ that decompose into the direct sum of Lorenz-irreps
\begin{equation} 
\begin{picture}(10,18)(0,7)
{
\put(05,00){\line(1,0){35}}%
\put(05,10){\line(1,0){35}}%
\put(05,0){\line(0,1){10}}%
\put(40,0){\line(0,1){10}}%
\put(17,3){\scriptsize  ${M}$}%
}
\end{picture}\begin{picture}(50,18)(10,7)
{\put(35,13){\scriptsize  ${N}$}
 \put(05,20){\line(1,0){60}}%
\put(05,10){\line(1,0){60}}%
\put(05,10){\line(0,1){10}}%
\put(65,10){\line(0,1){10}}%
}
\end{picture}\quad \bigg\vert_{\mathfrak{so}(3,2)} = \bigoplus_{\substack{i = 0}}^{N-M}\bigoplus_{\substack{j = 0}}^{M}
\begin{picture}(10,18)(0,7)
{
\put(05,00){\line(1,0){35}}%
\put(05,10){\line(1,0){35}}%
\put(05,0){\line(0,1){10}}%
\put(40,0){\line(0,1){10}}%
\put(20,3){\scriptsize  ${j}$}%
}
\end{picture}\begin{picture}(50,18)(10,7)
{\put(25,13){\scriptsize  ${M + i}$}
 \put(05,20){\line(1,0){60}}%
\put(05,10){\line(1,0){60}}%
\put(05,10){\line(0,1){10}}%
\put(65,10){\line(0,1){10}}%
}
\end{picture}\quad \bigg\vert_{\mathfrak{so}(3,1)}
\,,
\end{equation}
where $N = (\lambda_h + \lambda_v)/2$ and $M = \lambda_h/2$.

According to the classification theorem stated in section \ref{sec: Interpretation of cohomology} applied to the system of one-form fields, $H^0(\sigma_-)$ represents parameters of the differential gauge symmetries. It is spanned by the zero-forms
\begin{equation}
H^0(\sigma_-)=\left\{\mathcal{E}^{\vec{\lambda}}_{\frac{\lambda_h}{2}\,, \frac{\lambda_h}{2}}(u\,, \bar{u}) t^{\frac{\lambda_v}{2}} ; \quad \lambda_h, \lambda_v \in 2 \mathbb{Z}_{\geqslant 0}\right\}\,.
\end{equation}

Gauge parameters correspond to the shortest possible one-row $\mathfrak{so}(3,1)$ Young diagram that results from the restriction of the two-row $\mathfrak{so}(3,2)$ Young diagram.

This indicates that the $(adj\otimes adj)$ sector should encode symmetric massless fields or partially massless fields.

Cohomology $H^1(\sigma_-)$ represents the dynamical fields. It is spanned by the three 1-cocycles:
\begin{equation}
H^1\left(\sigma_{-}\right)=\left\{\begin{array}{ll}
\bullet \quad e(\dd, \bar{\dd}) \omega^{\vec{\lambda}}_{\frac{\lambda_h+\lambda_v}{2}+1\,, \frac{\lambda_h+\lambda_v}{2}+1}(u\,, \bar{u}) \chi^{\frac{\lambda v}{2}}\,;\\
\bullet \quad e(u, \bar{u}) \omega^{\vec{\lambda}}_{\frac{\lambda_h}{2}-1\,, \frac{\lambda_h}{2}-1}(u\,, \bar{u}) t^{\frac{\lambda v}{2}}\,; \\
\bullet \; \left(e(u, \bar{u}) \chi t^{\frac{\lambda_v}{2}-1}+\frac{\lambda_h}{\left(\lambda_h+\lambda_v+2\right)}\{e(\dd, \bar{u})+e(u, \bar{\dd})\}t^{\frac{\lambda v}{2}}\right) \omega^{\vec{\lambda}}_{\frac{\lambda_h}{2}\,, \frac{\lambda_h}{2}}(u\,,\bar{u})\,,\; \lambda_v \geqslant 2\,; \\
\lambda_h, \lambda_v \in 2 \mathbb{Z}_{\geqslant 0}\,, \quad n+m=\lambda_h\,, \quad 2(r+l)=\lambda_v
\end{array}\right\}\,.
\end{equation}
The first cocycle corresponds to the longest one-row $\mathfrak{so}(3,1)$ diagram in the decomposition of the two-row $\mathfrak{so}(3,2)$ Young diagram. The second and third ones are the shortest one-row $\mathfrak{so}(3,1)$ diagrams. 

As was shown in \cite{Skvortsov:2006at}, an unfolding of partially massless fields is implemented via one-form connections that have the symmetry properties of two-row $\mathfrak{so}(d-1,2)$ Young tableaux with the first row of length $(s-1)$ and the second row of length $(s - t)$, where $s$ is a spin and $t \in \{1,\dots, s\}$. Here $t$ defines the depth of masslessness, i.e., the equations of motion for partially massless fields are invariant under the gauge transformation with at least $t$ derivatives in a leading term.

Therefore, the 1-cocycles encode partially massless fields of depth $t = 1 + \frac{\lambda_v}{2}$, if $\lambda_v \neq 0$. The appearance of a gap in lengths of corresponding $\mathfrak{so}(3,1)$ Young diagrams is attributed to the placement of partially massless fields between massless and massive. In a description of massive spin-$s$ fields we need a set of auxiliary fields with symmetry properties of one-row Lorenz diagrams with lengths from $(s - 2)$ to $0$ in addition to the field represented by a one-row Lorenz diagram of length $s$ \cite{Fierz:1939ix, Singh:1974qz}. These auxiliary fields vanish on-shell. In case of partially massless fields \cite{Deser:2001pe, Deser:2001us, Zinoviev:2001dt, Zinoviev:2002ye} the list of auxiliary fields reduces to the $\{(s-2), \dots, (s-t-1)\}$. However, fields $(s-t)$ and $(s-t-1)$ do not vanish on-shell, so they appear in $H^1(\sigma_-)$. In the case of zero weight $\lambda_v$ depth of masslessness $t = 1$ and the last cocycle is absent. The leftover ones encode the traceless part of symmetric massless fields and their traces, i.e., Fronsdal fields of all spins $s \geq 1$.

Cohomology group $H^2(\sigma_-)$, which represents gauge invariant differential operators on the fields, is spanned by four families of 2-cocycles:
\begin{flalign}
    &H^2\left(\sigma_{-}\right)=\Bigg\{\bullet \quad H(\dd, \dd) R^{\vec{\lambda}}_{\frac{\lambda_v}{2}+\lambda_h+2\,, \frac{\lambda_v}{2}}(u\,,\bar{u}) \chi^{\frac{\lambda v}{2}} \quad \oplus \quad \bar{H}(\bar{\dd}, \bar{\dd}) R^{\vec{\lambda}}_{\frac{\lambda_v}{2}\,,\frac{\lambda_v}{2}+\lambda_h+2}(u\,,\bar{u}) \chi^{\frac{\lambda v}{2}}\,;\nonumber&&\\
    &\bullet \quad \left(H(\dd, \dd)+\bar{H}(\bar{\dd},\bar{\dd})\right) R^{\vec{\lambda}}_{\frac{\lambda_v+\lambda_h}{2}+1\,, \frac{\lambda_v+\lambda_h}{2}+1}(u\,,\bar{u}) \chi^{\frac{\lambda v}{2}}\,;\nonumber&&\\
    &\bullet \quad \left(H(u, u)+\bar{H}(\bar{u},\bar{u})\right) R^{\vec{\lambda}}_{\frac{\lambda_h}{2}-1\,, \frac{\lambda_h}{2}-1} (u\,,\bar{u})t^{\frac{\lambda v}{2}}\,, \quad \lambda_h \geqslant 2\,;\nonumber\\
     &\bullet \quad \left(H(\dd, \dd)+\bar{H}(\bar{\dd},\bar{\dd})\right) R^{\vec{\lambda}}_{\frac{\lambda_v+\lambda_h}{2}\,, \frac{\lambda_v+\lambda_h}{2}}(u\,,\bar{u}) \chi^{\frac{\lambda v}{2}-1}t \overset{\text{Im}(\sigma_-)}{\simeq} \nonumber&&\\
     &\overset{\text{Im}(\sigma_-)}{\simeq}\left(H(u, \dd)+\bar{H}(\bar{u},\bar{\dd})\right) R^{\vec{\lambda}}_{\frac{\lambda_v+\lambda_h}{2}\,, \frac{\lambda_v+\lambda_h}{2}} (u\,,\bar{u})\chi^{\frac{\lambda v}{2}} \overset{\text{Im}(\sigma_-)}{\simeq} \nonumber&&\\
     &\overset{\text{Im}(\sigma_-)}{\simeq} \left(H(u, \dd)\chi^{\frac{\lambda_v}{2}} - \xi H(\dd, \dd)\chi^{\frac{\lambda_v}{2}-1}t\right) R^{\vec{\lambda}}_{\frac{\lambda_v+\lambda_h}{2}\,,\frac{\lambda_v+\lambda_h}{2}}(u\,, \bar{u})\,, \lambda_v \geqslant 2\,;\nonumber&&
\end{flalign}
\begin{flalign}
    &\lambda_h, \lambda_v \in 2 \mathbb{Z}_{\geqslant 0} \;; \quad \xi = \frac{\lambda_v(\lambda_h + \lambda_v + 4)(\lambda_v + 2\lambda_h + 2)}{2(\lambda_h+2)(\lambda_h+\lambda_v+2)^2}\Bigg\}\,.\label{e: cohomology-2 spin-tensor}&&
\end{flalign}

The first family of 2-cocycles corresponds to the two-row $\mathfrak{so}(3,1)$ diagrams with row lengths $\left(\frac{\lambda_h + \lambda_v}{2} + 1,\frac{\lambda_h}{2} + 1\right)$. These cocycles have the highest grading and, therefore, can be associated with the Weyl cohomology (Chevalley-Eilenberg cocycle), which, when implemented via zero-form fields $C$, glues zero-form sector of the theory to the one-form sector (in a standard $4d$ HS it glues generalized Weyl tensors to the one-form fields at the highest value of grading $G$). The rest of the families describe differential field equations on primary fields. This can be verified by examining the tensor structure of the equations and seeing that the corresponding $\mathfrak{so}(3,1)$ diagrams have matching row lengths $\left(\frac{\lambda_h + \lambda_v}{2} + 1,0\right)$, $\left(\frac{\lambda_h}{2} - 1, 0\right)$ and $\left(\frac{\lambda_h + \lambda_v}{2},0\right)$. Note that the last family exists only if $\lambda_v \neq 0$, so the second and third families represent two irreducible components of the Fronsdal cocycle in case of the zero $\lambda_v$ weight.

The lower cohomology groups obtained for irreducible modules coincide with those of \cite{Skvortsov:2009nv}, where an analysis of $\sigma_-$ cohomology for an unfolded partially massless fields was performed. This means that in $(adj \otimes adj)$ one-form sector of $B_2$ CHS symmetric massless and partially massless $AdS_4$ fields are encoded, as was conjectured in \cite{Tarusov:2025sre}. In a unitary truncated system all fields are on-shell with partially massless fields being topological (the corresponding Weyl tensors are zero) and most of the copies of massless fields being topological as well. The status of fields in a non-truncated case depends on the content of the vertex on the r.h.s. of (\ref{e: adj-adj omega}).  

It is important to note that each family is defined up to $\sigma_-$-exact terms which will be important for the analysis of the gluing vertex in the r.h.s. of (\ref{e: adj-adj omega}).

Cohomology $H^{0,1,2}(\sigma_-)$ are provided in terms of auxiliary spinor variables ($u,\bar{u}$), but there is a straightforward way to rewrite it in terms of $Y_i$ oscillators, replacing each $\chi^r$ with $(\dd\, \bar{\dd} \,\bar{X})^r$, $t^l$ with $\left(\sum_{p+q=l} \Gamma_{n, m, r}^{p, q} \zeta^p \bar{\zeta}{}^q\right)$ and restoring $(\dd y_1)^n (\bar{\dd}\bar{y}_1)^m$ prefactor. Thus, the lowest cohomology groups restricted to the class of $\mathfrak{sl}^h(2)\oplus \mathfrak{sl}^v(2)$ (HW, LW) vectors are
\thickmuskip=2mu 
\medmuskip=2mu
\thinmuskip=1mu
\begin{equation}\label{e: cohomology-0}
H^0(\sigma)=\left\{\left(\sum_{p+q=\frac{\lambda_v}{2}} \Gamma_{\frac{\lambda_h}{2}, \frac{\lambda_h}{2}, 0}^{p, q} \zeta^p \bar{\zeta}{}^q\right)(\dd y_1)^{\frac{\lambda_h}{2}} (\bar{\dd}\bar{y}_1)^{\frac{\lambda_h}{2}}\mathcal{E}^{\vec{\lambda}}_{\frac{\lambda_h}{2}\,, \frac{\lambda_h}{2}}(u\,, \bar{u}); \quad \lambda_h, \lambda_v \in 2 \mathbb{Z}_{\geqslant 0}\right\}\,.
\end{equation}
\begin{flalign}
    &H^1\left(\sigma_{-}\right)=\Bigg\{\bullet \quad (\dd\, \bar{\dd} \,X)^{\frac{\lambda v}{2}}(\dd y_1)^{\frac{\lambda_h}{2}} (\bar{\dd}\bar{y}_1)^{\frac{\lambda_h}{2}}e(\dd, \bar{\dd}) \omega^{\vec{\lambda}}_{\frac{\lambda_h+\lambda_v}{2}+1\,, \frac{\lambda_h+\lambda_v}{2}+1}(u\,, \bar{u}) \,;\nonumber&&\\
    &\bullet \quad\left(\displaystyle\sum_{p+q=\frac{\lambda_v}{2}} \Gamma_{\frac{\lambda_h}{2}, \frac{\lambda_h}{2}, 0}^{p, q} \zeta^p \bar{\zeta}{}^q\right) (\dd y_1)^{\frac{\lambda_h}{2}} (\bar{\dd}\bar{y}_1)^{\frac{\lambda_h}{2}}e(u, \bar{u}) \omega^{\vec{\lambda}}_{\frac{\lambda_h}{2}-1\,, \frac{\lambda_h}{2}-1}(u\,, \bar{u}) \,; \nonumber&&\\
    &\bullet \quad \left[\left(\displaystyle\sum_{p+q=\frac{\lambda_v}{2}-1} \Gamma_{\frac{\lambda_h}{2}, \frac{\lambda_h}{2}, 1}^{p, q} \zeta^p \bar{\zeta}{}^q\right)(\dd\, \bar{\dd} \,X)(\dd y_1)^{\frac{\lambda_h}{2}} (\bar{\dd}\bar{y}_1)^{\frac{\lambda_h}{2}}e(u, \bar{u}) +\nonumber\right.&&
\end{flalign}
\begin{flalign}
    &\left.+\frac{\lambda_h}{\left(\lambda_h+\lambda_v+2\right)}\left\{\left(\displaystyle\sum_{p+q=\frac{\lambda_v}{2}} \Gamma_{\frac{\lambda_h}{2}-1, \frac{\lambda_h}{2}+1, 0}^{p, q} \zeta^p \bar{\zeta}{}^q\right)(\dd y_1)^{\frac{\lambda_h}{2}-1} (\bar{\dd}\bar{y}_1)^{\frac{\lambda_h}{2}+1}e(\dd, \bar{u})+ \nonumber\right.\right.&&\\
    &\left.\left.+\left(\displaystyle\sum_{p+q=\frac{\lambda_v}{2}} \Gamma_{\frac{\lambda_h}{2}+1, \frac{\lambda_h}{2}-1, 0}^{p, q} \zeta^p \bar{\zeta}{}^q\right)(\dd y_1)^{\frac{\lambda_h}{2}+1} (\bar{\dd}\bar{y}_1)^{\frac{\lambda_h}{2}-1}e(u, \bar{\dd})\right\}\right]\omega^{\vec{\lambda}}_{\frac{\lambda_h}{2}\,, \frac{\lambda_h}{2}}(u\,,\bar{u})\,, \lambda_v \geqslant 2\,; \nonumber&&\\
    &\lambda_h, \lambda_v \in 2 \mathbb{Z}_{\geqslant 0}\Bigg\}\,.\label{e: cohomology-1}
\end{flalign}
\begin{flalign}\label{e: cohomology-2}
&H^2\left(\sigma_{-}\right)=\Bigg\{ \bullet \quad (\dd\, \bar{\dd} \,X)^{\frac{\lambda v}{2}} (\dd y_1)^{\lambda_h} H(\dd, \dd) R^{\vec{\lambda}}_{\frac{\lambda_v}{2}+\lambda_h+2\,, \frac{\lambda_v}{2}}(u\,,\bar{u})  \quad \oplus \nonumber\\
&\oplus \quad (\dd\, \bar{\dd} \,X)^{\frac{\lambda v}{2}} (\bar{\dd}\bar{y}_1)^{\lambda_h} \bar{H}(\bar{\dd}, \bar{\dd}) R^{\vec{\lambda}}_{\frac{\lambda_v}{2}\,,\frac{\lambda_v}{2}+\lambda_h+2}(u\,,\bar{u}) \,; \nonumber\\
&\bullet \quad (\dd\, \bar{\dd} \,X)^{\frac{\lambda v}{2}}\left( (\dd y_1)^{\frac{\lambda_h}{2}-1} (\bar{\dd}\bar{y}_1)^{\frac{\lambda_h}{2}+1} H(\dd, \dd) + \right.\nonumber\\
&\left.+ (\dd y_1)^{\frac{\lambda_h}{2}+1} (\bar{\dd}\bar{y}_1)^{\frac{\lambda_h}{2}-1} \bar{H}(\bar{\dd},\bar{\dd})\right) R^{\vec{\lambda}}_{\frac{\lambda_v+\lambda_h}{2}+1\,, \frac{\lambda_v+\lambda_h}{2}+1}(u\,,\bar{u}) \,;\nonumber\\
&\bullet \quad \left[ \left(\displaystyle\sum_{p+q=\frac{\lambda_v}{2}} \Gamma_{\frac{\lambda_h}{2}+1, \frac{\lambda_h}{2}-1, 0}^{p, q} \zeta^p \bar{\zeta}{}^q\right)(\dd y_1)^{\frac{\lambda_h}{2}+1} (\bar{\dd}\bar{y}_1)^{\frac{\lambda_h}{2}-1} H(u, u) + \right.\nonumber\\
&\left.+ \left(\displaystyle\sum_{p+q=\frac{\lambda_v}{2}} \Gamma_{\frac{\lambda_h}{2}-1, \frac{\lambda_h}{2}+1, 0}^{p, q} \zeta^p \bar{\zeta}{}^q\right)(\dd y_1)^{\frac{\lambda_h}{2}-1} (\bar{\dd}\bar{y}_1)^{\frac{\lambda_h}{2}+1} \bar{H}(\bar{u},\bar{u})\right]\times\nonumber\\
&\times R^{\vec{\lambda}}_{\frac{\lambda_h}{2}-1\,, \frac{\lambda_h}{2}-1} (u\,,\bar{u})\,, \lambda_h \geqslant 2\,;\nonumber\\
&\bullet \quad (\dd\, \bar{\dd} \,X)^{\frac{\lambda v}{2}} (\dd y_1)^{\frac{\lambda_h}{2}} (\bar{\dd}\bar{y}_1)^{\frac{\lambda_h}{2}} \left(H(u, \dd)+\bar{H}(\bar{u},\bar{\dd})\right) R^{\vec{\lambda}}_{\frac{\lambda_v+\lambda_h}{2}\,, \frac{\lambda_v+\lambda_h}{2}} (u\,,\bar{u}) \,, \lambda_v \geqslant 2 \,;\nonumber\\
& \lambda_h, \lambda_v \in 2 \mathbb{Z}_{\geqslant 0} \Bigg\}\,.
\end{flalign}
\thickmuskip=5mu 
\medmuskip=4mu
\thinmuskip=3mu

Recall that to obtain $H^{0,1,2}(\sigma_-)$ we restricted the ring $\mathfrak{R}$ (\ref{e: Howe duality adj-adj}) to the set of (HW, LW) vectors of $\mathfrak{sl}^h(2)\oplus \mathfrak{sl}^v(2)$. To get a complete set of gauge parameters, fields and field equations one should apply $E_v^N F_h^M$ operator to all listed cohomology representatives for any $N,M \in \mathbb{Z}_{\geqslant 0}$. Therefore, we observe an infinite amount of copies of the massless and partially massless fields with corresponding gauge symmetries and equations. Note that for the fixed weights $\vec{\lambda}$ operator $E_v$ can be applied infinitely many times, but $F_h^{\lambda_h + 1}$ annihilates (HW, LW) vector.

After the analysis of $H^\bullet(\sigma_-)$ in the ($adj\otimes adj$) sector, the next step is to carry out the cohomological analysis of the $(adj\otimes tw)$ and $(tw\otimes adj)$ sectors, where the $C$ fields reside. Then we will be able to study the gluing of one-form and zero-form sectors governed by the vertex on the r.h.s. of (\ref{e: adj-adj omega}).

\section{Zero-form module $(tw\otimes adj)$ and $(adj\otimes tw)$}
\label{sec: zero-form tw-adj}

As stated in the introduction, after the Klein parity truncation, all zero-form fields $C$ belong to one of the possible tensor products of adjoint and twisted-adjoint modules of the standard $4d$ HS theory. For clarity we will be examining the $(tw \otimes adj)$ module given by the covariant consistency condition (\ref{e:cov2}), as the results for other modules can be obtained by a proper redefinitions of oscillators.

Once again, we consider the ring of differential forms valued in a $Y_i$ series $\mathfrak{R} = \Lambda^\bullet(M)\otimes \mathbb{C}[[y_1,y_2,\bar y_1, \bar y_2]]$ with homogeneous elements (\ref{e: hom one-form}). In a $(tw\otimes adj)$ module, $\mathfrak{so}(3,2)$ acts by generators 
\begin{gather}
P_{\alpha \dot{\alpha}} = - i e^{\alpha \dot\alpha}(y_{1\alpha }\bar y_{1\dot \alpha } - \dd_{1\alpha } \bar \dd_{1\dot\alpha }) + e^{\alpha \dot\alpha}(y_{2\alpha }\bar\dd_{2\dot \alpha } + \bar y_{2\dot\alpha }\dd_{2\alpha })\,,\label{e: so32 adj-tw 1}\\
L_{\alpha\alpha} = \delta^{ij}y_{i\alpha}\dd_{j\alpha}\,, \quad\bar{L}_{\dot{\alpha}\dot{\alpha}} = \delta^{ij}\bar{y}_{i\dot{\alpha}}\bar{\dd}_{j\dot{\alpha}}\,. \label{e: so32 adj-tw 2}
\end{gather}
If in case of $(adj\otimes adj)$ algebra $\mathfrak{so}(3,2)$ preserves the total monomial degree, in $(tw\otimes adj)$ (and $(adj\otimes tw)$) momentum generator $P_{\alpha \dot{\alpha}}$ can infinitely increase the monomial degree, so any $\mathfrak{so}(3,2)$ submodule of $\mathfrak{R}$ is infinite-dimensional. For semisimple Lie algebras over a field of characteristic zero, Weyl’s Theorem states \cite{humphreys2012introduction} that every finite-dimensional representation is semisimple, i.e., any reducible representation can be decomposed into a direct sum of irreducible ones. But this does not hold for infinite-dimensional modules. Therefore, the module $\mathfrak{R}$ decomposes into the direct sum of infinite-dimensional indecomposable modules, but some of those modules can be reducible. If a reducible indecomposable module appears, it signals that the corresponding massless or partially massless field is an off-shell one, and is thus auxiliary, not describing physical dynamics of the system (for off-shell completion of the standard theory see \cite{Misuna:2019ijn, Misuna:2020fck}). Direct analysis of the structure of $\mathfrak{R}$ as $\mathfrak{so}(3,2)$ module is complicated by the infinite dimensionality of submodules, so we tackle this problem via a commuting $\mathfrak{sl}(2)\oplus \mathfrak{sl}(2)$ algebra.

One can check that $\mathfrak{so}(3,2)$ generators (\ref{e: so32 adj-tw 1})--(\ref{e: so32 adj-tw 2}) commute with the algebra $\mathfrak{sl}(2)\oplus \mathfrak{sl}(2)$ generated by
\begin{equation}
H_1=\bar{N}_1-N_1+N_2+\bar{N}_2+2, \quad E_1=y_2^\alpha \dd_{1 \alpha}-i \bar{y}_{1 \dot{\alpha}} \bar{y}_2^{\dot{\alpha}}, \quad F_1=y_1^\alpha \dd_{2 \alpha}-i \bar{\dd}_{1 \dot{\alpha}} \bar{\dd}_2^{\dot{\alpha}}\,,
\end{equation}
\begin{equation}
H_2=N_1-\bar{N}_1+N_2+\bar{N}_2+2, \quad E_2=\bar{y}_2^{\dot{\alpha}} \bar{\dd}_{1 \dot{\alpha}}-i y_{1 \alpha} y_2^{\alpha}, \quad F_2=\bar{y}_1^{\dot{\alpha}} \bar{\dd}_{2 \dot{\alpha}}-i \dd_{1 \alpha} \dd_2^\alpha\,.
\end{equation}
We see that $\mathfrak{sl}(2)\oplus \mathfrak{sl}(2)$ algebra has only infinite-dimensional submodules in $\mathfrak{R}$ as raising operators $E_i$ add traces to any $Y_i$-monomial and the terms with differentiation do not act on these traces. Thus, with respect to $\mathfrak{sl}(2)\oplus \mathfrak{sl}(2)$ the ring $\mathfrak{R}$ also decomposes into the direct sum of indecomposable (possibly reducible) modules. Therefore, the simple Howe decomposition does not work, as in general the Howe-duality demands the decomposition: 
\begin{equation}\label{e: duality tw-adj}
    \mathfrak{R} \simeq \bigoplus_i N_i \otimes M_i\,,
\end{equation}
where $N_i$ is irreducible $\mathfrak{sl}(2)\oplus \mathfrak{sl}(2)$ module and $M_i$ is irreducible $\mathfrak{so}(3,2)$ module. The commutation of $\mathfrak{sl}(2)\oplus \mathfrak{sl}(2)$ with $\mathfrak{so}(3,2)$ only guarantees that 
\begin{equation}
    \rho(\mathfrak{so}(3,2)) \subset \text{End}_{\mathfrak{sl}(2)\oplus \mathfrak{sl}(2)}(\mathfrak{R})\,, \quad \rho(\mathfrak{sl}(2)\oplus \mathfrak{sl}(2)) \subset \text{End}_{\mathfrak{so}(3,2)}(\mathfrak{R})\,,
\end{equation}
i.e., the image of the representation map $\rho$ of one algebra acts on the multiplicity spaces of modules of another. However, the fact that the action of $\mathfrak{sl}(2)\oplus \mathfrak{sl}(2)$ does not change the type of $\mathfrak{so}(3,2)$ module will be sufficient to determine the $\sigma_-$ cohomology. Indeed, an element of $\mathfrak{R}$ can be represented as a linear combination
\thickmuskip=2.5mu 
\medmuskip=2mu
\thinmuskip=1.5mu
\begin{equation}
\Psi\left(Y_1, Y_2\right)=\sum_{\substack{a, b, c \\ \bar{a}, \bar{b}, \bar{c}}} \frac{1}{(a+b)!(\bar{a}+\bar{b})!} \zeta^c \bar{\zeta}\,{}^{\bar{c}}\left(\dd y_1\right)^a\left(\dd y_2\right)^b\left(\bar{\dd} \bar{y}_1\right)^{\bar{a}}\left(\bar{\dd} \bar{y}_2\right)^{\bar{b}} \Psi_{a, b, c}^{\bar{a}, \bar{b}, \bar{c}} \psi_{a+b, \bar{a}+\bar{b}}(u, \bar{u})\,,
\end{equation}
\thickmuskip=5mu 
\medmuskip=4mu
\thinmuskip=3mu
where we introduce auxiliary spinor variables $(u,\bar{u})$ and use notation (\ref{e: traces}), (\ref{e: u notation}). The function $\psi_{a+b, \bar{a}+\bar{b}}(u, \bar{u})$ encodes spin-tensor with $(a+b)$ undotted and $(\bar{a}+\bar{b})$ dotted indices, while the $\Psi_{a, b, c}^{\bar{a}, \bar{b}, \bar{c}}$ is a complex number. It is a standard representation of a rank-two field as a sum of rank-one $\mathfrak{so}(3,1)$ fields
\begin{equation}
\zeta^c \bar{\zeta}\,{}^{\bar{c}}\left(\dd y_1\right)^a\left(\dd y_2\right)^b\left(\bar{\dd} \bar{y}_1\right)^{\bar{a}}\left(\bar{\dd} \bar{y}_2\right)^{\bar{b}} \psi_{a+b, \bar{a}+\bar{b}}(u, \bar{u})\,,
\end{equation}
that form a countable basis in a space of $Y_i$ series. This basis of $\mathfrak{R}$ respects the $\mathfrak{sl}(2)\oplus \mathfrak{sl}(2)$ weight structure
\begin{align}
     H_1 \zeta^c \bar{\zeta}\,{}^{\bar{c}}\left(\dd y_1\right)^a\left(\dd y_2\right)^b\left(\bar{\dd} \bar{y}_1\right)^{\bar{a}}\left(\bar{\dd} \bar{y}_2\right)^{\bar{b}} &= \lambda_1 \zeta^c \bar{\zeta}\,{}^{\bar{c}}\left(\dd y_1\right)^a\left(\dd y_2\right)^b\left(\bar{\dd} \bar{y}_1\right)^{\bar{a}}\left(\bar{\dd} \bar{y}_2\right)^{\bar{b}}\,,\\
     H_2 \zeta^c \bar{\zeta}\,{}^{\bar{c}}\left(\dd y_1\right)^a\left(\dd y_2\right)^b\left(\bar{\dd} \bar{y}_1\right)^{\bar{a}}\left(\bar{\dd} \bar{y}_2\right)^{\bar{b}} &= \lambda_2\zeta^c \bar{\zeta}\,{}^{\bar{c}}\left(\dd y_1\right)^a\left(\dd y_2\right)^b\left(\bar{\dd} \bar{y}_1\right)^{\bar{a}}\left(\bar{\dd} \bar{y}_2\right)^{\bar{b}}\,,
\end{align}
where the weights are
\begin{equation}
    \lambda_1 = (2\bar{c}+\bar{a} + \bar{b} + b -a +2)\,, \quad    \lambda_2 = (2c+a+b+\bar{b}-\bar{a}+2)\,.
\end{equation}
Note that
\begin{equation}\label{e: tw-adj weights sum}
    \lambda_+=\lambda_1 + \lambda_2 = 2(b+\bar{b}+c+\bar{c}+2) > 0\,, \quad \lambda_-=\lambda_2 - \lambda_1 = 2(a + c - \bar{a} - \bar{c})\,,
\end{equation}
thus both weights cannot be simultaneously non-positive. 

One can see that $\mathfrak{R}$ is locally nilpotent with respect to $F_i$. This means that for any basis vector $v$ there are $n_i \in \mathbb{Z}_{\geq0}: F_i^{n_i}  v = 0$. If $\mathfrak{R}$ decomposes into the direct sum of $\mathfrak{sl}(2)\oplus \mathfrak{sl}(2)$ irreps, then one should be able to represent any basis vector of $\mathfrak{R}$ as a finite linear combination of $\mathfrak{sl}(2)\oplus \mathfrak{sl}(2)$ (LW, LW) vectors and their descendants. However, local nilpotency cannot guarantee semisimplicity of $\mathfrak{R}$ and we will see that there are submodules of $\mathfrak{R}$ in which any element can be mapped to the (LW, LW) vector by lowering operators $F_i$, but not all elements are descendants of the (LW, LW) vectors.

We proceed with a search of (LW, LW) functions. The lowest-weight conditions are
\begin{align}
&F_1 \Psi\left(Y_1, Y_2\right)=\sum_{\substack{a, b, c \\ \bar{a}, \bar{b}, \bar{c}}} \frac{1}{(a+b)!(\bar{a}+\bar{b})!}\left\{b \zeta^{c} \bar{\zeta}\,{}^{\bar{c}}\left(\dd y_1\right)^{a+1}\left(\dd y_2\right)^{b-1}\left(\bar{\dd} \bar{y}_1\right)^{\bar{a}}\left(\bar{\dd} \bar{y}_2\right)^{\bar{b}}- \right.\nonumber\\
&\left.-i \bar{c}(\bar{a}+\bar{b}+\bar{c}+1) \zeta^c \bar{\zeta}\,{}^{\bar{c}-1}\left(\dd y_1\right)^a\left(\dd y_2\right)^b\left(\bar{\dd} \bar{y}_1\right)^{\bar{a}}\left(\bar{\dd} \bar{y}_2\right)^{\bar{b}} \right\} \Psi_{a, b, c}^{\bar{a}, \bar{b}, \bar{c}} \psi_{a+b, \bar{a}+\bar{b}}(u, \bar{u}) = 0\,,
\end{align}
\begin{align}
&F_2 \Psi\left(y_1, y_2\right)=\sum_{\substack{a, b, c \\ \bar{a}, \bar{b}, \bar{c}}} \frac{1}{(a+b)!(\bar{a}+\bar{b})!}\left\{\bar{b} \zeta^c \bar{\zeta}\,{}^{\bar{c}}\left(\dd y_1\right)^a\left(\dd y_2\right)^b\left(\bar{\dd} \bar{y}_1\right)^{\bar{a}+1}\left(\bar{\dd} \bar{y}_2\right)^{\bar{b}-1}-\right.\nonumber\\
&\left.-i c(a+b+ c+1) \zeta^{c-1} \bar{\zeta}\,{}^{\bar{c}}\left(\dd y_1\right)^a\left(\dd y_2\right)^b \left(\dd \bar{y}_1\right)^{\bar{a}}\left(\bar{\dd} \bar{y}_2\right)^{\bar{b}}\right\} \Psi_{a, b, c}^{\bar{a}, \bar{b}, \bar{c}} \psi_{a+b, \bar{a}+\bar{b}}(u, \bar{u}) = 0\,.
\end{align}

The solution is particularly simple for $\lambda_+ = 4$, the lowest-weight vectors take the form
\begin{equation}\label{e: tw-adj HW1}
\Psi^{\vec{\lambda}}{}^{\bar{a}, 0}_{a, 0}(Y_1,Y_2)=\frac{1}{a!\bar{a}!}\left(\dd y_1\right)^a\left(\bar{\dd} \bar{y}_1\right)^{\bar{a}}\\
\Psi_{a, 0,0}^{\bar{a}, 0,0} \psi^{\vec{\lambda}}_{a, \bar{a}}(u, \bar{u})\,, \quad \lambda_+ = 4\,, \quad \lambda_- = 2(a - \bar{a})\,.
\end{equation}
For a general weights $\lambda_+ > 4$ the (LW, LW) vectors are
\begin{align}
&\Psi^{\vec{\lambda}}{}^{\bar{a}, \bar{b}}_{a, b}(Y_1,Y_2)=\sum_{\substack{k=0,\dots, b \\ l=0,\dots,\bar{b}}} \frac{b!\bar{b}!(a+b+1)(\bar{a}+\bar{b}+1)}{i^{k+l} k! l! (b-k)!(\bar{b}-l)! (a+b+l+1)!(\bar{a}+\bar{b}+k+1)!} \times \nonumber\\
&\times\Psi_{a, b, 0}^{\bar{a}, \bar{b}, 0} \zeta^l\bar{\zeta}{}^k\left(\dd y_1\right)^{a+k}(\bar{\dd} \bar{y}_1)^{\bar{a}+l}\left(\dd y_2\right)^{b-k}(\bar{\dd} \bar{y}_2)^{\bar{b}-l} \psi^{\vec{\lambda}}_{a+b, \bar{a}+\bar{b}}(u, \bar{u})\,, \label{e: tw-adj HW2} \\ 
&\lambda_+= 2(b+\bar{b}+2)\,, \; \lambda_- = 2(a  - \bar{a})\,. \nonumber
\end{align}
Each of these vectors generates a lowest-weight $\mathfrak{sl}(2)\oplus \mathfrak{sl}(2)$ Verma module $V(\lambda_1, \lambda_2)$. In case of $\mathfrak{sl}(2)$ a lowest-weight Verma module $V(\lambda)$ has a singular submodule if $\lambda \leq 0$. Therefore, there are three cases of $V(\lambda_1, \lambda_2)$ inner structure:
\begin{itemize}
    \item $\lambda_i > 0$: irreducible module;
    \item $\lambda_1 \leq 0\,, \lambda_2 > 0$ or the opposite: in the first case $V(\lambda_1, \lambda_2)$ has one singular submodule $V(-\lambda_1+2, \lambda_2)$ generated from $E_1^{-\lambda_1+1}\Psi^{\vec{\lambda}}$. The second case is obtained by interchanging algebra label $1$ with $2$;
    \item $\lambda_i \leq 0$: $V(\lambda_1, \lambda_2)$ has three singular submodules $V(-\lambda_1+2, \lambda_2)$, $V(\lambda_1, -\lambda_2+2)$ and $V(-\lambda_1+2, -\lambda_2 + 2)$. Singular submodules intersect with each other.
\end{itemize}
As we noted in (\ref{e: tw-adj weights sum}), at least one of the weights is always non-negative, so the last type of $V(\lambda_1, \lambda_2)$ do not appear in $\mathfrak{R}$.

Having found the (LW, LW) vectors we can provide an example of an element which is not representable as a sum of (LW, LW) vectors and their descendants. The element 
$\Phi = \left(\bar{\dd} \bar{y}_1\right)^n \left(\bar{\dd} \bar{y}_2\right)^m \psi_{0, n+m}(u, \bar{u})$ is weighted with respect to the $\mathfrak{sl}(2) \oplus \mathfrak{sl}(2)$ algebra 
\begin{equation}
    H_1 \Phi = (n+m+2)\Phi=h_1 \Phi\,, \quad H_2 \Phi = (m-n+2)\Phi=h_2 \Phi\,,
\end{equation}
and satisfies local nilpotency as $F_1 \Phi =  F_2^{m+1} \Phi = 0$. Suppose it could be expressed as a sum of (LW, LW) vectors and their descendants. The search is confined by the nilpotency conditions and allowed weights to 
\begin{equation}
    \Phi = \sum_{(a,b,\bar{a},\bar{b})} \Psi^{(h_1,h_2)}{}^{\bar{a},\bar{b}}_{a,b} +  \dots+ \sum_{(a,b,\bar{a},\bar{b})} E_2^k \Psi^{(h_1,h_2-2k)}{}^{\bar{a},\bar{b}}_{a,b}+\dots +\sum_{(a,b,\bar{a},\bar{b})} E_2^m \Psi^{(h_1,h_2-2m)}{}^{\bar{a},\bar{b}}_{a,b}\,.
\end{equation}

By applying powers of $F_2$ to the decomposition of $\Phi$ and using the LW condition, it is easy to arrive at a contradiction. The source of the contradiction are terms $F_2^{l}E_2^k \Psi^{(h_1,h_2-2k)}{}^{\bar{a},\bar{b}}_{a,b}$: $h_2 \leq 2k$ and $k \leq h_2 - 1$, which for $l=h_2-k-1$ result in a singular submodule growing from a vector $E_2^{2k-h_2+1} \Psi^{(h_1,h_2-2k)}{}^{\bar{a},\bar{b}}_{a,b}$ . Application of $F_2$ annihilates singular vectors creating a \enquote{gap} in a system of equations on the LW vectors and resulting in inability to decompose $\Phi$. For example, consider $\Phi = \left(\bar{\dd} \bar{y}_1\right)^2\left(\bar{\dd} \bar{y}_2\right)^2 \psi_{0, 4}(u, \bar{u})$. One can check that
\begin{align}
    H_1 \Phi &= 6 \Phi\,, &\quad H_2 \Phi &= 2\Phi\,,\\
    F_1 \Phi &= 0\,, &\quad F_2^3 \Phi &= 0\,.
\end{align}
Therefore, the decomposition is confined to 
\begin{equation}
    \Phi = \left(\bar{\dd} \bar{y}_1\right)^2\left(\bar{\dd} \bar{y}_2\right)^2 \psi_{0, 4}(u, \bar{u})=\sum_{(a,b,\bar{a},\bar{b})} \Psi^{(6,2)}{}^{\bar{a},\bar{b}}_{a,b} + \sum_{(a,b,\bar{a},\bar{b})} E_2 \Psi^{(6,0)}{}^{\bar{a},\bar{b}}_{a,b} + \sum_{(a,b,\bar{a},\bar{b})} E_2^2 \Psi^{(6,-2)}{}^{\bar{a},\bar{b}}_{a,b}\,.
\end{equation}
Application of $F_2^2$ to both sides results in
\begin{equation}
    \sum_{(a,b,\bar{a},\bar{b})} \Psi^{(6,-2)}{}^{\bar{a},\bar{b}}_{a,b} = \frac{1}{2}\left(\bar{\dd} \bar{y}_2\right)^4 \psi_{0, 4}(u, \bar{u})\,.
\end{equation}
Substituting it into the decomposition of $\Phi$ and applying $F_2$, we arrive at a contradiction
\begin{equation}
    0 = \left(2\left(\bar{\dd} \bar{y}_1\right)\left(\bar{\dd} \bar{y}_2\right)^3  + i \zeta \left(\bar{\dd} \bar{y}_2\right)^4\right) \psi_{0, 4}(u, \bar{u})\,
\end{equation}
due to the singular vector $\sum_{(a,b,\bar{a},\bar{b})} E_2 \Psi^{(6,0)}{}^{\bar{a},\bar{b}}_{a,b}$.

The presence of vectors that are not descendants of the LW ones, but are nevertheless nilpotent with respect to the lowering generators $F_i$, suggests the appearance of non-split extension modules, i.e., there are modules $E \subset\mathfrak{R}$:
\begin{equation}
0 \rightarrow V\left(\lambda_1, \lambda_2\right) \rightarrow E \rightarrow V\left(\mu_1, \mu_2\right) \rightarrow 0
\end{equation}
and $E$ does not split
\begin{equation}
E \not\simeq V\left(\lambda_1, \lambda_2\right) \oplus V\left(\mu_1, \mu_2\right)\,.
\end{equation}
In appendix \ref{sec: nonsplit extension} we show that 
\begin{align}
&E \text{ is a non-split extension of $V\left(\mu_1, \mu_2\right)$ by $V\left(\lambda_1, \lambda_2\right)$}\,,\;  \text{if}\,\nonumber\\
&\text{and only if}\;\left(\mu_1, \mu_2\right)=\left(2-\lambda_1, \lambda_2\right)\text { or } \left(\mu_1, \mu_2\right)=\left(\lambda_1,2-\lambda_2\right)\,,\label{NonSplit}
\end{align}
where $\lambda_1 \leq 0$ in the first case and $\lambda_2 \leq 0$ in the second case (if both weights $\lambda_i \leq 0$, two additional cases of non-split extensions occur, however such lowest weights do not appear in $\mathfrak{R}$).
This means that singular vectors are responsible for non-split extensions of lowest-weight Verma modules in the decomposition of $\mathfrak{R}$. Vectors that become singular after a quotient by a singular submodule are sometimes called subsingular in the literature. In the context of unfolded equations subsingular and subsubsingular modules were discussed in \cite{Shaynkman:2004vu} in the context of conformal algebra $\mathfrak{o}(M, 2)$.

For example, consider the short exact sequence
\begin{equation}
0 \rightarrow V\left(6, 0\right) \overset{i}{\rightarrow} E \overset{p}{\rightarrow} V\left(6, 2\right) \rightarrow 0\,.
\end{equation}
Module $E$ is generated from two vectors $v_0$ and $\hat{w}_2$:
\begin{align}
    v_0 &= \left[\left(\bar{\dd} \bar{y}_1\right)^3\left(\bar{\dd} \bar{y}_2\right)-\frac{i}{2} e\left(\bar{\dd} \bar{y}_1\right)^4\right]\psi^{\vec{\lambda}}_{0, 4}(u, \bar{u})\,,
\\
    \hat{w}_2 &= \left[\left(\bar{\dd} \bar{y}_1\right)^2\left(\bar{\dd} \bar{y}_2\right)^2-\frac{i}{2} e(\bar{\dd} \bar{y}_1)^3\left(\bar{\dd} \bar{y}_2\right)\right]\psi^{\vec{\lambda}}_{0, 4}(u, \bar{u})\,.
\end{align}

For which
\begin{alignat}{4}
    F_i v_0 &= 0\,, &\quad H_1 v_0 &= 6 v_0\,, &\quad H_2 v_0 &= 0\,;\\
     F_1 \hat{w}_2 &= 0\,, &\, F_2 \hat{w}_2 &= v_0\,, &\, H_1 \hat{w}_2 &= 6 \hat{w}_2\,, \quad H_2 \hat{w}_2 = 2 \hat{w}_2\,.
\end{alignat}
Vector $v_0$ and its descendants span submodule $V(6,0) \subset E$ and vector $\hat{w}_2 = p^{-1}(w_2)$ is a pre-image of a projection map $p$, where $w_2$ is a lowest-weight vector of $V\left(6, 2\right)$. Note that $\hat{w}_2 \neq E_2 v_0$ as
\begin{equation}
    E_2 v_0 = \left[\left(\bar{\dd} \bar{y}_1\right)^2\left(\bar{\dd} \bar{y}_2\right)^2-ie(\bar{\dd} \bar{y}_1)^3\left(\bar{\dd} \bar{y}_2\right) - \frac{1}{6} e^2\left(\bar{\dd} \bar{y}_1\right)^4\right]\psi^{\vec{\lambda}}_{0, 4}(u, \bar{u})\,.
\end{equation}
One can check that
\begin{equation}
    F_2  E_2^n \hat{w}_2=E_2^n v_0-n(n+1) E_2^{n-1} \hat{w}_2\,.
\end{equation}
Therefore, a quotient of $E$ by $V(6,0)$ is isomorphic to $V(6,2)$, i.e., $E$ is a genuine non-split extension of $V(6,2)$ by $V(6,0)$.

Thus, we assume that as a $\mathfrak{sl}(2)\oplus \mathfrak{sl}(2)$ module $\mathfrak{R}$ decomposes into the sum of Verma modules and non-split extension modules
\begin{equation}
    \mathfrak{R} \simeq \bigoplus_{\lambda_1, \lambda_2} V(\lambda_1\,, \lambda_2) \bigoplus_{E\in\text{non-split}} E\,.
\end{equation}

Consequently, the multiplicity spaces of $\mathfrak{so}(3,2)$ modules are indecomposable locally nilpotent $\mathfrak{sl}(2)\oplus \mathfrak{sl}(2)$ modules, so the restriction to the class of (LW, LW) $\mathfrak{sl}(2)\oplus \mathfrak{sl}(2)$ vectors (which is always possible due to the local nilpotence) covers all types of indecomposable (possibly reducible) $\mathfrak{so}(3,2)$ modules present in $\mathfrak{R}$. In the following $\sigma_-$ cohomology analysis we will limit ourselves to (LW, LW) vectors (\ref{e: tw-adj HW1}), (\ref{e: tw-adj HW2}). 

As in section \ref{sec: one-form adj-adj} we restrict the covariant constancy equation (\ref{e:cov2}) to (LW, LW) vectors (\ref{e: tw-adj HW2}), extract equations on the corresponding spin-tensors, and then introduce an auxiliary generating function for a proper definition of $\sigma_-$.

The Lorentz derivative $D_L$ has an expected action on the spin-tensor
\thickmuskip=2.5mu 
\medmuskip=2mu
\thinmuskip=1.5mu
\begin{align}
&D_L \Psi^{\vec{\lambda}}{}^{\bar{a}, \bar{b}}_{a, b}(Y_1,Y_2)=\sum_{\substack{k=0,\dots, b \\ l=0,\dots,\bar{b}}} \frac{b!\bar{b}!(a+b+1)(\bar{a}+\bar{b}+1)}{i^{k+l} k! l! (b-k)!(\bar{b}-l)! (a+b+l+1)!(\bar{a}+\bar{b}+k+1)!}\nonumber \\
&\Psi_{a, b, 0}^{\bar{a}, \bar{b}, 0}\zeta^l\bar{\zeta}{}^k\left(\dd y_1\right)^{a+k}(\bar{\dd} \bar{y}_1)^{\bar{a}+l}\left(\dd y_2\right)^{b-k}(\bar{\dd} \bar{y}_2)^{\bar{b}-l}\left\{\d_x+\omega(u, \dd)+\bar{\omega}(\bar{u}, \bar{\dd})\right\} \psi^{\vec{\lambda}}_{a+b, \bar{a}+\bar{b}}(u, \bar{u}) \,.
\end{align}
\thickmuskip=5mu 
\medmuskip=4mu
\thinmuskip=3mu
The action of the momentum operator $P$ results in a cumbersome expression presented in appendix \ref{sec: momentum action}.

All these facts together allow us to transfer the action of the covariant derivative $\mathcal{D}$ to the spin-tensor level (with the term $b!\bar{b}!(a+b+1)(\bar{a}+\bar{b}+1)$ absorbed into $\psi^{\vec{\lambda}}_{a+b, \bar{a}+\bar{b}}$)
\begin{align}
&D_L \psi^{\vec{\lambda}}_{a+b, \bar{a}+\bar{b}}(u, \bar{u})-i \frac{(a+b+\bar{b}+1)(\bar{a}+b+\bar{b}+1)}{(a+b)(\bar{a}+\bar{b})} e(u, \bar{u}) \psi^{\vec{\lambda}}_{a+b-1, \bar{a}+\bar{b}-1}(u, \bar{u}) +\nonumber\\
&+ i \frac{(a+1)(\bar{a}+1)}{(a+b+2)(\bar{a}+\bar{b}+2)} e (\dd, \bar{\dd}) \psi^{\vec{\lambda}}_{a+b+1, \bar{a}+\bar{b}+1} (u, \bar{u}) +\nonumber\\
&+ \frac{b(b+a-\bar{a})}{(a+b)(\bar{a}+\bar{b}+2)} e(u, \bar{\dd}) \psi^{\vec{\lambda}}_{a + b-1, \bar{a} + \bar{b}+1}(u, \bar{u}) + \nonumber\\
&+\frac{\bar{b}(\bar{b}+\bar{a}-a)}{(a+b+2)(\bar{a}+\bar{b})} e (\dd, \bar{u}) \psi^{\vec{\lambda}}_{a + b+1,\bar{a} + \bar{b} - 1}(u, \bar{u}) = 0\,, \label{e: tw-adj spin-tensor equation} \\
& \text{where} \quad\lambda_+= 2(b+\bar{b}+2)\,, \quad \lambda_- = 2(a  - \bar{a})\,. \nonumber
\end{align}

These different spin-tensors can be concisely written using a generating function
\begin{equation}\label{tw adj generating function}
\Psi\left(u, \bar{u}, v, \bar{v}\right)=\sum_{a, b, \bar{a}, \bar{b}} \psi^{\vec{\lambda}}_{a+b, \bar{a}+\bar{b}}(u, \bar{u}) v^a \bar{v}^{\bar{a}}\,,
\end{equation}
where variables $v\,, \bar{v}$ are introduced to keep the information about the number of $y_1$ and $\bar{y}_1$ oscillators in a traceless term of (\ref{e: tw-adj HW2}), i.e., values of $a, \bar{a}$ in the corresponding vector $\Psi^{\vec{\lambda}}{}^{\bar{a}, \bar{b}}_{a, b}$.

With the help of a generating function (\ref{tw adj generating function}) the equation for spin-tensors turns into
\begin{align}
&\left(D_L-i \frac{\left(\hat{N}_u+\hat{\bar{N}}_u-\hat{\bar{N}}_v+1\right)\left(\hat{N}_u+\hat{\bar{N}}_u-\hat{N}_v+1\right)}{\hat{N}_u \hat{\bar{N}}_u} e(u, \bar{u})\, v\, \bar{v} +\right.\nonumber\\ 
&+\frac{i}{\left(\hat{N}_u+2\right)\left(\hat{\bar{N}}_u+2\right)} e(\dd, \bar{\dd})\, \dd_v \,\bar{\dd}_v
+ \frac{\left(\hat{N}_u-\hat{N}_v\right)\left(\hat{N}_u-\hat{\bar{N}}_v\right)}{\hat{N}_u\left(\hat{\bar{N}}_u+2\right)} e(u, \bar{\dd})+\nonumber\\
& \left.+\frac{\left(\hat{\bar{N}}_u-\hat{\bar{N}}_v\right)\left(\hat{\bar{N}}_u-\hat{N}_v\right)}{\left(\hat{N}_u+2\right) \hat{\bar{N}}_u} e(\dd, \bar{u})\right) \Psi\left(u, \bar{u}, v, \bar{v}\right)=0\,,
\end{align}
where
\begin{alignat}{2}
\hat{N}_u&=u^\alpha \frac{\dd}{\dd u^\alpha}\,, &\quad \hat{\bar{N}}_u&=\bar{u}^{\dot{\alpha}} \frac{\dd}{\dd \bar{u}^{\dot{\alpha}}}\,, \\
 \hat{N}_v&=v \dd_v\,, &\quad \hat{\bar{N}}_v&=\bar{v} \bar{\dd}_v\,.
\end{alignat}

With the grading $G = |\hat{N}_u - \hat{\bar{N}}_u| + 2\min(\hat{N}_v\,,\hat{\bar{N}}_v)$, the properly defined $\sigma_-$ is
\begin{align}
&\sigma_-=\frac{i}{\left(\hat{N}_u+2\right)\left(\hat{N}_u+2\right)} e(\dd, \bar{\dd}) \,\dd_v \, \bar{\dd}_v+\nonumber\\
&+\frac{\left(\hat{N}_u-\hat{N}_v\right)\left(\hat{N}_u-\hat{\bar{N}}_v\right)}{\hat{N}_u\left(\hat{\bar{N}}_u+2\right)} e(u, \bar{\dd}) \theta\left(\hat{\bar{N}}_u-\hat{N}_u-\hat{\bar{N}}_v+\hat{N}_v\right)+\nonumber\\
&+\frac{\left(\hat{\bar{N}}_u-\hat{\bar{N}}_v\right)\left(\hat{\bar{N}}_u-\hat{N}_v\right)}{\left(\hat{N}_u+2\right) \hat{\bar{N}}_u} e(\dd, \bar{u}) \theta\left(\hat{N}_u-\hat{\bar{N}}_u-\hat{N}_v+\hat{\bar{N}}_v\right)\,.\label{e: sigma tw-adj even}
\end{align}

Similarly to the case of the $(adj\otimes adj)$ module, $\sigma_-$ acts on $\mathfrak{so}(3,1)$ two-row Young diagrams by cutting off boxes from the first and the second rows. However, in the case of the $(tw\otimes adj)$ module with fixed weights $\vec{\lambda}$ we deal with an infinite set of $\mathfrak{so}(3,1)$ Young diagrams, and not with a finite one, as in the module $(adj\otimes adj)$, since only a difference $\lambda_- = 2(a  - \bar{a})$ is fixed. The first term in $\sigma_-$ maps $\Psi^{\vec{\lambda}}{}^{\bar{a}, \bar{b}}_{a, b}$ to $\Psi^{\vec{\lambda}}{}^{\bar{a}-1, \bar{b}}_{a-1, b}$ and the remaining ones map $\Psi^{\vec{\lambda}}{}^{\bar{a}, \bar{b}}_{a, b}$ to $\Psi^{\vec{\lambda}}{}^{\bar{a}, \bar{b}\pm 1}_{a, b\mp 1}$. As a result, the grading is reduced by $2$.

It is important to note that (\ref{e: sigma tw-adj even}) in that form is defined for $|b-\bar{b}| \in 2 \mathbb{Z}$ as can be seen from $\theta$-functions.
\begin{equation}
    \left(\hat{N}_u-\hat{\bar{N}}_u-\hat{N}_v+\hat{\bar{N}}_v\right)\psi^{\vec{\lambda}}_{a+b, \bar{a}+\bar{b}}(u, \bar{u}) v^a \bar{v}^{\bar{a}} = (b-\bar{b}) \psi^{\vec{\lambda}}_{a+b, \bar{a}+\bar{b}}(u, \bar{u}) v^a \bar{v}^{\bar{a}}\,.
\end{equation}

However, zero-form fields with $|b-\bar{b}| \not\in 2 \mathbb{Z}$ also appear in the context of bosonic massless and partially massless fields of section \ref{sec: one-form adj-adj}. In that case the \enquote{bosonic} $\sigma_-$ (\ref{e: sigma tw-adj even}) and the set of $C$ fields split into the regions $b-\bar{b} \geq 1$ and $\bar{b} - b \geq 1$ like in a case of fermions in standard HS theory \cite{Bychkov:2021zvd} (as $\sigma_-$ defined in (\ref{e: sigma tw-adj even}) maps vectors with $b-\bar{b} = \pm1$ to $b-\bar{b} = \mp1$, it does not reduce grading $G$ and can not be seen as a proper $\sigma_-$ operator). This means that the correct definition would be
\thickmuskip=1.5mu
\medmuskip=1.5mu
\thinmuskip=1mu
{\small
\begin{align}
&\sigma_-\bigg|_{b-\bar{b} \geq 1}=\frac{i}{\left(\hat{N}_u+2\right)\left(\hat{N}_u+2\right)} e(\dd, \bar{\dd}) \,\dd_v \, \bar{\dd}_v+\frac{\left(\hat{\bar{N}}_u-\hat{\bar{N}}_v\right)\left(\hat{\bar{N}}_u-\hat{N}_v\right)}{\left(\hat{N}_u+2\right) \hat{\bar{N}}_u} e(\dd, \bar{u}) \theta\left(\hat{N}_u-\hat{\bar{N}}_u-\hat{N}_v+\hat{\bar{N}}_v - 1\right)\,,\nonumber\\
&\sigma_-\bigg|_{b-\bar{b} \leq 1}=\frac{i}{\left(\hat{N}_u+2\right)\left(\hat{N}_u+2\right)} e(\dd, \bar{\dd}) \,\dd_v \, \bar{\dd}_v+\frac{\left(\hat{N}_u-\hat{N}_v\right)\left(\hat{N}_u-\hat{\bar{N}}_v\right)}{\hat{N}_u\left(\hat{\bar{N}}_u+2\right)} e(u, \bar{\dd}) \theta\left(\hat{\bar{N}}_u-\hat{N}_u-\hat{\bar{N}}_v+\hat{N}_v - 1\right)\,.\label{e: sigma tw-adj odd}
\end{align}
}
\thickmuskip=5mu
\medmuskip=4mu
\thinmuskip=3mu

In regions $|b - \bar{b}| \geq 3$ the theta functions are constant and the cohomology classes of (\ref{e: sigma tw-adj even}) coincide with those of (\ref{e: sigma tw-adj odd}). Therefore, we will have to consider separately only the cases of $|b-\bar{b}| = 1$, where the definition of $\sigma_-$ is different.

\subsection{Computation of $H^\bullet(\sigma_-)$}

Now we proceed with derivation of lower cohomology groups $H^{0,1}(\sigma_-)$ in the $(tw \otimes adj)$ sector. As zero-form fields $C$ are subject to the covariant constancy $(tw \otimes adj)$ equation, $H^0(\sigma_-)$ lists the primary fields and $H^{1}(\sigma_-)$ gives invariant differential operators on the fields. We follow the calculation procedure used in section \ref{sec: cohomology computation adj-adj}. Unfortunately, due to the presence of $\mathfrak{sl}(2)\oplus \mathfrak{sl}(2)$ reducible indecomposable modules we can not lift cohomology representatives in terms of auxiliary variables $(u, \bar{u}, v, \bar{v})$ to the full list of $Y_i$-functions. The lift can be performed in a way similar to the one used in section \ref{sec: cohomology computation adj-adj} if we restrict ourself to the class of (LW, LW) vector and their descendants. The cohomology groups $H^{0,1}(\sigma_-)$ in the $(adj \otimes tw)$ sector result from the swap of parameters $(a,b)$ and $(\bar{a},\bar{b})$ in all formulas below.

\subsubsection{$H^0(\sigma_-)$}

A general zero-form is
\begin{equation}
C^{\vec{\lambda}}\left(u, \bar{u}, v, \bar{v}\right)=\sum_{\substack{\lambda_+= 2(b+\bar{b}+2)\\ \lambda_- = 2(a  - \bar{a})}} C^{\vec{\lambda}}_{a+b, \bar{a}+\bar{b}}(u, \bar{u}) \,v^a \, \bar{v}^{\bar{a}}\,.
\end{equation}

Then application of $\sigma_-$ (\ref{e: sigma tw-adj even}) yields
\begin{align}
& \sigma_{-} C^{\vec{\lambda}}\left(u, \bar{u}, v, \bar{v}\right)=\sum_{\substack{\lambda_+= 2(b+\bar{b}+2)\\ \lambda_- = 2(a  - \bar{a})}}\Bigg\{\frac{i a \bar{a}}{(a+b+1)(\bar{a}+\bar{a}+1)} e(\dd, \bar{\dd}) C^{\vec{\lambda}}_{a+b, \bar{a}+\bar{b}}(u, \bar{u}) v^{a-1} \bar{v}^{\bar{a}-1}+\nonumber\\
&+ \theta(\bar{b}-b) \frac{(b+1)(a+b-\bar{a}+1)}{(a+b+1)(\bar{a}+\bar{b}+1)} e(u, \bar{\dd}) C^{\vec{\lambda}}_{a+b, \bar{a}+\bar{b}}(u, \bar{u}) v^a  \bar{v}^{\bar{a}}+\nonumber\\
&  + \theta(b-\bar{b}) \frac{(\bar{b}+1)(\bar{a}+\bar{b}-a+1)}{(a+b+1)(\bar{a}+\bar{b}+1)} e(\dd, \bar{u})  C^{\vec{\lambda}}_{a+b, \bar{a}+\bar{b}}(u, \bar{u}) v^a \bar{v}^{\bar{a}}\Bigg\}=0 \,.\label{e: sigma tw-adj zero-form action}
\end{align}

Recall that Lorenz-irreducible one-forms (\ref{e: irreducible one-forms}) are independent, so each term in the sum should vanish independently. Since $\text{Im}(\sigma_-) \vert_{\deg \Lambda = 0} = \{\emptyset\}$, the cohomology group is
\thickmuskip=2mu
\begin{align}
&H^0(\sigma_-)=\ker(\sigma_-) \vert_{\deg \Lambda = 0} \supset \Bigg\{ \lambda_{-},\lambda_{+} \in 2\mathbb{Z}:\quad \lambda_+= 2(b+\bar{b}+2)\geq 4\,, \quad \lambda_- = 2(a  - \bar{a})\nonumber\\
&\bullet \quad C^{\vec{\lambda}}_{\frac{|\lambda_-|}{2}+\frac{\lambda_+}{4} - 1,  \;\frac{\lambda_+}{4} - 1}(u, \bar{u}) v^{\frac{|\lambda_-|}{2}}\,, \lambda_-\geq 0 \, \oplus \, C^{\vec{\lambda}}_{\frac{\lambda_+}{4} - 1, \;\frac{|\lambda_-|}{2}+\frac{\lambda_{+}}{4} - 1}(u, \bar{u}) \bar{v}^{\frac{|\lambda_-|}{2}}\,, \lambda_-\leq 0\,: \lambda_{+} \in 4\mathbb{Z} \,; \nonumber \\
&\bullet \quad C^{\vec{\lambda}}_{\frac{|\lambda_-|}{2}-1, \frac{\lambda_+}{2}-1}(u, \bar{u}) \bar{v}^{\frac{|\lambda_-|}{2}}\,, \lambda_-\leq -2 \, \oplus \, C^{\vec{\lambda}}_{\frac{\lambda_+}{2}-1, \frac{|\lambda_-|}{2}-1}(u, \bar{u}) v^{\frac{|\lambda_-|}{2}}\,, \lambda_-\geq 2\,: \lambda_+>2|\lambda_-| \,\Bigg\}\,.\label{e: H0 tw-adj}
\end{align}
\thickmuskip=5mu
Now consider the case of $|b - \bar{b}| = 1$ and $\sigma_-$ (\ref{e: sigma tw-adj odd}). We get the $\ker(\sigma_-)$ from (\ref{e: sigma tw-adj zero-form action}) by omitting last two terms. Thus, the additional cohomology
\thickmuskip=2mu
\begin{align}
&H^0(\sigma_-)=\ker(\sigma_-) \vert_{\deg \Lambda = 0} \supset \Bigg\{ \lambda_{-} \in 2\mathbb{Z}\,,\, \lambda_{+} \in 2\mathbb{Z}:\quad \lambda_+= 2(b+\bar{b}+2)> 4\,, \quad \lambda_- = 2(a  - \bar{a})\nonumber\\
&C_{\frac{|\lambda_-|}{2}+b, \bar{b}}(u, \bar{u})\, v^{\frac{|\lambda_-|}{2}}\,, \lambda_- \geq 0 \,\oplus\, C_{b, \frac{|\lambda_-|}{2}+\bar{b}}(u, \bar{u})\, \bar{v}^{\frac{|\lambda_-|}{2}}\,, \lambda_- \leq 0: |b-\bar{b}|=1 \Bigg\}\,.\label{e: H0 tw-adj additional}
\end{align}
\thickmuskip=5mu

We see that there are two distinctive \enquote{bosonic} cohomology classes. The first family in (\ref{e: H0 tw-adj}) exists for any allowed values of the $\mathfrak{sl}(2)\oplus \mathfrak{sl}(2)$ weights, i.e., for any indecomposable $\mathfrak{so}(3,2)$ module parametrized by weights $\lambda_{-} \in 2\mathbb{Z}$, $\lambda_{+} \in 4\mathbb{Z}$, while the second family of 0-cocycles appears only if $\lambda_+ > 2|\lambda_-| \geq 4$. The analogous observations holds for the \enquote{fermionic} case. Representatives of (\ref{e: H0 tw-adj additional}) always exist and the second family in (\ref{e: H0 tw-adj}) vanishes for $\lambda_+ \leq 2|\lambda_-|$. Note that (\ref{e: H0 tw-adj additional}) and the second class in (\ref{e: H0 tw-adj}) correspond to $\lambda_+ > 4$, which implies a non-trivial dependence on $Y_2$ variables responsible for the adjoint module. 

At the boundary value $\lambda_+ = 4$ the second cohomology family of (\ref{e: H0 tw-adj}) and representatives of (\ref{e: H0 tw-adj additional}) vanish. The leftover 0-cocycles reduce to
\begin{equation}
    C^{\vec{\lambda}}_{\frac{|\lambda_-|}{2},  \;0}(u, \bar{u}) v^{\frac{|\lambda_-|}{2}} \quad \oplus \quad C^{\vec{\lambda}}_{0, \;\frac{|\lambda_-|}{2}}(u, \bar{u}) \bar{v}^{\frac{|\lambda_-|}{2}}\,,
\end{equation}
i.e., rectangular $\mathfrak{so}(3,1)$ Young diagram, which can be associated with a generalized Weyl tensor for symmetric massless fields. Indeed, truncation to the unitary submodule within $(tw\otimes adj)$ involves fixing the representation in the adjoint factor to trivial, i.e., eliminating dependence on $y_2, \bar{y}_2$ oscillators from zero-form $C$ fields. At the linear level this can be accomplished via imposing boundary conditions on the fields \cite{Tarusov:2025sre}. The remaining $Y_1$ dependent field $C$ describes a standard twisted-adjoint HS module that encodes a generalized Weyl tensors for spin-$s$ massless particles.  

At $\lambda_+ > 4$ the 0-cocycles have the form of a general two-row Young diagram, so it might be possible to relate some of them with a generalized Weyl tensor for massless and partially massless fields provided the correct gluing with found 2-cocycles in the $(adj \otimes adj)$ one-form sector (\ref{e: cohomology-2}) via $\Upsilon^{\eta}(\Omega_{AdS},\Omega_{AdS},C)$ vertex.

\subsubsection{$H^1(\sigma_-)$}
A general one-form decomposes with respect to irreducible one-forms as
\begin{align}
&\mathcal{E}^{\vec{\lambda}}(u,\bar{u},v,\bar{v})=\sum_{\substack{\lambda_+= 2(b+\bar{b}+2)\\ \lambda_- = 2(a  - \bar{a})}}\bigg(e(\dd, \bar{\dd}) \varepsilon^{\vec{\lambda}}_{A}{}_{a+b+1, \bar{a}+\bar{b}+1}(u,\bar{u})\, v^a \bar{v}^{\bar{a}}+\nonumber\\
&+e(u, \bar{u}) \varepsilon^{\vec{\lambda}}_{B}{}_{a+b-1, \bar{a}+\bar{b}-1}(u,\bar{u})\, v^a \bar{v}^{\bar{a}} + e(\dd, \bar{u}) \varepsilon^{\vec{\lambda}}_{C}{}_{a+b+1, \bar{a}+\bar{b}-1}(u,\bar{u})\, v^a \bar{v}^{\bar{a}}+\nonumber\\
&+e(u, \bar{\dd}) \varepsilon^{\vec{\lambda}}_{D}{}_{a+b-1, \bar{a}-\bar{b}+1}(u,\bar{u})\, v^a \bar{v}^{\bar{a}}\bigg)\,.
\end{align}

The $\text{Im}(\sigma_-)|_{\deg \Lambda = 2}$ of (\ref{e: sigma tw-adj even}) decomposes as
\begin{flalign}
    &\sigma_-\mathcal{E}^{\vec{\lambda}}(u,\bar{u},v,\bar{v}) = \sum_{\substack{\lambda_+= 2(b+\bar{b}+2)\\\lambda_- = 2(a  - \bar{a})}}\Bigg\{\frac{i a \bar{a}}{(a+b+1)\left(\bar{a}+\bar{b}+1\right)}\times \nonumber&&\\
    &\times\left(\frac{\bar{a}+\bar{b}+1}{2} H\left(u, \dd \right) +\frac{a+b+1}{2} H(\bar{u}, \bar{\dd})\right) \varepsilon^{\vec{\lambda}}_{B}{}_{a+b-1, \bar{a}+\bar{b}-1}(u,\bar{u})\, v^{a-1} \bar{v}^{\bar{a}-1}+ \nonumber\\
    &+\frac{i a \bar{a}\left(\bar{a}+\bar{b}+1\right)}{2(a+b+1)\left(\bar{a}+\bar{b}+1\right)} H(\dd, \dd) \varepsilon^{\vec{\lambda}}_{C}{}_{a+b+1, \bar{a}+\bar{b}-1}(u, \bar{u}) v^{a-1} \bar{v}^{\bar{a}-1}+\nonumber&&
\end{flalign}
\begin{flalign}
    &+\frac{i a \bar{a}(a+b+1)}{2(a+b+1)\left(\bar{a}+\bar{b}+1\right)} \bar{H}(\bar{\dd}, \bar{\dd}) \varepsilon^{\vec{\lambda}}_{D}{}_{a+b-1, \bar{a}+\bar{b}+1}(u,\bar{u}) v^{a-1} \bar{v}^{\bar{a}-1}-\nonumber&&\\
    &-\theta(\bar{b}-b) \frac{(b+1)(a+b-\bar{a}+1)}{2(\bar{a}+\bar{b}+1)} \bar{H}(\bar{\dd}, \bar{\dd}) \varepsilon^{\vec{\lambda}}_{A}{}_{a+b+1, \bar{a}+\bar{b}+1}(u,\bar{u}) v^a \bar{v}^{\bar{a}}+\nonumber&&\\
    &+\theta(\bar{b}-b) \frac{(b+1)(a+b-\bar{a}+1)}{2(a+b+1)} H(u, u) \varepsilon^{\vec{\lambda}}_{B}{}_{a+b-1, \bar{a}+ \bar{b}-1}(u,\bar{u})  v^a \bar{v}^{\bar{a}}+\nonumber &&\\
    &+\theta(\bar{b}-b) \frac{(b+1)(a+b-\bar{a}+1)}{(a+b+1)(\bar{a}+\bar{b}+1)}\left(\frac{\bar{a}+\bar{b}+1}{2} H(u, \dd)- \right.\nonumber&&\\
    &\left.-\frac{a+b+1}{2} \bar{H}(\bar{u}, \bar{\dd})\right) \varepsilon^{\vec{\lambda}}_{C}{}_{a+b+1, \bar{a}+\bar{b}-1}(u,\bar{u}) v^a \bar{v}^{\bar{a}}-\nonumber&&\\
    &-\theta(b-\bar{b}) \frac{(\bar{b}+1)(\bar{a}+\bar{b}-a+1)}{2\left(a+b+1\right)} H(\dd, \dd) \varepsilon^{\vec{\lambda}}_{A}{}_{a+b+1, \bar{a}+\bar{b}+1}(u,\bar{u}) v^a \bar{v}^{\bar{a}}+ \nonumber&&\\
    &+\theta(b-\bar{b}) \frac{(\bar{b}+1)(\bar{a}+\bar{b}-a+1)}{2(\bar{a}+\bar{b}+1)} \bar{H}(\bar{u}, \bar{u}) \varepsilon^{\vec{\lambda}}_{B}{}_{a+b-1, \bar{a}+\bar{b}-1}(u,\bar{u}) v^a \bar{v}^{\bar{a}}+\nonumber&&\\
    &+\theta(b-\bar{b}) \frac{(\bar{b}+1)(\bar{a}+\bar{b}-a+1)}{\left(a+b+1\right)(\bar{a}+\bar{b}+1)}\left(\frac{a+b+1}{2} \bar{H}(\bar{u}, \bar{\dd})-\right.\nonumber&&\\
    &\left.-\frac{\bar{a}+\bar{b}+1}{2} H(u, \dd)\right) \varepsilon^{\vec{\lambda}}_{D}{}_{a+b-1, \bar{a}+\bar{b}+1}(u,\bar{u}) v^a \bar{v}^{\bar{a}}\Bigg\}\,.
\end{flalign}

Using linear independence of the irreducible two-forms, we can find the null solutions:
\begin{flalign}
    &\ker(\sigma_-)\vert_{\deg \Lambda = 1} = \Bigg\{ \lambda_{-} \in 2\mathbb{Z}\,, \lambda_{+} \in 4\mathbb{Z}:\quad \lambda_+= 2(b+\bar{b}+2)\geq 4\,, \quad \lambda_- = 2(a  - \bar{a}):\nonumber&&\\
    &\bullet \quad e(u, \bar{u}) \varepsilon^{\vec{\lambda}}_{0, 0} v \, \bar{v}\,;\nonumber&&\\
    &\bullet \quad e(\dd, \bar{\dd}) \varepsilon^{\vec{\lambda}}_{a+\frac{\lambda_+}{4}, \,\bar{a}+\frac{\lambda_+}{4}} v^a \bar{v}^{\bar{a}}\,,\quad \forall a\,,\bar{a} \in \mathbb{Z}_{\geq 0}\,;\nonumber&&\\
    &\bullet \quad e(\dd, \bar{\dd}) \varepsilon^{\vec{\lambda}}_{a+b+1, \bar{a}+\bar{b}+1} v^a \bar{v}^{\bar{a}}: b>\bar{b},\; a=\bar{a}+\bar{b}+1 \;\oplus \nonumber&&\\
    &\oplus \;e(\dd, \bar{\dd}) \varepsilon^{\vec{\lambda}}_{a+b+1, \bar{a}+\bar{b}+1} v^a \bar{v}^{\bar{a}}: b<\bar{b},\;  \bar{a}=a+b+1\,;\nonumber&&\\
    &\bullet \quad e(u, \bar{u}) \varepsilon^{\vec{\lambda}}_{\frac{\lambda_+}{4} - 2, \bar{a}+\frac{\lambda_+}{4} - 2} \bar{v}^{\bar{a}} \,,\quad \forall \bar{a} \in \mathbb{Z}_{\geq 0}\quad\oplus  \quad  e(u, \bar{u}) \varepsilon^{\vec{\lambda}}_{a+ \frac{\lambda_+}{4} - 2, \frac{\lambda_+}{4} - 2} v^a \,,\quad \forall a \in \mathbb{Z}_{\geq 0}\,; \nonumber&&\\
    &\bullet \quad e(u, \bar{u}) \varepsilon^{\vec{\lambda}}_{b-1, \bar{a}+\bar{b}-1} \bar{v}^{\bar{a}}:   b<\bar{b}\,,\; \bar{a}=b+1 \quad \oplus \quad e(u, \bar{u}) \varepsilon^{\vec{\lambda}}_{a+b-1, \bar{b}-1} v^a: b>\bar{b}\,,\; a=\bar{b}+1\,;\nonumber&&\\
    &\bullet \quad e(\dd, \bar{u}) \varepsilon^{\vec{\lambda}}_{b+1, \bar{a}+\bar{b}-1} \bar{v}^{\bar{a}}: b \geq \bar{b}\,, \forall \bar{a} \in \mathbb{Z}_{\geq 0} \quad \oplus \quad e(\dd, \bar{u}) \varepsilon^{\vec{\lambda}}_{a+b+1, \bar{b}-1} v^a\,: b\geq \bar{b}\,, \forall a \in \mathbb{Z}_{\geq 0}\,;\nonumber&&\\
    &\bullet \quad e(u, \bar{\dd}) \varepsilon^{\vec{\lambda}}_{b-1, \bar{a}+\bar{b}+1} \bar{v}^{\bar{a}}: b \leq \bar{b}\,, \forall \bar{a} \in \mathbb{Z}_{\geq 0} \quad \oplus \quad e(u, \bar{\dd}) \varepsilon^{\vec{\lambda}}_{a+b-1, \bar{b}+1} v^a\,: b\leq \bar{b}\,, \forall a \in \mathbb{Z}_{\geq 0}\,;\nonumber&&\\
    &\bullet \quad e(\dd, \bar{u}) \varepsilon^{\vec{\lambda}}_{b+1, \bar{a}+\bar{b}-1} \bar{v}^{\bar{a}}: b < \bar{b}\,, \bar{a} = b + 1 \quad \oplus \quad e(u, \bar{\dd}) \varepsilon^{\vec{\lambda}}_{a+b-1, \bar{b}+1} v^a\,: b>\bar{b}\,, a = \bar{b} + 1\,;\nonumber&&\\
    &\bullet \quad e(u, \bar{u}) \varepsilon^{\vec{\lambda}}_{\frac{\lambda_+}{4} - 1, \,\bar{a}+\frac{\lambda_+}{4} - 1} v\, \bar{v}^{\bar{a}+1}-\frac{2i(\bar{a}+1)\left(\lambda_+ + 4 \bar{a}\right)}{\left(\lambda_+-4 \bar{a} - 4\right)\left(\lambda_+ +2\bar{a}\right)} e(\dd, \bar{u}) \varepsilon^{\vec{\lambda}}_{\frac{\lambda_+}{4} - 1\,,\, \bar{a}+\frac{\lambda_+}{4} - 1} \bar{v}^{\bar{a}}-\nonumber&&
\end{flalign}
\begin{flalign}
&-\frac{2i(\bar{a}+1)\left(\lambda_+ +4 \bar{a} - 4\right)\lambda_+}{\left(\lambda_+ - 4\right)\left(\lambda_+ +4 \bar{a} -4\right)\left(\lambda_+ +2\bar{a}\right)} e(u, \bar{\dd}) \varepsilon^{\vec{\lambda}}_{\frac{\lambda_+}{4} - 1, \bar{a}+\frac{\lambda_+}{4} - 1} \bar{v}^{\bar{a}}:\quad \bar{a} \neq \frac{\lambda_+}{4} - 1\,, \lambda_+ \neq 4\,;\nonumber&&\\
& \bullet \quad e(u, \bar{u}) \varepsilon^{\vec{\lambda}}_{a+\frac{\lambda_+}{4} - 1, \,\frac{\lambda_+}{4} - 1} v^{a+1}\, \bar{v}-\frac{2i(a+1)\left(\lambda_+ + 4 a\right)}{\left(\lambda_+-4 a - 4\right)\left(\lambda_+ +2a\right)} e(u, \bar{\dd}) \varepsilon^{\vec{\lambda}}_{a+\frac{\lambda_+}{4} - 1\,,\, \frac{\lambda_+}{4} - 1} v^a-\nonumber&&\\
&-\frac{2i(a+1)\left(\lambda_+ +4 a - 4\right)\lambda_+}{\left(\lambda_+ - 4\right)\left(\lambda_+ +4 a -4\right)\left(\lambda_+ +2a\right)} e(\dd, \bar{u}) \varepsilon^{\vec{\lambda}}_{a + \frac{\lambda_+}{4} - 1, \frac{\lambda_+}{4} - 1} v^a:\quad a \neq \frac{\lambda_+}{4} - 1\,, \lambda_+ \neq 4\,;\nonumber&&\\
&\bullet \quad e(\dd, \bar{\dd}) \varepsilon^{\vec{\lambda}}_{a+\frac{\lambda_+}{4}+k,\, \bar{a} +\frac{\lambda_+}{4}-k} v^a \bar{v}^{\bar{a}} - \xi e(\dd, \bar{u}) \varepsilon^{\vec{\lambda}}_{a+\frac{\lambda_+}{4}+k, \bar{a}+\frac{\lambda_+}{4}-k} v^{a+1} \bar{v}^{\bar{a}+1}\,: \nonumber&&\\
&k \in\left\{2\,, \dots \,,\frac{\lambda_+}{4}-1\right\}\,,\,\xi = \frac{i\left(\lambda_+ -4 k\right)\left(\lambda_+ -4 k-4a+4 \bar{a}\right)}{16(a+1)(\bar{a}+1)}\,;\nonumber&&\\
&\bullet \quad e(\dd, \bar{\dd}) \varepsilon^{\vec{\lambda}}_{a+\frac{\lambda_+}{4}-k,\, \bar{a} +\frac{\lambda_+}{4} + k} v^a \bar{v}^{\bar{a}}- \xi e(u, \bar{\dd}) \varepsilon^{\vec{\lambda}}_{a+\frac{\lambda_+}{4}-k, \bar{a}+\frac{\lambda_+}{4}+k} v^{a+1} \bar{v}^{\bar{a}+1}\,:\nonumber&&\\
&k \in\left\{2\,, \dots \,,\frac{\lambda_+}{4}-1\right\}\,,\, \xi = \frac{i\left(\lambda_+ -4 k\right)\left(\lambda_+ -4 k-4a+4 \bar{a}\right)}{16(a+1)(\bar{a}+1)}\,\Bigg\}\,. \label{e: one-form kernel tw-adj}
\end{flalign}

Filtering the exact solutions using the same logic as in section \ref{sec: cohomology computation adj-adj}, we arrive at the cohomology
\begin{align}
& H^1(\sigma_-) \supset \Bigg\{\lambda_{-} \in 2\mathbb{Z}\,, \lambda_{+} \in 4\mathbb{Z}:\quad \lambda_+= 2(b+\bar{b}+2)\geq 4\,, \quad \lambda_- = 2(a  - \bar{a}): \nonumber\\
&\bullet \quad e(u, \bar{u}) \varepsilon^{\vec{\lambda}}_{0, 0} v \, \bar{v}\,;\nonumber\\
&\bullet \quad e(u, \bar{u}) \varepsilon^{\vec{\lambda}}_{\frac{\lambda_+}{4} - 2, \bar{a}+\frac{\lambda_+}{4} - 2} \bar{v}^{\bar{a}} \,,\quad \forall \bar{a} \in \mathbb{Z}_{\geq 0}\quad\oplus  \quad  e(u, \bar{u}) \varepsilon^{\vec{\lambda}}_{a+ \frac{\lambda_+}{4} - 2, \frac{\lambda_+}{4} - 2} v^a \,,\quad \forall a \in \mathbb{Z}_{\geq 0}\,; \nonumber\\
&\bullet \quad e(u, \bar{u}) \varepsilon^{\vec{\lambda}}_{b-1, \bar{a}+\bar{b}-1} \bar{v}^{\bar{a}}:   b<\bar{b}\,,\; \bar{a}=b+1 \quad \oplus \quad e(u, \bar{u}) \varepsilon^{\vec{\lambda}}_{a+b-1, \bar{b}-1} v^a: b>\bar{b}\,,\; a=\bar{b}+1\,;\nonumber
\end{align}
\begin{align}
&\bullet \quad e(\dd, \bar{u}) \varepsilon^{\vec{\lambda}}_{b+1, \bar{a}-1} \bar{v}^{\bar{a}}: \bar{a} \geq 1 \quad \oplus \quad e(\dd, \bar{u}) \varepsilon^{\vec{\lambda}}_{b+\bar{b}+1, \bar{b}-1} v^{a}\,: b\geq \bar{b}\,,\, a = \bar{b}\,;\nonumber\\
&\bullet \quad e(u, \bar{\dd}) \varepsilon^{\vec{\lambda}}_{b-1, b+\bar{b}+1} \bar{v}^{\bar{a}}: b \leq \bar{b}\,, \,\bar{a} = b \quad \oplus \quad e(u, \bar{\dd}) \varepsilon^{\vec{\lambda}}_{a-1, \bar{b}+1} v^a\,: a \geq 1\,;\nonumber\\
&\bullet \quad e(\dd, \bar{u}) \varepsilon^{\vec{\lambda}}_{b+1, b+\bar{b}} \bar{v}^{\bar{a}}: b < \bar{b}\,,\, \bar{a} = b+1 \quad \oplus \quad e(u, \bar{\dd}) \varepsilon^{\vec{\lambda}}_{b+\bar{b}, \bar{b}+1} v^{a}\,: b>\bar{b}\,,\, a = \bar{b} + 1\,;\nonumber\\
& \bullet \quad e(u, \bar{u}) \varepsilon^{\vec{\lambda}}_{\frac{\lambda_+}{4} - 1, \, \bar{a}+\frac{\lambda_+}{4} - 1} v\, \bar{v}^{\bar{a}+1}-\frac{2i(\bar{a}+1)\left(\lambda_+ + 4 \bar{a}\right)}{\left(\lambda_+-4 \bar{a} - 4\right)\left(\lambda_+ +2\bar{a}\right)} e(\dd, \bar{u}) \varepsilon^{\vec{\lambda}}_{\frac{\lambda_+}{4} - 1\,,\, \bar{a}+\frac{\lambda_+}{4} - 1} \bar{v}^{\bar{a}}-\nonumber\\
&-\frac{2i(\bar{a}+1)\lambda_+}{\left(\lambda_+ - 4\right)\left(\lambda_+ +2\bar{a}\right)} e(u, \bar{\dd}) \varepsilon^{\vec{\lambda}}_{\frac{\lambda_+}{4} - 1, \bar{a}+\frac{\lambda_+}{4} - 1} \bar{v}^{\bar{a}}:\quad \bar{a} \neq \frac{\lambda_+}{4} - 1\,, \lambda_+ \neq 4\,;\nonumber\\
& \bullet \quad e(u, \bar{u}) \varepsilon^{\vec{\lambda}}_{a+\frac{\lambda_+}{4} - 1, \,\frac{\lambda_+}{4} - 1} v^{a+1}\, \bar{v}-\frac{2i(a+1)\left(\lambda_+ + 4 a\right)}{\left(\lambda_+-4 a - 4\right)\left(\lambda_+ +2a\right)} e(u, \bar{\dd}) \varepsilon^{\vec{\lambda}}_{a+\frac{\lambda_+}{4} - 1\,,\, \frac{\lambda_+}{4} - 1} v^a-\nonumber\\
&-\frac{2i(a+1)\lambda_+}{\left(\lambda_+ - 4\right)\left(\lambda_+ +2a\right)} e(\dd, \bar{u}) \varepsilon^{\vec{\lambda}}_{a + \frac{\lambda_+}{4} - 1, \frac{\lambda_+}{4} - 1} v^a:\quad a \neq \frac{\lambda_+}{4} - 1\,, \lambda_+ \neq 4\,\Bigg\}\,.
\end{align}

Now consider the case of $|b - \bar{b}| = 1$ and $\sigma_-$ (\ref{e: sigma tw-adj odd}). Following the same procedure we come to the conclusion that the additional cohomology is
\begin{align}
& H^1(\sigma_-) \supset \Bigg\{|b-\bar{b}| = 1: \nonumber\\
&\bullet \quad e(u, \bar{u}) \varepsilon_{b-1, \bar{a}+\bar{b}-1} \bar{v}^{\bar{a}}\quad \oplus \quad e(u, \bar{u}) \varepsilon_{a+b-1, \bar{b}-1} v^a\,;\nonumber\\
&\bullet \quad e(u, \bar{\dd}) \varepsilon_{b-1, \bar{a}+\bar{b}+1} \bar{v}^{\bar{a}}\quad \oplus \quad e(u, \bar{\dd}) \varepsilon_{a+b-1, \bar{b}+1} v^a\,: b-\bar{b} = 1 \,;\nonumber\\
&\bullet \quad e(\dd, \bar{u}) \varepsilon_{b+1, \bar{a}+\bar{b}-1} \bar{v}^{\bar{a}}\quad \oplus \quad e(\dd, \bar{u}) \varepsilon_{a+b+1, \bar{b}-1} v^a\,: b-\bar{b} = -1 \,\Bigg\}\,.
\end{align}

One can notice that most of cohomology families exist only if $\lambda_+ > 4$, i.e., an adjoint factor in the $(tw\otimes adj)$ module is not trivial as $b\neq 0$ or $\bar{b}\neq0$. Due to the non-trivial dependence on variables $Y_2$ we see multiple new cohomology classes. Since the r.h.s. of (\ref{e:cov2}) is zero, i.e., the 1-cocycles are absent, that puts some constraints on the primary fields (\ref{e: H0 tw-adj}) and (\ref{e: H0 tw-adj additional}). As at least some of the primary fields can be attributed to the generalized Weyl tensors, a part of found 1-cocycles should encode Bianchi identities. 

In case of $\lambda_+ = 4$, for which the corresponding $\mathfrak{so}(3,2)$ module survives a unitary truncation, we are left with
\begin{gather}
     e(u, \bar{u}) \varepsilon^{\vec{\lambda}}_{0, 0} v \, \bar{v}\,, \\
     e(\dd, \bar{u}) \varepsilon^{\vec{\lambda}}_{1, \bar{a}-1} \bar{v}^{\bar{a}}: \bar{a} \geq 1 \quad \oplus \quad e(u, \bar{\dd}) \varepsilon^{\vec{\lambda}}_{a-1, 1} v^a\,: a \geq 1\,,
\end{gather}
which coincides with the analysis of \cite{Gelfond:2003vh, Gelfond:2013lba}.
The first cohomology is a massless Klein-Gordon equation in a standard HS theory, and the second family encodes equations on generalized Weyl tensors.

We leave a more involved analysis of these new cohomology families beyond $\lambda_+ = 4$ (which are absent in the unitary truncated system) as a topic for further research. The non-unitary zero-forms, including those in non-split modules can still play a role in the full system  as auxiliary fields which might, for example, generate additional coupling constants. For now we move on to the study of gluing between a $(adj\otimes adj)$ one-forms and $(tw\otimes adj)$, $(adj\otimes tw)$ zero-forms implemented via a linear vertices.

\section{Linear vertex analysis}
\label{sec: linear vertext analysis}

In this section we consider the gluing between the $(adj\otimes adj)$ one-forms and  the $(tw\otimes adj)$, $(adj\otimes tw)$ zero-form sectors using the obtained knowledge about $\sigma_-$ cohomology. After setting $\eta_2 = \bar{\eta}_2 = 0$ for simplicity the linear vertex on the r.h.s. of (\ref{e: adj-adj omega}) is
\begin{align}
&\Upsilon^{\eta_1, \bar{\eta}_1}(\Omega,\Omega,C) = \bigg[\nonumber\\
&-\frac{i \eta_1}{2}\bar{H}(\bar{\dd}_1,\bar{\dd}_1) C(0,y_2,\bar{y}_1,\bar{y}_2;\hat{k}_1 | x) * \hat{k}_1 - \frac{i \eta_1}{2}\bar{H}(\bar{\dd}_2, \bar{\dd}_2) C(y_1,0,\bar{y}_1,\bar{y}_2; \hat{k}_2 |x)*\hat{k}_2 \nonumber\\
&-\frac{i \bar{\eta}_1}{2} H(\dd_1, \dd_1) C(y_1,y_2,0,\bar{y}_2;\hat{\bar{k}}_1 | x) * \hat{\bar{k}}_1 - \frac{i \bar{\eta}_1}{2} H(\dd_2,\dd_2) C(y_1,y_2,\bar{y}_1,0; \hat{\bar{k}}_2 |x)*\hat{\bar{k}}_2\bigg] * I_1 I_2\,. \label{e: reduced vertex}   
\end{align}

In section \ref{sec: one-form adj-adj}, with the help of $\mathfrak{sl}^h(2)\oplus \mathfrak{sl}^v(2)$ algebra dual to the $\mathfrak{so}(3,2)$ we have found the decomposition of the $(adj\otimes adj)$ one-form field into the $\mathfrak{so}(3,2)$ irreps and extracted dynamical content encoded in the covariant derivative $\mathcal{D}$. Ideally, one would decompose all terms in $\Upsilon^{\eta_1, \bar{\eta}_1}(\Omega,\Omega,C)$ with respect to the $\mathfrak{so}(3,2)$ irreps basis, i.e., (\ref{e: rank-one adjadj}) and its $\mathfrak{sl}^h(2)\oplus \mathfrak{sl}^v(2)$ descendants, and compare the result with $H^2(\sigma_-)$ (\ref{e: cohomology-2}). Unfortunately, such a decomposition is hard to achieve (zeroing of only one variable in the arguments of fields $C$ in (\ref{e: reduced vertex}) does not lead to the significant reduction in the basis change formulas) and can be performed manually only for the simplest cases. The closest thing we can achieve in a general case is to represent $\Upsilon^{\eta_1, \bar{\eta}_1}(\Omega,\Omega,C)$ in a rank-one like form
\begin{equation}\label{e: simple decomposition}
    F(Y_1,Y_2|x) = \sum_{\substack{a,b,c \\ \bar{a},\bar{b},\bar{c}}} \frac{1}{(a+b)!(\bar{a}+\bar{b})!}\zeta^c \bar{\zeta}\,{}^{\bar{c}} (\dd y_1)^a (\dd y_2)^b (\bar{\dd}\bar{y}_1)^{\bar{a}}(\bar{\dd}\bar{y}_2)^{\bar{b}}F^{\bar{a},\bar{b},\bar{c}}_{a,b,c}(u,\bar{u}|x)\,,
\end{equation}
where
\begin{equation}
    F^{\bar{a},\bar{b},\bar{c}}_{a,b,c}(u,\bar{u}|x) := F^{\bar{a},\bar{b},\bar{c}}_{a,b,c}{}_{\alpha(a+b),\dot{\alpha}(\bar{a}+\bar{b})}(x)\, u^{\alpha(a+b)} \bar{u}^{\dot{\alpha}(\bar{a}+\bar{b})}\,.
\end{equation}

One can check that
\thickmuskip=1mu 
\medmuskip=1mu
\thinmuskip=1mu
\begin{align}
& \bar{H}\left(\bar{\dd}_1, \bar{\dd}_1\right) C\left(0, y_2, \bar{y}_1, \bar{y}_2; \hat{k}_1 |x\right) * \hat{k}_1=\sum_{\substack{\bar{a}, \bar{b}, \bar{c}\\b}} \frac{1}{b!(\bar{a}+\bar{b})!}\left\{\frac{\bar{a}(\bar{a}-1)(\bar{a}+\bar{b}+\bar{c})(\bar{a}+\bar{b}+\bar{c}+1)}{(\bar{a}+\bar{b})(\bar{a}+\bar{b}+1)}\times\right.\nonumber\\
& \left. \times\bar{\zeta}\,{}^{\bar{c}}\left(\bar{\dd} \bar{y}_1\right)^{\bar{a}-2}\left(\dd y_2\right)^b\left(\bar{\dd} \bar{y}_2\right)^{\bar{b}} \bar{H}(\bar{\dd}, \bar{\dd})-\frac{2 \bar{a}\, \bar{c}(\bar{a}+\bar{b}+\bar{c}+1)}{(\bar{a}+\bar{b})(\bar{a}+\bar{b}+2)} \bar{\zeta}\,{}^{\bar{c}-1}\left(\bar{\dd} \bar{y}_1\right)^{\bar{a}-1}\left(\dd y_2\right)^b\left(\bar{\dd} \bar{y}_2\right)^{\bar{b}+1} \bar{H}(\bar{u}, \bar{\dd})+ \right.\nonumber\\
&\left.+\frac{\bar{c}(\bar{c}-1)}{(\bar{a}+\bar{b}+1)(\bar{a}+\bar{b}+2)} \bar{\zeta}\,{}^{\bar{c}-2}\left(\bar{\dd} \bar{y}_1\right)^{\bar{a}}\left(\dd y_2\right)^b\left(\bar{\dd} \bar{y}_2\right)^{\bar{b}+2} \bar{H}(\bar{u}, \bar{u})\right\} C_{0, b, 0}^{\bar{a}, \bar{b}, \bar{c}}\left(u, \bar{u}; \hat{k}_1 | x\right) * \hat{k}_1\,,\label{e: vertex bar 1}
\end{align}
\thickmuskip=5mu 
\medmuskip=4mu
\thinmuskip=3mu

\thickmuskip=1mu 
\medmuskip=1mu
\thinmuskip=1mu
\begin{align}
& H\left(\dd_1, \dd_1\right) C\left(y_1, y_2, 0, \bar{y}_2; \hat{\bar{k}}_1 | x\right) * \hat{\bar{k}}_1=\sum_{\substack{a,b, c \\
\bar{b}}} \frac{1}{\bar{b}!(a+b)!} \left\{\frac{a(a-1)(a+b+c)(a+b+c+1)}{(a+b)(a+b+1)}\times \right.\nonumber\\ 
&\left. \times\zeta^c\left(\dd y_1\right)^{a-2}\left(\dd y_2\right)^b\left(\bar{\dd} \bar{y}_2\right)^{\bar{b}} H(\dd, \dd) 
 -\frac{2 a c(a+b+c+1)}{(a+b)(a+b+2)} \zeta^{c-1}\left(\dd y_1\right)^{a-1}\left(\dd y_2\right)^{b+1}\left(\bar{\dd} \bar{y}_2\right)^{\bar{b}} H(u, \dd) +\right.\nonumber\\
 &\left.+\frac{c(c-1)}{(a+b+1)(a+b+2)} \zeta^{c-2}\left(\dd y_1\right)^a\left(\dd y_2\right)^{b+2}\left(\bar{\dd} \bar{y}_2\right)^{\bar{b}} H(u, u)\right\} C_{a, b, c}^{0, \bar{b}, 0}\left(u, \bar{u} ; \hat{\bar{k}}_1 | x\right) * \hat{\bar{k}}_1\,,\label{e: vertex 1}
\end{align}
\thickmuskip=5mu 
\medmuskip=4mu
\thinmuskip=3mu

\thickmuskip=1mu 
\medmuskip=1mu
\thinmuskip=1mu
\begin{align}
& \bar{H}\left(\bar{\dd}_2, \bar{\dd}_2\right) C\left(y_1, 0, \bar{y}_1, \bar{y}_2; \hat{k}_2 |x\right) * \hat{k}_2=\sum_{\substack{\bar{a}, \bar{b}, \bar{c} \\
a}} \frac{1}{a!(\bar{a}+\bar{b})!}\left\{\frac{\bar{b}(\bar{b}-1)(\bar{a}+\bar{b}+\bar{c})(\bar{a}+\bar{b}+\bar{c}+1)}{(\bar{a}+\bar{b})(\bar{a}+\bar{b}+1)}\times\right.\nonumber\\
& \left. \times\bar{\zeta}\,{}^{\bar{c}}\left(\dd y_1\right)^a\left(\bar{\dd} \bar{y}_1\right)^{\bar{a}}\left(\bar{\dd} \bar{y}_2\right)^{\bar{b}-2} \bar{H}(\bar{\dd}, \bar{\dd})+\frac{2 \bar{b} \bar{c}(\bar{a}+\bar{b}+\bar{c}+1)}{(\bar{a}+\bar{b})(\bar{a}+\bar{b}+2)} \bar{\zeta}\,{}^{\bar{c}-1}\left(\dd y_1\right)^a\left(\bar{\dd} \bar{y}_1\right)^{\bar{a}+1}\left(\bar{\dd} \bar{y}_2\right)^{\bar{b}-1} \bar{H}(\bar{u}, \bar{\dd})+ \right.\nonumber\\
&\left.+\frac{\bar{c}(\bar{c}-1)}{(\bar{a}+\bar{b}+1)(\bar{a}+\bar{b}+2)} \bar{\zeta}\,{}^{\bar{c}-2}\left(\dd y_1\right)^a\left(\bar{\dd} \bar{y}_1\right)^{\bar{a}+2}\left(\bar{\dd} \bar{y}_2\right)^{\bar{b}} \bar{H}(\bar{u}, \bar{u})\right\} C_{a, 0, 0}^{\bar{a}, \bar{b}, \bar{c}}\left(u, \bar{u}; \hat{k}_2 | x\right) * \hat{k}_2\,,\label{e: vertex bar 2}
\end{align}
\thickmuskip=5mu 
\medmuskip=4mu
\thinmuskip=3mu

\thickmuskip=1mu 
\medmuskip=1mu
\thinmuskip=1mu
\begin{align}
& H\left(\dd_2, \dd_2\right) C\left(y_1, y_2, \bar{y}_1, 0; \hat{\bar{k}}_2 | x\right) * \hat{\bar{k}}_2=\sum_{\substack{a,b, c \\
\bar{a}}} \frac{1}{\bar{a}!(a+b)!} \left\{\frac{b(b-1)(a+b+c)(a+b+c+1)}{(a+b)(a+b+1)}\times \right.\nonumber\\ 
&\left. \times\zeta^c \left(\dd y_1\right)^a\left(\bar{\dd} \bar{y}_1\right)^{\bar{a}} \left(\dd y_2\right)^{b-2} H(\dd, \dd) 
 +\frac{2 b c(a+b+c+1)}{(a+b)(a+b+2)} \zeta^{c-1} \left(\dd y_1\right)^{a+1}\left(\bar{\dd} \bar{y}_1\right)^{\bar{a}} \left(\dd y_2\right)^{b-1} H(u, \dd) +\right.\nonumber\\
 &\left.+\frac{c(c-1)}{(a+b+1)(a+b+2)} \zeta^{c-2} \left(\dd y_1\right)^{a+2}\left(\bar{\dd} \bar{y}_1\right)^{\bar{a}} \left(\dd y_2\right)^b H(u, u)\right\} C_{a, b, c}^{\bar{a}, 0, 0}\left(u, \bar{u} ; \hat{\bar{k}}_2 | x\right) * \hat{\bar{k}}_2\,.\label{e: vertex 2}
\end{align}
\thickmuskip=5mu 
\medmuskip=4mu
\thinmuskip=3mu

Despite the complexity of changing basis in $\mathfrak{R}\simeq\Lambda^\bullet(M)\otimes \mathbb{C}[[y_1,y_2,\bar y_1, \bar y_2]]$
\thickmuskip=2.5mu 
\begin{align}
    \zeta^c \,\bar{\zeta}\,{}^{\bar{c}} \,(\dd y_1)^a (\dd y_2)^b (\bar{\dd}\bar{y}_1)^{\bar{a}}(\bar{\dd}\bar{y}_2)^{\bar{b}}&: a,b,c,\bar{a},\bar{b},\bar{c}\in 
\mathbb{Z}_{\geq0}\nonumber\\
&\updownarrow\nonumber\\
E_v^N F_h^M\left(\sum_{p+q=l} \Gamma_{n, m,r}^{p, q} \zeta^p \bar{\zeta}{}^q\right)(\dd \bar{\dd} X)^r\left(\dd y_1\right)^n\left(\bar{\dd} \bar{y}_1\right)^m &: \lambda_h = n+m\,, \lambda_v = 2(r+l)\,, \lambda_i \in 2\mathbb{Z}_{\geq 0}\,,
\end{align}
\thickmuskip=5mu 

we can make important observations on the structure of r.h.s. in (\ref{e: adj-adj omega}). As we will see, $\Upsilon^{\eta_1, \bar{\eta}_1}(\Omega,\Omega,C)$ contains equation 2-cocycles in addition to the Weyl cocycles and non-cohomological terms. Moreover, $(\omega, C)$ field variables chosen in sections \ref{sec: one-form adj-adj} and \ref{sec: zero-form tw-adj} do not lead to the one-to-one gluing between one-form and zero-form $\mathfrak{so}(3,2)$ modules. We conjecture that a proper choice of homotopy procedure (field redefinition) will eliminate non-cohomological terms in the vertex (\ref{e: adj-adj omega}) and provide a clear gluing between $\omega$ and $C$. Leaving the execution of this procedure for later study we consider the structure of $\Upsilon^{\eta_1, \bar{\eta}_1}(\Omega,\Omega,C)$ in a lowest spin examples.

\subsection{Lowest spin examples}

We proceed with the simplest cases of gluing between $\mathfrak{so}(3,2)$ submodules of $(adj\otimes adj)$ and zero-forms $(tw\otimes adj)$ and $(adj \otimes tw)$ based on the $\mathfrak{sl}^h(2)\oplus \mathfrak{sl}^v(2)$ (HW, LW) vectors and their $F_h^M$ descendants with $\lambda_v = 0$ ($\eta_2 = \bar{\eta}_2 = 0$). 

In case of $\lambda_h = 0$ we have one-dimensional $\mathfrak{sl}^h(2)$ module $\{\omega^{\vec{\lambda}}{}^{0}_{0, 0, 0}\}$ and the r.h.s. of (\ref{e: adj-adj omega}) takes a Maxwell-like form
\begin{align}
    &\Upsilon^{\eta_1, \bar{\eta}_1}(\Omega,\Omega,C)\bigg|_{\substack{\lambda_h = 0 \\
    \lambda_v = 0}} = \nonumber\\
    &= -\frac{i \eta_1}{2} \bar{H}(\bar{\dd}, \bar{\dd}) C_{0,0,0}^{2,0,0}\left(u, \bar{u} ; \hat{k}_1 |x\right) * \hat{k}_1
    -\frac{i \bar{\eta}_1}{2} H(\dd, \dd) C_{2,0,0}^{0,0,0}\left(u, \bar{u} ; \hat{\bar{k}}_1 | x\right) * \hat{\bar{k}}_1 \nonumber\\
    &-\frac{i \eta_1}{2} \bar{H}(\bar{\dd}, \bar{\dd}) C_{0, 0,0}^{0, 2, 0}\left(u, \bar{u} ; \hat{k}_2 | x\right) * \hat{k}_2 
    -\frac{i \bar{\eta}_1}{2} H(\dd, \dd) C_{0, 2,0}^{0, 0,0}\left(u, \bar{u} ; \hat{\bar{k}}_2 | x\right) * \hat{\bar{k}}_2 \,.
\end{align}
Each of the terms above belongs to the $H^2(\sigma_-)$ in the $(adj\otimes adj)$ sector. Comparing with (\ref{e: tw-adj HW1}) and (\ref{e: H0 tw-adj}), we see that components of the fields $C$ above are primary with respect to $\sigma_-$ in $(tw\otimes adj)$ and $(adj \otimes tw)$ sectors. Nevertheless, only the sum $C_{0,0,0}^{2,0,0}\left(u, \bar{u} ; \hat{k}_1 |x\right)*\hat{k}_1 + C_{0, 0,0}^{0, 2, 0}\left(u, \bar{u} ; \hat{k}_2 | x\right)*\hat{k}_2$ and its conjugated can be interpreted as the Faraday tensor. Therefore, a genuine Weyl module should reside in a mix of $(tw\otimes adj)$ and $(adj \otimes tw)$ modules. 

If $\lambda_h = 2$, we have a 3-dimensional $\mathfrak{sl}^h(2)$ module $\{\omega^{\vec{\lambda}}\,, F_h\omega^{\vec{\lambda}}\,, F_h^2\omega^{\vec{\lambda}}\}$. Vector $\omega^{\vec{\lambda}}$ corresponds to three Lorenz representations $\{\omega^{\vec{\lambda}}{}^{0}_{1, 1, 0}\,,\omega^{\vec{\lambda}}{}^{0}_{2, 0, 0}\,,\omega^{\vec{\lambda}}{}^{0}_{0, 2, 0}\}$ of grades $G=0, \pm 2$ with respect to $G=|\hat{N}_u-\hat{\bar{N}}_u|+2\hat{N}_\chi$. According to the analysis of section \ref{sec: one-form adj-adj}, primary fields reside in $G = 0$, i.e., $\{\omega^{\vec{\lambda}}{}^{0}_{1, 1, 0}\,, F_h \omega^{\vec{\lambda}}{}^{0}_{1, 1, 0}\,, F_h^2 \omega^{\vec{\lambda}}{}^{0}_{1, 1, 0}\}$, and can be associated with gravity.

The r.h.s. of (\ref{e: adj-adj omega}) projected onto the (HW, LW) sub-sector of weights $(2,0)$ is a gravity like vertex
\thickmuskip=1mu 
\medmuskip=1mu
\thinmuskip=1mu
\begin{align}
& \Upsilon^{\eta_1, \bar{\eta}_1}(\Omega,\Omega,C)\bigg|_{\substack{\lambda_h = 2 \\
    \lambda_v = 0}} = \nonumber\\
&=-\frac{i \eta_1}{4}\left(\bar\dd \bar{y}_1\right)^2 \bar{H}(\bar\dd, \bar\dd) C_{0,0,0}^{4,0,0}\left(u, \bar{u}; \hat{k}_1 | x\right) * \hat{k}_1 - \frac{i \bar{\eta}_1}{4}\left(\dd y_1\right)^2 H(\dd, \dd) C_{4,0,0}^{0,0,0}\left(u, \bar{u}; \hat{\bar{k}}_1 | x\right) * \hat{\bar{k}}_1 \nonumber\\
& -\frac{i \eta_1}{2}\left(\bar\dd \bar{y}_1\right)^2 \bar{H}(\bar{u}, \bar{u}) C_{0,0,0}^{0,0,2}\left(u, \bar{u} ; \hat{k}_2 | x\right) * \hat{k}_2-\frac{i \bar{\eta}_1}{2}\left(\dd y_1\right)^2 H(u, u) C_{0,0,2}^{0,0,0}\left(u, \bar{u} ; \hat{\bar{k}}_2 | x\right) * \hat{\bar{k}}_2 \nonumber\\
& -i \eta_1\left(\dd y_1\right)\left(\bar{\dd} \bar{y}_1\right) \bar{H}(\bar{u}, \bar{\dd}) C_{1,0,0}^{0,1,1}\left(u, \bar{u}; \hat{k}_2 | x\right) * \hat{k}_2 -i \bar{\eta}_1\left(\dd y_1\right)\left(\bar{\dd} \bar{y}_1\right) H(u, \dd) C_{0,1,1}^{1,0,0}\left(u, \bar{u} ; \hat{\bar{k}}_2 | x\right) * \hat{\bar{k}}_2\nonumber\\
& -\frac{i \eta_1}{4}\left(\bar{\dd} \bar{y}_1\right)^2 \bar{H}\left(\bar{u}, \bar{\dd}\right) C_{0,0,0}^{1,1,1}\left(u, \bar{u}; \hat{k}_2 | x\right) * \hat{k}_2 -\frac{i \bar{\eta}_1}{4}\left(\dd y_1\right)^2 H(u, \dd) C_{1,1,1}^{0,0,0}\left(u, \bar{u}; \hat{\bar{k}}_2 | x\right) * \hat{\bar{k}}_2\nonumber\\
& -\frac{i \eta_1}{6}\left(\dd y_1\right)(\bar\dd \bar{y}_1) \bar{H}(\bar{\dd}, \bar{\dd}) C_{1,0,0}^{1,2,0}\left(u, \bar{u} ; \hat{k}_2 | x\right) * \hat{k}_2-\frac{i \bar{\eta}_1}{6}\left(\dd y_1\right)\left(\bar{\dd} \bar{y}_1\right) H(\dd, \dd) C_{1,2,0}^{1,0,0}\left(u, \bar{u} ; \hat{\bar{k}}_2 | x\right) * \hat{\bar{k}}_2 \nonumber\\
& -\frac{i \eta_1}{4}\left(\dd y_1\right)^2 \bar{H}(\bar\dd, \bar{\dd}) C_{2,0,0}^{0,2,0}\left(u, \bar{u}; \hat{k}_2 | x\right) * \hat{k}_2-\frac{i \bar{\eta}_1}{4}\left(\bar{\dd} \bar{y}_1 \right)^2 H(\dd, \dd) C_{0,2,0}^{2,0,0}\left(u, \bar{u} ; \hat{\bar{k}}_2 | x\right) * \hat{\bar{k}}_2 \nonumber\\
& -\frac{i \eta_1}{24}(\bar{\dd} \bar{y}_1)^2 \bar{H}(\bar{\dd}, \bar{\dd}) C_{0,0,0}^{2,2,0}\left(u, \bar{u} ; \hat{k}_2 | x\right) * \hat{k}_2-\frac{i \bar{\eta}_1}{24}\left(\dd y_1\right)^2 H(\dd, \dd) C_{2,2,0}^{0,0,0}\left(u, \bar{u} ; \hat{\bar{k}}_2 | x\right) * \hat{\bar{k}}_2\,.
\end{align}
\thickmuskip=5mu 
\medmuskip=4mu
\thinmuskip=3mu
One can see that the components of field $C$ along $\hat{k}_1$ and $\hat{\bar{k}}_1$ produce a Weyl cohomology. The terms with Klein operators $\{\hat{k}_2\,,\hat{\bar{k}}_2\}$ are divided into the three groups: terms that can be complemented to a 2-cocycle; $\sigma_-$-exact terms; not $\sigma_-$-closed terms. The presence of the latter two groups shows that the choice of a standard shifted homotopy in \cite{Tarusov:2025sre} is not the best one for studying the gluing between one-form and zero-form sectors. At the moment, the appropriate choice of homotopy is unknown, so we obtain the redefinition of the field variables manually for a simple case of low spin. Such  redefinitions involve equations (\ref{e:cov2}) and (\ref{e:cov3}) rewritten in terms of rank-one like fields $C^{\bar{a},\bar{b},\bar{c}}_{a,b,c}(u,\bar{u};\hat{K}|x)$. The equations for $(tw\otimes adj)$ and $(adj \otimes tw)$ zero-forms are provided in appendix \ref{sec: rank-one zero-form eq}. Note that these equations do not resemble (\ref{e: tw-adj spin-tensor equation}), as (\ref{e: tw-adj spin-tensor equation}) has been derived for the indecomposable $\mathfrak{so}(3,2)$ module, but now we are dealing with modules $C$ that are a resummation of indecomposable modules. 

At the $G = 0$ the r.h.s. can be expressed in a $\sigma_-$-exact form and absorbed by the redefinition of $\{\omega^{\vec{\lambda}}{}^{0}_{2, 0, 0}(u,\bar{u}|x)\,,\omega^{\vec{\lambda}}{}^{0}_{0, 2, 0}(u,\bar{u}|x)\}$:
\begin{align}
    &\omega^{\vec{\lambda}}{}^{0}_{0, 2, 0}(u,\bar{u}|x) = \tilde{\omega}^{\vec{\lambda}}{}^{0}_{0, 2, 0}(u,\bar{u}|x) -\frac{i \eta_1}{4}  e(\dd, \bar{u}) C_{1,0,0}^{0,1,1}\left(u, \bar{u}; \hat{k}_2 | x\right) * \hat{k}_2 \nonumber\\
    &-\frac{3i \bar{\eta}_1}{4}  e(\dd, \bar{u}) C_{0,1,1}^{1,0,0}\left(u, \bar{u}; \hat{\bar{k}}_2 |x\right) * \hat{\bar{k}}_2+\frac{i \eta_1}{3} e(\partial, \bar{\partial}) C_{1,0,0}^{1,2,0}\left(u, \bar{u}; \hat{k}_2 | x\right) * \hat{k}_2\,,
\end{align}
\begin{align}
    &\omega^{\vec{\lambda}}{}^{0}_{2, 0, 0}(u,\bar{u}|x) = \tilde{\omega}^{\vec{\lambda}}{}^{0}_{2, 0, 0}(u,\bar{u}|x) -\frac{3i \eta_1}{4}  e(u, \bar{\partial}) C_{1,0,0}^{0,1,1}\left(u, \bar{u}; \hat{k}_2 | x\right) * \hat{k}_2 \nonumber\\
    &-\frac{i \bar{\eta}_1}{4}  e(u, \bar{\partial}) C_{0,1,1}^{1,0,0}\left(u, \bar{u} ; \hat{\bar{k}}_2 | x\right) * \hat{\bar{k}}_2 + \frac{i \bar{\eta}_1}{3} e(\dd, \bar{\partial}) C_{1,2,0}^{1,0,0}\left(u, \bar{u}; \hat{\bar{k}}_2 | x\right) * \hat{\bar{k}}_2\,.
\end{align}
It leads to the equations
\begin{equation}
    D_L \omega^{\vec{\lambda}}{}^{0}_{1, 1, 0}(u,\bar{u}|x) + e(u,\bar\dd)\tilde{\omega}^{\vec{\lambda}}{}^{0}_{0, 2, 0}(u,\bar{u}|x) + e(\dd,\bar u)\tilde{\omega}^{\vec{\lambda}}{}^{0}_{2, 0, 0}(u,\bar{u}|x) = 0\,,
\end{equation}
\begin{align}
    &D_L \tilde{\omega}^{\vec{\lambda}}{}^{0}_{2, 0, 0}(u,\bar{u}|x) + e(u,\bar{\dd}) \omega^{\vec{\lambda}}{}^{0}_{1, 1, 0}(u,\bar{u}|x) = - \frac{i\bar{\eta}_1}{2} H(\dd, \dd) C_{4,0,0}^{0,0,0}\left(u, \bar{u}; \hat{\bar{k}}_1 | x\right) *\hat{\bar{k}}_1+\nonumber\\
    &+\frac{\bar{\eta}_1}{4}  H(\dd, \dd) C_{1,3,0}^{0,0,1}\left(u, \bar{u}; \hat{\bar{k}}_2 | x\right) * \hat{\bar{k}}_2    -\frac{i\bar{\eta}_1}{4}  H(\dd, \dd) C_{2,2,0}^{0,0,0}\left(u, \bar{u}; \hat{\bar{k}}_2 | x\right) * \hat{\bar{k}}_2 + \nonumber\\
    &+ \frac{1}{2}\bar{H}(\bar\dd, \bar\dd)\left[\eta_1 C_{1,1,0}^{0,2,1}\left(u, \bar{u}; \hat{k}_2 | x\right) * \hat{k}_2+ \bar{\eta}_1 C_{0,2,1}^{1,1,0}\left(u, \bar{u}; \hat{\bar{k}}_2 |x\right) * \hat{\bar{k}}_2\right] \nonumber\\
    &-\frac{3i}{4} \bar{H}(\bar{\dd}, \bar{\dd})\left[\eta_1 C_{2,0,0}^{0,2,0}\left(u, \bar{u}; \hat{k}_2 | x\right) * \hat{k}_2+\bar{\eta}_1 C_{0,2,0}^{2,0,0}\left(u, \bar{u}; \hat{\bar k}_2 | x\right)*\hat{\bar k}_2\right] \nonumber\\
    &- \frac{3i}{2} H(u, u)\left[\eta_1 C_{0,0,0}^{0,0,2}\left(u, \bar{u} ; \hat{k}_2 | x\right) * \hat{k}_2+\bar{\eta}_1 C_{0,0,2}^{0,0,0}\left(u, \bar{u}; \hat{\bar{k}}_2 | x\right) * \hat{\bar{k}}_2\right]\,, \label{e: grav HW 1}
\end{align}
\begin{align}
    &D_L \tilde{\omega}^{\vec{\lambda}}{}^{0}_{0, 2, 0}(u,\bar{u}|x) + e(\dd,\bar{u}) \omega^{\vec{\lambda}}{}^{0}_{1, 1, 0}(u,\bar{u}|x) = - \frac{i\eta_1}{2} \bar{H}(\bar{\dd}, \bar{\dd}) C_{0,0,0}^{4,0,0}\left(u, \bar{u}; \hat{k}_1 | x\right) *\hat{k}_1+ \nonumber\\
    &+\frac{\eta_1}{4}  \bar{H}(\bar{\dd}, \bar{\dd}) C_{0,0,1}^{1,3,0}\left(u, \bar{u}; \hat{k}_2 | x\right) * \hat{k}_2 
    -\frac{i\eta_1}{4}  \bar{H}(\bar{\dd}, \bar{\dd}) C_{0,0,0}^{2,2,0}\left(u, \bar{u}; \hat{k}_2 | x\right) * \hat{k}_2 \nonumber\\
    &+ \frac{1}{2}H(\dd, \dd)\left[\eta_1 C_{1,1,0}^{0,2,1}\left(u, \bar{u}; \hat{k}_2 | x\right) * \hat{k}_2+ \bar{\eta}_1 C_{0,2,1}^{1,1,0}\left(u, \bar{u}; \hat{\bar{k}}_2 |x\right) * \hat{\bar{k}}_2\right] \nonumber\\
    &-\frac{3i}{4} H(\dd, \dd)\left[\eta_1 C_{2,0,0}^{0,2,0}\left(u, \bar{u}; \hat{k}_2 | x\right) * \hat{k}_2+\bar{\eta}_1 C_{0,2,0}^{2,0,0}\left(u, \bar{u}; \hat{\bar k}_2 | x\right)*\hat{\bar k}_2\right] \nonumber\\
    &- \frac{3i}{2} \bar{H}(\bar{u}, \bar{u})\left[\eta_1 C_{0,0,0}^{0,0,2}\left(u, \bar{u} ; \hat{k}_2 | x\right) * \hat{k}_2+\bar{\eta}_1 C_{0,0,2}^{0,0,0}\left(u, \bar{u}; \hat{\bar{k}}_2 | x\right) * \hat{\bar{k}}_2\right]\,. \label{e: grav HW 2}
\end{align}
Comparing with $H^2(\sigma_-)$ (\ref{e: cohomology-2 spin-tensor}) we see that all 2-cocycles are present on the r.h.s. These equations are similar to the decomposition of the Riemann tensor into the Weyl tensor, the traceless Ricci tensor and the scalar curvature. 

We can perform the same $\omega$ redefinition procedure for a set of vectors $F_h \omega^{\vec{\lambda}}$ where the relevant part of vertex $\Upsilon^{\eta_1, \bar{\eta}_1}(\Omega,\Omega,C)$ is
\thickmuskip=1mu 
\medmuskip=1mu
\thinmuskip=1mu
\begin{flalign}
    & \Upsilon^{\eta_1, \bar{\eta}_1}(\Omega,\Omega,C)=\nonumber&&\\
    &-\frac{i \eta_1}{2}\left(\bar{\dd} \bar{y}_1\right)\left(\dd y_2\right) \bar{H}(\bar{\dd}, \bar{\dd}) C_{0,1,0}^{3,0,0}\left(u, \bar{u}; \hat{k}_1 | x\right) * \hat{k}_1-\frac{i \eta_1}{2}\left(\dd y_1\right)\left(\bar\dd \bar{y}_2\right) \bar{H}(\bar{\dd}, \bar{\dd}) C_{1,0,0}^{0,3,0}\left(u, \bar{u} ; \hat{k}_2 | x\right) * \hat{k}_2\nonumber&&\\
    &-\frac{i \bar{\eta}_1}{2}\left(\dd y_1\right)\left(\bar\dd \bar{y}_2\right) H(\dd, \dd) C_{3,0,0}^{0,1,0}\left(u, \bar{u} ; \hat{\bar{k}}_1 | x\right) * \hat{\bar{k}}_1 -\frac{i \bar{\eta}_1}{2}\left(\bar{\dd} \bar{y}_1\right)\left(\dd y_2\right) H(\dd, \dd) C_{0,3,0}^{1,0,0}\left(u, \bar{u}; \hat{\bar{k}}_2 | x\right) * \hat{\bar{k}}_2\nonumber&&\\
    &-\frac{i \eta_1}{8}(\bar\dd \bar{y}_1)\left(\bar{\dd} \bar{y}_2\right) \bar{H}(\bar{\dd}, \bar{\dd}) C_{0,0,0}^{3,1,0}\left(u, \bar{u} ; \hat{k}_1 | x\right) * \hat{k}_1+\frac{i \eta_1}{2}\left(\bar\dd \bar{y}_1\right)\left(\bar{\dd} \bar{y}_2\right) \bar{H}(\bar{u}, \bar{\dd}) C_{0,0,0}^{2,0,1}\left(u, \bar{u} ; \hat{k}_1 | x\right) * \hat{k}_1 \nonumber&&\\
    & -\frac{i \eta_1}{8}\left(\bar{\dd} \bar{y}_1\right)\left(\bar{\dd} \bar{y}_2\right) \bar{H}(\bar{\dd}, \bar{\dd}) C_{0,0,0}^{1,3,0}\left(u, \bar{u} ; \hat{k}_2 | x\right) * \hat{k}_2-\frac{i \eta_1}{2}\left(\bar{\dd} \bar{y}_1\right)\left(\bar{\dd} \bar{y}_2\right) \bar{H}\left(\bar{u}, \bar{\dd}\right) C_{0,0,0}^{0,2,1}\left(u, \bar{u}; \hat{k}_2 | x\right) * \hat{k}_2\nonumber&&
\end{flalign}
\begin{flalign}
    &-\frac{i \bar{\eta}_1}{8}\left(\dd y_1\right)\left(\dd y_2\right) H(\dd, \dd) C_{3,1,0}^{0,0,0}\left(u, \bar{u} ; \hat{\bar{k}}_1 | x\right) * \hat{\bar{k}}_1 +\frac{i \bar{\eta}_1}{2}\left(\dd y_1\right)\left(\dd y_2\right) H(u, \dd) C_{2,0,1}^{0,0,0}\left(u, \bar{u} ; \hat{\bar k}_1 | x\right) * \hat{\bar k}_1\nonumber\\
    &-\frac{i \bar{\eta}_1}{8}\left(\dd y_1\right)\left(\dd y_2\right) H(\dd, \dd) C_{1,3,0}^{0,0,0}\left(u, \bar{u} ; \hat{\bar k}_2 | x\right) * \hat{\bar{k}}_2-\frac{i \bar{\eta}_1}{2}\left(\dd y_1\right)\left(\dd y_2\right) H(u, \dd) C_{0,2,1}^{0,0,0}\left(u, \bar{u} ; \hat{\bar{k}}_2 | x\right) * \hat{\bar{k}}_2\,.
\end{flalign}
\thickmuskip=5mu 
\medmuskip=4mu
\thinmuskip=3mu

After some calculations using the fact that
\begin{alignat}{2}
    \left(\dd y_1\right)\left(\dd y_2\right) &= \frac{1}{2}F_h\left(\dd y_1\right)\left(\dd y_1\right)\,, &\quad \left(\bar{\dd} \bar{y}_1\right)\left(\bar{\dd} \bar{y}_2\right) &= \frac{1}{2}F_h \left(\bar{\dd} \bar{y}_1\right)\left(\bar{\dd} \bar{y}_1\right)\,, \\
    \left(\dd y_1\right) \left(\bar{\dd} \bar{y}_2\right) &= \frac{1}{2}\left\{F_h \left(\dd y_1\right) \left(\bar{\dd} \bar{y}_1\right) + \left(\dd \bar{\dd} X\right)\right\}\,, &\quad \left(\dd y_2\right) \left(\bar{\dd} \bar{y}_1\right) &= \frac{1}{2}\left\{F_h \left(\dd y_1\right) \left(\bar{\dd} \bar{y}_1\right) - \left(\dd \bar{\dd} X\right)\right\}
\end{alignat}
we obtain 
\begin{equation}
    D_L \omega^{\vec{\lambda}}{}^{0}_{1, 1, 0}(u,\bar{u}|x) + e(u\,,\bar\dd)\tilde{\omega}^{\vec{\lambda}}{}^{0}_{0, 2, 0}(u,\bar{u}|x) + e(\dd\,,\bar u)\tilde{\omega}^{\vec{\lambda}}{}^{0}_{2, 0, 0}(u,\bar{u}|x) = 0\,,
\end{equation}
\begin{align}
    &D_L \tilde{\omega}^{\vec{\lambda}}{}^{0}_{2, 0, 0}(u,\bar{u}|x) + e(u, \bar{\dd})\omega^{\vec{\lambda}}{}^{0}_{1, 1, 0}(u,\bar{u}|x) = \nonumber\\
    &= -\frac{i \bar{\eta}_1}{4}  H(\dd, \dd)\left\{C_{3,1,0}^{0,0,0}\left(u, \bar{u} ; \hat{\bar{k}}_1 |x\right) * \hat{\bar{k}}_1 + C_{1,3,0}^{0,0,0}\left(u, \bar{u}; \hat{\bar{k}}_2 | x\right) *\hat{\bar{k}}_2 \right\}\nonumber\\
    &-\frac{\bar{\eta}_1}{2} H(\dd, \dd)\left\{ C_{4,0,0}^{0,0,1}\left(u, \bar{u} ; \hat{\bar{k}}_1 |x\right) * \hat{\bar{k}}_1 - C_{0,4,0}^{0,0,1}\left(u, \bar{u}; \hat{\bar{k}}_2 | x\right) * \hat{\bar{k}}_2\right\}\,,\label{e: gravity on the middle 1}
\end{align}
\begin{align}
    &D_L \tilde{\omega}^{\vec{\lambda}}{}^{0}_{0, 2, 0}(u,\bar{u}|x) + e(\dd, \bar{u})\omega^{\vec{\lambda}}{}^{0}_{1, 1, 0}(u,\bar{u}|x) = \nonumber\\
    &= -\frac{i \eta_1}{4}  \bar{H}(\bar{\dd}, \bar{\dd})\left\{C^{3,1,0}_{0,0,0}\left(u, \bar{u} ; \hat{k}_1 |x\right) * \hat{k}_1 + C^{1,3,0}_{0,0,0}\left(u, \bar{u}; \hat{k}_2 | x\right) *\hat{k}_2 \right\}\nonumber\\
    &-\frac{\eta_1}{2} \bar{H}(\bar{\dd}, \bar{\dd})\left\{C^{4,0,0}_{0,0,1}\left(u, \bar{u} ; \hat{k}_1 |x\right) * \hat{k}_1 - C^{0,4,0}_{0,0,1}\left(u, \bar{u}; \hat{k}_2 | x\right) * \hat{k}_2\right\}\,,\label{e: gravity on the middle 2}
\end{align}
where
\thickmuskip=3mu 
\medmuskip=2mu
\thinmuskip=2mu
\begin{align}
    &\omega^{\vec{\lambda}}{}^{0}_{0, 2, 0}(u,\bar{u}|x) = \tilde{\omega}^{\vec{\lambda}}{}^{0}_{0, 2, 0}(u,\bar{u}|x) +\frac{i \eta_1}{2}  e(\dd, \bar{\dd}) \left\{C_{0,1,0}^{3,0,0}\left(u, \bar{u}; \hat{k}_1 | x\right) * \hat{k}_1 + C_{1,0,0}^{0,3,0}\left(u, \bar{u}; \hat{k}_2 | x\right) * \hat{k}_2 \right\}\,,\\
    &\omega^{\vec{\lambda}}{}^{0}_{2, 0, 0}(u,\bar{u}|x) = \tilde{\omega}^{\vec{\lambda}}{}^{0}_{2, 0, 0}(u,\bar{u}|x) +\frac{i \eta_1}{2}  e(\dd, \bar{\dd}) \left\{C^{0,1,0}_{3,0,0}\left(u, \bar{u}; \hat{\bar{k}}_1 | x\right) * \hat{\bar{k}}_1 + C^{1,0,0}_{0,3,0}\left(u, \bar{u}; \hat{\bar{k}}_2 | x\right) * \hat{\bar{k}}_2 \right\}\,,
\end{align}
\thickmuskip=5mu 
\medmuskip=4mu
\thinmuskip=3mu

provide a necessary change of field variables. Once again comparing to $H^2(\sigma_-)$ (\ref{e: cohomology-2 spin-tensor}) we see that only the Weyl-like 2-cocycles are present. It is important to note that $F_h \omega^{\vec{\lambda}}$ vectors have zero $\mathfrak{sl}^h(2)$ weight. Combined with the spin-$1$ example, it suggests that one-form fields of zero $\mathfrak{sl}^h(2)$ weight, i.e., built out of monomials with equal number of $Y_1$ and $Y_2$ oscillators, do not have equation 2-cocycles on the r.h.s. of (\ref{e: adj-adj omega}).

Summarizing the observations made when considering examples of the simplest cases of gluing between one-forms $\omega$ and zero-forms $C$, we conclude that the r.h.s. of (\ref{e: adj-adj omega}) contains a lot of undesirable non-cohomological terms. The presence of the equation 2-cocycles on the r.h.s. of (\ref{e: adj-adj omega}) signals that some of the copies of massless and partially massless fields are off-shell (respective primary fields are not subject to any differential field equations, see section \ref{sec: Interpretation of cohomology}). In addition to that, there is no one-to-one gluing between one- and zero-forms since $C^{\bar{a},\bar{b},\bar{c}}_{a,b,c}(u,\bar{u};\hat{K}|x)$ on the r.h.s. of (\ref{e: grav HW 1}), (\ref{e: grav HW 2}) and (\ref{e: gravity on the middle 1}), (\ref{e: gravity on the middle 2}) are represented as a sum of (LW, LW) vectors, subsingular vectors and their descendants. This means that a single copy of a massless or partially massless field is attached to a sum of components of some indecomposable $\mathfrak{so}(3,2)$ zero-form submodules of $(tw\otimes adj)$ and $(adj \otimes tw)$. Therefore, if a one-to-one gluing is to be achieved (which is greatly desirable for clear interpretation of the fields in a system), a resummation of zero-form modules into some combinations, including reducible ones, should be performed. Moreover, since the Weyl cocycles in our examples are sourced by zero-form fields $C$ along different Klein operators simultaneously, after resummations within the $(tw\otimes adj)$ and $(adj \otimes tw)$ sectors an additional resummation of fields along different Klein operators is needed. We suppose that it can be performed by a suitable homotopy procedure, which could also be used to eliminate non-cohomological terms, potentially supplemented by the reduction to the invariant subsystem of an automorphism $\hat{k}_{12}\hat{\bar{k}}_{12}*f(Y,Z,I;\hat{K})*\hat{k}_{12}\hat{\bar{k}}_{12}$ in (\ref{e:nonlinear system 1})--(\ref{e:nonlinear system 5}).

It is important to note that the equation cocycles appear due to the non-trivial dependence on the adjoint factor in each of the zero-form modules. If we perform a unitary truncation, i.e., eliminate $Y_i$ oscillators associated with the adjoint module, then all equation 2-cocycles vanish, which imposes differential equations on the primary fields, reproducing the on-shell dynamics of the standard HS theory. As a subject for future study the proper $(\omega, C)$ field variables that lead to the explicit structure of the corresponding one-form and zero-form modules as well as their gluing would allow to extend this analysis for the general non-truncated case.

\section{Conclusion}\label{sec: conclusion}

In this paper we have carried out the $\sigma_-$-cohomological analysis of $B_2$ Coxeter extended $4d$ HS theory. This analysis is necessary to understand the dynamical contents of the $B_2$ CHS as lower order cohomology $H^\bullet(\sigma_-)$ provide an exhaustive description of the gauge parameters, primary fields and gauge-invariant differential operators on primary fields present in the system.

To that end we have described the structure of CHS modules that correspond to zero- and one-form fields of the Klein-parity truncated system in terms of $\mathfrak{so}(3,2)$ indecomposable modules via representations of a commuting $\mathfrak{sl}(2) \oplus \mathfrak{sl}(2)$ algebra. In the case of $(adj \otimes adj)$ module where one-form $\omega$ resides, the algebras form a reductive pair, i.e., a Howe-dual pair, therefore the field $\omega$ decomposes into the direct sum of $\mathfrak{so}(3,2)$ finite-dimensional irreps with multiplicity spaces organized into $\mathfrak{sl}(2)\oplus \mathfrak{sl}(2)$ irreducible modules. These multiplicity spaces are described by (HW, LW) representations of $\mathfrak{sl}(2) \oplus \mathfrak{sl}(2)$ on polynomials of oscillators $Y_i$. For $(adj \otimes tw)$ and $(tw \otimes adj)$ modules, in which zero-forms $C$ take value, it has been found that $\mathfrak{sl}(2) \oplus \mathfrak{sl}(2)$ representations on $Y_i$ polynomials have a much more intricate structure, featuring indecomposable non-split extension modules. The presence of reducible indecomposable $\mathfrak{sl}(2) \oplus \mathfrak{sl}(2)$ modules relaxes the decomposition of the corresponding zero-form module to direct sums of tensor product modules $N_i\otimes M_i$, where $N_i$ is indecomposable (possibly reducible) $\mathfrak{sl}(2)\oplus \mathfrak{sl}(2)$ module and $M_i$ is indecomposable (possibly reducible) $\mathfrak{so}(3,2)$ module. Despite some loss of generality, cohomological analysis can be performed on the subset of (LW, LW) $\mathfrak{sl}(2)\oplus \mathfrak{sl}(2)$ vectors as the spectrum of $(adj \otimes tw)$ and $(tw \otimes adj)$ sectors is fully determined by this subset.

In the three cases of $B_2$ CHS modules $(adj \otimes adj)$, $(adj \otimes tw)$ and $(tw \otimes adj)$ a grading and a $\sigma_-$ operator are defined on the $\mathfrak{so}(3,2)$ indecomposable representations, encoded via auxiliary variables, and the lower orders cohomology are obtained. The straightforward calculation involves finding $\ker(\sigma_-)$ and then filtering out the $\sigma_-$-exact forms. The cohomology representatives are described through functions of auxiliary variables, which makes transitioning back to $Y_i$ variables somewhat involved. In case of $(adj \otimes adj)$ module we have a full list of cohomology representatives as $Y_i$ functions, however in $(adj \otimes tw)$ and $(tw \otimes adj)$ modules the complete list of cocycles as $Y_i$ functions is unknown due to the presence of subsingular $\mathfrak{sl}(2)\oplus \mathfrak{sl}(2)$ vectors.

For the one-forms $\omega$, the sector $(adj\otimes adj)$ is of interest, as the content of $(tw\otimes tw)$ module has been studied in \cite{Gelfond:2003vh,Gelfond:2013lba} and found to contain only non-dynamical fields, i.e., $H^2(\sigma_-) = \{\emptyset\}$. The novel entangled module can be also dismissed as topological, since the non-local transformation coupled with a conjugation breaking change of $Y_i$ variables that relates it to $(tw \otimes tw)$ preserves the cohomological analysis. Therefore, the r.h.s. of equations for $(tw\otimes tw)$ and entangled modules found in \cite{Tarusov:2025sre} can be eliminated by a change of $(\omega, C)$ field variables or the proper choice of the homotopy procedure. The $(adj\otimes adj)$ fields describe massless and partially massless fields as $H^{0,1,2}(\sigma_-)$ match with previously known results for massless \cite{Bychkov:2021zvd} and partially massless \cite{Skvortsov:2009nv} fields. Thus, the conjectured appearance of partially massless fields in $B_2$ CHS has been proven, and the 2-particle field $\omega$ encodes all 1-particle $AdS_4$ fields in an infinite number of copies, i.e., $(adj\otimes adj)$ is an infinitely reducible $\mathfrak{so}(3,2)$ module. While partially massless fields are not unitary in $AdS$, they become non-dynamical after the unitary truncation in the zero-form sector, i.e., the corresponding vertices on the r.h.s. of linearized equations for these $\omega$-fields vanish. Components of zero forms $C$ lie in indecomposable modules and cohomological analysis reveals primary fields and equations not found in the standard HS theory. The interpretation of the additional fields and equations is dependent on the gluing to the known submodules of the one-form $(adj \otimes adj)$ module at the vertices. It can be said, however, that the unitary truncation yields modules equivalent to the standard twisted-adjoint ones and reproduces gluing at the linear vertices identical to that in the standard theory.

The analysis of gluing between one-forms $(adj\otimes adj)$ and zero-forms $(adj \otimes tw)$ and $(tw \otimes adj)$ shows that the linearized vertex contains a lot of undesirable non-$\sigma_-$ cohomological terms. While it demonstrates that the choice of standard homotopy in \cite{Tarusov:2025sre} has proven not to be optimal, it is possible to find the suitable $\omega$ field redefinition by hand, as was done for some spin-2 vertices. Such a redefinition results in $\Upsilon(\Omega,\Omega,C)$ built out of 2-cocycles. Comparison of 2-cocycles and linearized vertices showed that for the special case of zero $\mathfrak{sl}^h(2)$ weight vectors only the Weyl type cohomology is present. In other cases 2-cocycles associated with the equations are present. The non-trivial dependence on oscillators describing the adjoint modules in $(adj \otimes tw)$ and $(tw \otimes adj)$ leads to the appearance of these 2-cocycles. We note that in the current form the vertices admit no one-to-one gluing between one-forms $\omega$ and zero-forms $C$, for example the Weyl cohomology is represented by a combination of components of zero-form fields along different Klein operators. This suggests the necessity of a field redefinition to reorganize the zero-form sector into a set of both irreducible and potentially reducible modules that would be uniquely glued to the $\mathfrak{so}(3,2)$ irreps in one-form sector. Such a field redefinition is a good subject for future study via differential homotopy technique.

Eliminating the dependency on the oscillators responsible for the adjoint module in the $C$-fields achieves both the unitarity of the zero-form module and accomplishes clear one-to-one gluing at the vertices, allowing for interpretation of fields $\omega$ and $C$ and equations for them. While the unitary truncation proposed in \cite{Tarusov:2025sre} accomplishes that at the linear level, we see it as a high priority for further research to extend this truncation to the full nonlinear $B_2$ CHS system.

\section*{Acknowledgement}

The authors are thankful to M.A. Vasiliev for deeply insightful discussions and feedback on the draft. The authors thank Yu.A. Tatarenko, O.A. Gelfond, N.G. Misuna and P.T. Kirakosiants for helpful comments.
The work was supported by the Foundation for the Advancement of Theoretical Physics and Mathematics “BASIS”.

\newpage

\section{Appendix}

\subsection{A: Non-split extensions}\label{sec: nonsplit extension}

Consider the following problem. Let $A$ be an algebra, and let $(V;W)$ be a pair of representations of $A$. We would like to classify representations $U$ of $A$ such that $V$ is a subrepresentation of $U$ and $U/V \simeq W$, i.e., we want to list all short exact sequences
\begin{equation}
    0 \rightarrow V \rightarrow U \rightarrow W \rightarrow 0\,.
\end{equation}
There is an obvious example $U = V \oplus W$, but for a suitable choice of algebra $A$ and pair $(V;W)$ there is an $A$-module $U$ such that the representation map $\rho_U(a)$ has block triangular form for any $a\in A$. Such a module $U$ is called a non-split extension of $W$ by $V$. Solution to this problem reduces to the cohomology analysis for a particular nilpotent linear map $\mathrm{d}: \operatorname{Hom}_\mathbb{C}(W, V) \rightarrow \operatorname{Hom}_\mathbb{C}\left(A, \operatorname{Hom}_\mathbb{C}(W, V)\right)$. Then all non-split extensions are classified up to an isomorphism by the group $\operatorname{Ext}^1_A(W,V) \simeq \ker \d/\text{Im}\, \d$ (\cite{Humphreys2008, etingof2011introductionrepresentationtheory} introduce $\operatorname{Ext}$ as a extension of modules; \cite{Benson_1991, Cartan} use category theory and homological algebra to define $\operatorname{Ext}$).

We start with the case of $\mathfrak{g} = \mathfrak{sl}(2)$. Consider the short exact sequence
\begin{equation}
0 \rightarrow V\left(\lambda\right) \overset{i}{\rightarrow} E \overset{p}{\rightarrow} V\left(\mu\right) \rightarrow 0\,,
\end{equation}
where $V\left(\lambda\right)$ and $V\left(\mu\right)$ are the lowest-weight Verma modules.
A lowest-weight Verma module $V\left(\lambda\right)$ is generated by a vector $w_\lambda$ satisfying:
\begin{equation}
h  w_\lambda=\lambda w_\lambda\,, \quad f  w_\lambda=0\,.
\end{equation}
The module $V\left(\lambda\right)$ has a basis $\left\{w_\lambda, e w_\lambda, e^2 w_\lambda, \ldots\right\}$, with the $\mathfrak{g}$-action
\begin{equation}
\begin{aligned}
& h  e^n w_\lambda=(\lambda+2 n) e^n w_\lambda\,, \\
& f  e^n w_\lambda=-n(\lambda+n-1) e^{n-1} w_\lambda\,. \\
\end{aligned}
\end{equation}

Let $v \in V(\mu)$ be a lowest-weight vector. Choose a preimage $w\in E,\;p(w)=v$. Since $p$ is a $\mathfrak{g}$-homomorphism:
\begin{equation}
p(h  w)=h  p(w)=h  v=\mu v\,.
\end{equation}
Thus $h  w-\mu w \in \operatorname{ker}(p)=i(V(\lambda))$. By adjusting $w$ by an element of $i(V(\lambda))$ we may assume that $h  w = \mu w$. Note that
\begin{equation}
p(f w)=f  p(w)=f  v=0 \Rightarrow f  w \in \operatorname{ker}(p)=i(V(\lambda))\,,
\end{equation}
so $x=f w \in V(\lambda)_{\mu-2}$, i.e., lies in a subspace of $V(\lambda)$ spanned by vectors of weight $\mu-2$, as $f$ lowers weight by $2$. The weight $(\mu-2)$ occurs in $V(\lambda)$ only if $\mu = \lambda + 2k$ for $k \in \mathbb{N}$.

The extension splits if and only if there exists a homomorphism $s: V(\mu) \rightarrow E$ such that $p \circ s=\mathrm{id}_{V(\mu)}$.

\begin{lemma}
$E$ is a split extension iff $x \in \operatorname{Im}\left(f: V(\lambda)_\mu \rightarrow V(\lambda)_{\mu-2}\right)$.
\end{lemma}
\begin{proof}

1) Assume the extension splits. Then $s(v)$ is a lowest-weight vector in $E$ of weight $\mu$, since
\begin{equation}
f s(v)=s(f  v)=0, \quad h s(v)=s(h  v)=\mu s(v)\,.
\end{equation}
Now, $w$ and $s(v)$ both project to $v$, so $w-s(v) \in \operatorname{ker}(p)=V(\lambda)$. Let $u=w-s(v) \in V(\lambda)_\mu$, then
\begin{equation}
f w=f(s(v)+u)=f s(v)+f u=0+f u=f u\,.
\end{equation}
However, $x = f w$, so
\begin{equation}
x=f u \in \operatorname{Im}\left(f: V(\lambda)_\mu \rightarrow V(\lambda)_{\mu-2}\right)\,.
\end{equation}

2) Assume $x \in \operatorname{Im}\left(f: V(\lambda)_\mu \rightarrow V(\lambda)_{\mu-2}\right)$. Then there exists $u \in V(\lambda)_{\mu}:$ $x = fu$. Define $w^\prime = w - u$. Then
\begin{align}
 p\left(w^{\prime}\right)&=p(w)-p(u)=v-0=v\,, \nonumber\\
 h  w^{\prime}&=h (w-u)=\mu w-\mu u=\mu w^{\prime}\,, \nonumber\\
 f  w^{\prime}&=f (w-u)=x-x=0\,.
\end{align}
So $w^\prime$ is a lowest-weight vector in $E$ of weight $\mu$. Define $s:V(\mu)\rightarrow E$ by $s(v) = w^\prime$ and extend to all of $V(\mu)$ by the raising generators $e$. Then $s$ is a homomorphism and $p \circ s=\mathrm{id}_{V(\mu)}$, so the extension splits.

\end{proof}

Consider subspaces $V(\lambda)_\mu=\operatorname{span}\left\{e^k w_\lambda\right\}$ and $V(\lambda)_{\mu-2}=\operatorname{span}\left\{e^{k-1} w_\lambda\right\}$. The map $f$ acts as:
\begin{equation}
f\left(e^k w_\lambda\right)=-k(\lambda+k-1) e^{k-1} w_\lambda\,.
\end{equation}
This map is zero iff:
\begin{equation}
k(\lambda+k-1)=0 \quad \Rightarrow \quad \lambda=1-k \leq 0\,.
\end{equation}
Then
\begin{equation}
\mu=\lambda+2 k=(1-k)+2 k=1+k \quad \Rightarrow \quad \mu = 2 - \lambda\,.
\end{equation}

Therefore, a non-split extension exists iff $\mu = 2 - \lambda$, $\lambda \leq 0$, i.e.,
\begin{equation}\label{e: ext for sl(2)}
    \operatorname{Ext}^1_{\mathfrak{sl}(2)}(V(\mu), V(\lambda)) \not\simeq 0 \,, \quad \text{iff} \quad \mu = 2 - \lambda\,,\,\lambda \leq 0\,.
\end{equation}

Now we can proceed with the $\mathfrak{g} = \mathfrak{sl}(2) \oplus \mathfrak{sl}(2)$. The non-split short exact sequences
\begin{equation}
0 \rightarrow V\left(\lambda_1, \lambda_2\right) \rightarrow E \rightarrow V\left(\mu_1, \mu_2\right) \rightarrow 0
\end{equation}

are classified by $\operatorname{Ext}^1_{\mathfrak{g}}(V(\mu_1, \mu_2), V(\lambda_1, \lambda_2))$.

As $V\left(\lambda_1, \lambda_2\right) \simeq V(\lambda_1) \otimes V(\lambda_2)$ the Künneth formula \cite{Benson_1991, Cartan} gives :
\begin{align}
&\operatorname{Ext}_{\mathfrak{g}}^1\left(V\left(\mu_1, \mu_2\right), V\left(\lambda_1, \lambda_2\right)\right) \simeq\left(\operatorname{Ext}_{\mathfrak{sl}(2)}^{1}(V(\mu_{1}), V( \lambda_{1})) \otimes \operatorname {Hom}_{\mathfrak{sl}(2)}\left(V\left(\mu_2\right), V\left(\lambda_2\right)\right)\right) \oplus\nonumber\\
&\oplus\left(\operatorname{Hom}_{\mathfrak{s} \mathfrak{l}(2)}\left(V\left(\mu_1\right), V\left(\lambda_1\right)\right) \otimes \operatorname{Ext}_{\mathfrak{sl}(2)}^1\left(V\left(\mu_2\right), V\left(\lambda_2\right)\right)\right)\,.    
\end{align}
We have already shown when $\operatorname{Ext}^1_{\mathfrak{sl}(2)}(V(\mu), V(\lambda))$ is non-trivial (\ref{e: ext for sl(2)}), calculation of $\operatorname{Hom}_{\mathfrak{sl}(2)}(V(\mu), V(\lambda))$ is not a complex task.

Let $\varphi \in \operatorname {Hom}_{\mathfrak{sl}(2)}\left(V\left(\mu\right), V\left(\lambda\right)\right)$. Suppose $v_\mu \in V\left(\mu\right)$ is a LW vector. Since $\varphi$ is a homomorphism, we have:
\begin{equation}
    h \varphi(v_\mu) = \varphi(hv_\mu) = \mu\varphi(v_\mu)\,, \quad f\varphi(v_\mu) = \varphi(fv_\mu) = \varphi(0) = 0\,,
\end{equation}
so $\varphi(v_\mu)$ is a singular vector of weight $\mu$ in $V\left(\lambda\right)$. There are two possibilities in a general case:
\begin{itemize}
    \item $V\left(\lambda\right)$ is an irreducible module, i.e., $\lambda \not\in \mathbb{Z}_{\leq 0}$, then $\varphi(v_\mu)$ must be a scalar multiple of $v_\lambda$ and $\mu = \lambda$\,;
    \item $V\left(\lambda\right)$ is a reducible indecomposable module, i.e., $\lambda \in \mathbb{Z}_{\leq 0}$. Then there is an additional singular vector $e^{-\lambda + 1}v_\lambda$ of weight $(-\lambda+2)$ and $\mu = 2 - \lambda$\,.
\end{itemize}

However, in the context of section \ref{sec: zero-form tw-adj} both weights of $\mathfrak{sl}(2)\oplus \mathfrak{sl}(2)$ can not be non-positive simultaneously (\ref{e: tw-adj weights sum}). As non-trivial $\operatorname{Ext}^1_{\mathfrak{sl}(2)}(V(\mu_i), V(\lambda_i))$ forces $\lambda_i \leq 0$, the second weight $\lambda_j\,, j\neq i$ has to be positive, so the extra option in $\operatorname {Hom}_{\mathfrak{sl}(2)}\left(V\left(\mu_j\right), V\left(\lambda_j\right)\right)$ is always discarded.

Therefore, we have derived the statement of formula (\ref{NonSplit})
\begin{equation}
    \operatorname{Ext}_{\mathfrak{g}}^1\left(V\left(\mu_1, \mu_2\right), V\left(\lambda_1, \lambda_2\right)\right) \not\simeq 0\,,  \text{iff}\, \left(\mu_1, \mu_2\right)=\left(2-\lambda_1, \lambda_2\right)\text { or } \left(\mu_1, \mu_2\right)=\left(\lambda_1,2-\lambda_2\right)\,,
\end{equation}
where $\lambda_1 \leq 0, \lambda_2 >0$ in the first case and $\lambda_2 \leq 0, \lambda_1 >0$ in the second case.

\newpage
\subsection{B: Momentum action in $(tw \otimes adj)$ module}
\label{sec: momentum action}

\begin{align}
& e^{\alpha \dot{\alpha}} P_{\alpha \dot{\alpha}} \Psi^{\vec{\lambda}}{}^{\bar{a}, \bar{b}}_{a, b}(Y_1,Y_2)=\sum_{\substack{k=0,\dots, b \\ l=0,\dots,\bar{b}}}\frac{b!\bar{b}!(a+b+1)(\bar{a}+\bar{b}+1)}{i^{k+l} k! l! (b-k)!(\bar{b}-l)! (a+b+l+1)!(\bar{a}+\bar{b}+k+1)!}\Psi_{a, b, 0}^{\bar{a}, \bar{b}, 0} \times\nonumber\\
&\times\bigg[\frac{k}{(a+b+1)\left(\bar{a}+\bar{b}+1\right)} \zeta^l \bar{\zeta}{}^{k-1}\left(\dd y_1\right)^{a+k}\left(\dd y_2\right)^{b-k+1}(\bar{\dd} \bar{y}_1)^{\bar{a}+l+1}(\bar{\dd} \bar{y}_2)^{\bar{b}-l} e\left(u, \bar{u}\right)\nonumber+\\
&+\frac{l}{(a+b+1)\left(\bar{a}+\bar{b}+1\right)} \zeta^{l-1} \bar{\zeta}{}^k\left(\dd y_1\right)^{a+k+1}\left(\dd y_2\right)^{b-k}\left(\bar{\dd} \bar{y}_1\right)^{\bar{a}+l}\left(\bar{\dd} \bar{y}_2\right)^{\bar{b}-l+1} e(u, \bar{u})\nonumber-\\
&-\frac{i}{(a+b+1)(\bar{a}+\bar{b}+1)} \zeta^l\bar{\zeta}{}^{k}\left(\dd y_1\right)^{a+k+1}\left(\dd y_2\right)^{b-k}\left(\bar{\dd} \bar{y}_1\right)^{\bar{a}+l+1}\left(\bar{\dd} \bar{y}_2\right)^{\bar{b}-l} e(u, \bar{u})\nonumber+\\
&+\frac{i\, l\, k}{(a+b+1)(\bar{a}+\bar{b}+1)} \zeta^{l-1} \bar{\zeta}{}^{k-1}\left(\dd y_1\right)^{a+k}\left(\dd y_2\right)^{b-k+1}\left(\bar{\dd} \bar{y}_1\right)^{\bar{a}+l}\left(\bar{\dd} \bar{y}_2\right)^{\bar{b}-l+1} e(u, \bar{u})\nonumber-\\
&-\frac{(a+k)(\bar{b}-l)(\bar{a}+\bar{b}+k+1)}{(a+b+1)(\bar{a}+\bar{b}+1)} \zeta^{l+1} \bar{\zeta}{}^{k}\left(\dd y_1\right)^{a+k-1}\left(\dd y_2\right)^{b-k}\left(\bar{\dd} \bar{y}_1\right)^{\bar{a}+l}\left(\bar{\dd} \bar{y}_2\right)^{\bar{b}-l-1} e(\dd, \bar{\dd})\nonumber-\\
& -\frac{(b-k)(\bar{a}+l)(a+b+l+1)}{(a+b+1)(\bar{a}+\bar{b}+1)} \zeta^l \bar{\zeta}{}^{k+1}\left(\dd y_1\right)^{a+k}\left(\dd y_2\right)^{b-k-1}\left(\bar{\dd} \bar{y}_1\right)^{\bar{a}+l-1}\left(\bar{\dd} \bar{y}_2\right)^{\bar{b}-l} e(\dd, \bar{\dd}) \nonumber-\\
&-\frac{i(b-k)(\bar{b}-l)}{(a+b+1)(\bar{a}+\bar{b}+1)} \zeta^{l+1} \bar{\zeta}{}^{k+1}\left(\dd y_1\right)^{a+k}\left(\dd y_2\right)^{b-k-1}(\bar{\dd} \bar{y}_1)^{\bar{a}+l}(\bar{\dd} \bar{y}_2)^{\bar{b}-l-1} e(\dd, \bar{\dd}) \nonumber+\\
& +\frac{i(a+k)(\bar{a}+l)(a+b+l+1)(\bar{a}+\bar{b}+k+1)}{(a+b+1)(\bar{a}+\bar{b}+1)} \times \nonumber\\
&\times \zeta^l \bar{\zeta}{}^k\left(\dd y_1\right)^{a+k-1}\left(\dd y_2\right)^{b-k}\left(\bar{\dd} \bar{y}_1\right)^{\bar{a}+l-1}\left(\bar{\dd} \bar{y}_2\right)^{\bar{b}-l} e(\dd, \bar{\dd})\nonumber+\\
&+\frac{(\bar{b}-l)(\bar{a}+\bar{b}+k+1)}{(a+b+1)(\bar{a}+\bar{b}+1)} \zeta^l \bar{\zeta}{}^{k}\left(\dd y_1\right)^{a+k}\left(\dd y_2\right)^{b-k+1}(\bar{\dd} \bar{y}_1)^{\bar{a}+l}\left(\bar{\dd} \bar{y}_2\right)^{\bar{b}-l-1} e(u, \bar{\dd})\nonumber-\\
&-\frac{l(\bar{a}+l)}{(a+b+1)(\bar{a}+\bar{b}+1)} \zeta^{l-1} \bar{\zeta}{}^{k+1}\left(\dd y_1\right)^{a+k+1}\left(\dd y_2\right)^{b-k}\left(\bar{\dd} \bar{y}_1\right)^{\bar{a}+l-1}\left(\bar{\dd} \bar{y}_2\right)^{\bar{b}-l} e(u, \bar{\dd})\nonumber-\\
&-\frac{i(\bar{b}-l)}{(a+b+1)(\bar{a}+\bar{b}+1)} \zeta^l \bar{\zeta}{}^{k+1}\left(\dd y_1\right)^{a+k+1}\left(\dd y_2\right)^{b-k}\left(\bar{\dd} \bar{y}_1\right)^{\bar{a}+l}\left(\bar{\dd} \bar{y}_2\right)^{\bar{b}-l-1} e(u, \bar{\dd})\nonumber-\\
&-\frac{i l(\bar{a}+l)(\bar{a}+\bar{b}+k+1)}{(a+b+1)(\bar{a}+\bar{b}+1)} \zeta^{l-1} \bar{\zeta}{}^{k}\left(\dd y_1\right)^{a+k}\left(\dd y_2\right)^{b-k+1}\left(\bar{\dd} \bar{y}_1\right)^{\bar{a}+l-1}(\bar{\dd} \bar{y}_2)^{\bar{b}-l} e(u, \bar{\dd})\nonumber-\\
&-\frac{k(a+k)}{(a+b+1)(\bar{a}+\bar{b}+1)} \zeta^{l+1}\bar{\zeta}{}^{k-1}\left(\dd y_1\right)^{a+k-1}\left(\dd y_2\right)^{b-k}\left(\bar{\dd} \bar{y}_1\right)^{\bar{a}+l+1}\left(\bar{\dd} \bar{y}_2\right)^{\bar{b}-l} e(\dd, \bar{u})\nonumber+\\
&+\frac{(b-k)(a+b+l+1)}{(a+b+1)(\bar{a}+\bar{b}+1)} \zeta^l \bar{\zeta}{}^{k}\left(\dd y_1\right)^{a+k}\left(\dd y_2\right)^{b-k-1}\left(\bar{\dd} \bar{y}_1\right)^{\bar{a}+l}\left(\bar{\dd} \bar{y}_2\right)^{\bar{b}-l+1} e(\dd, \bar{u})\nonumber-\\
&-\frac{i(b-k)}{(a+b+1)(\bar{a}+\bar{b}+1)} \zeta^{l+1} \bar{\zeta}{}^{k}\left(\dd y_1\right)^{a+k}\left(\dd y_2\right)^{b-k-1}\left(\bar{\dd} \bar{y}_1\right)^{\bar{a}+l+1}\left(\bar{\dd} \bar{y}_2\right)^{\bar{b}-l} e(\dd, \bar{u})\nonumber-\\
&-\frac{i k(a+k)(a+b+l+1)}{(a+b+1)(\bar{a}+\bar{b}+1)} \zeta^l\bar{\zeta}{}^{k-1}\left(\dd y_1\right)^{a+k-1}\left(\dd y_2\right)^{b-k}\left(\bar{\dd} \bar{y}_1\right)^{\bar{a}+l}\left(\bar{\dd} \bar{y}_2\right)^{\bar{b}-l+1} e(\dd, \bar{u})\bigg] \times \nonumber\\
&\times \psi_{a+b, \bar{a}+\bar{b}}(u, \bar{u})\,.
\end{align}

\subsection{C: Rank-one like equations on $(adj \otimes tw)$ and $(tw\otimes adj)$ modules}
\label{sec: rank-one zero-form eq}

Module $(adj \otimes tw)$:
\begin{align}
& D_L C_{a, b, c}^{\bar{a}, \bar{b}, \bar{c}}(u, \bar{u}; \hat{K} | x)-(\bar{c}+1) e(u, \bar{u}) C_{a-1, b, c}^{\bar{a}, \bar{b}-1, \bar{c}+1}(u, \bar{u} ; \hat{K} | x) -\nonumber\\
&-(c+1) e(u, \bar{u}) C_{a, b-1, c+1}^{\bar{a}-1, \bar{b}, \bar{c}}(u, \bar{u} ; \hat{K} | x) - \nonumber\\
&-i e(u, \bar{u}) C_{a, b-1, c}^{\bar{a}, \bar{b}-1, \bar{c}}(u, \bar{u}; \hat{K} | x)  +i(c+1)(\bar{c}+1) e(u, \bar{u}) C_{a-1, b, c+1}^{\bar{a}-1, \bar{b}, \bar{c}+1}(u, \bar{u}; \hat{K} | x) + \nonumber\\
&+\frac{(b+1)(\bar{a}+1)(\bar{a}+\bar{b}+\bar{c}+2)}{(a+b+1)(a+b+2)(\bar{a}+\bar{b}+1)(\bar{a}+\bar{b}+2)} e(\dd, \bar{\dd}) C_{a, b+1, c-1}^{\bar{a}+1, \bar{b}, \bar{c}}(u, \bar{u} ; \hat{K} | x) + \nonumber\\
&+\frac{(a+1)(\bar{b}+1)(a+b+c+2)}{(a+b+1)(a+b+2)(\bar{a}+\bar{b}+1)(\bar{a}+\bar{b}+2)} e(\dd, \bar{\dd}) C_{a+1, b, c}^{\bar{a}, \bar{b}+1, \bar{c}-1}(u, \bar{u}; \hat{K} | x) - \nonumber\\
& -\frac{i(a+1)(\bar{a}+1)}{(a+b+1)(a+b+2)(\bar{a}+\bar{b}+1)(\bar{a}+\bar{b}+2)} e(\dd, \bar{\dd}) C_{a+1, b, c-1}^{\bar{a}+1, \bar{b}, \bar{c}-1}(u, \bar{u}; \hat{K} | x) + \nonumber\\
&+\frac{i(b+1)(\bar{b}+1)(a+b+c+2)(\bar{a}+\bar{b}+\bar{c}+2)}{(a+b+1)(a+b+2)(\bar{a}+\bar{b}+1)(\bar{a}+\bar{b}+2)} e(\dd, \bar{\dd}) C_{a, b+1, c}^{\bar{a}, \bar{b}+1, \bar{c}}(u, \bar{u}; \bar{K} | x) + \nonumber\\
& +\frac{(\bar{a}+1)(\bar{a}+\bar{b}+\bar{c}+2)}{(\bar{a}+\bar{b}+1)(\bar{a}+\bar{b}+2)} e(u, \bar{\dd}) C_{a-1, b, c}^{\bar{a}+1, \bar{b}, \bar{c}}(u, \bar{u}; \hat{K} | x)- \nonumber\\
&-\frac{(c+1)(\bar{b}+1)}{(\bar{a}+\bar{b}+1)(\bar{a}+\bar{b}+2)} e(u, \bar{\dd}) C_{a, b-1, c+1}^{\bar{a}, \bar{b}+1, \bar{c}-1}(u, \bar{u}; \hat{K} | x) + \nonumber\\
& +\frac{i(\bar{a}+1)}{(\bar{a}+\bar{b}+1)(\bar{a}+\bar{b}+2)} e(u, \bar{\dd}) C_{a, b-1, c}^{\bar{a}+1, \bar{b}, \bar{c}-1}(u, \bar{u}; \hat{K} | x) + \nonumber\\
&+\frac{i(\bar{b}+1)(c+1)(\bar{a}+\bar{b}+\bar{c}+2)}{(\bar{a}+\bar{b}+1)(\bar{a}+\bar{b}+2)} e(u, \bar{\dd}) C_{a-1, b, c+1}^{\bar{a}, \bar{b}+1, \bar{c}}(u, \bar{u}; \hat{K} | x) - \nonumber\\
&-\frac{(b+1)(\bar{c}+1)}{(a+b+2)(a+b+1)} e(\dd, \bar{u}) C_{a, b+1, c-1}^{\bar{a}, \bar{b}-1, \bar{c}+1}(u, \bar{u}; \hat{K} | x)+ \nonumber\\
&+\frac{(a+1)(a+b+c+2)}{(a+b+2)(a+b+1)} e(\dd, \bar{u}) C_{a+1, b, c}^{\bar{a}-1, \bar{b}, \bar{c}}(u, \bar{u} ; \hat{K} | x) + \nonumber\\
& +\frac{i(a+1)}{(a+b+2)(a+b+1)} e(\dd, \bar{u}) C_{a+1, b, c-1}^{\bar{a}, \bar{b}-1, \bar{c}}(u, \bar{u} ; \hat{K} | x) + \nonumber\\
&+\frac{i(b+1)(\bar{c}+1)(a+b+c+2)}{(a+b+2)(a+b+1)} e(\dd, \bar{u}) C_{a, b+1, c}^{\bar{a}-1, \bar{b}, \bar{c}+1}(u, \bar{u}; \hat{K} | x)=0\,.
\end{align}

\newpage

Module $(tw \otimes adj)$:
\begin{align}
& D_L C_{a, b, c}^{\bar{a}, \bar{b}, \bar{c}}(u, \bar{u}; \hat{K} | x)+(\bar{c}+1) e(u, \bar{u}) C_{a, b-1, c}^{\bar{a}-1, \bar{b}, \bar{c}+1}(u, \bar{u} ; \hat{K} | x)+ \nonumber\\
&+(c+1) e(u, \bar{u}) C_{a-1, b, c+1}^{\bar{a}, \bar{b}-1, \bar{c}}(u, \bar{u} ; \hat{K} | x) - \nonumber\\
&-i e(u, \bar{u}) C_{a-1, b, c}^{\bar{a}-1, \bar{b}, \bar{c}}(u, \bar{u}; \hat{K} | x)  +i(c+1)(\bar{c}+1) e(u, \bar{u}) C_{a, b-1, c+1}^{\bar{a}, \bar{b}-1, \bar{c}+1}(u, \bar{u}; \hat{K} | x) - \nonumber\\
&-\frac{(a+1)(\bar{b}+1)(\bar{a}+\bar{b}+\bar{c}+2)}{(a+b+1)(a+b+2)(\bar{a}+\bar{b}+1)(\bar{a}+\bar{b}+2)} e(\dd, \bar{\dd}) C_{a+1, b, c-1}^{\bar{a}, \bar{b}+1, \bar{c}}(u, \bar{u} ; \hat{K} | x)- \nonumber\\
&-\frac{(b+1)(\bar{a}+1)(a+b+c+2)}{(a+b+1)(a+b+2)(\bar{a}+\bar{b}+1)(\bar{a}+\bar{b}+2)} e(\dd, \bar{\dd}) C_{a, b+1, c}^{\bar{a}+1, \bar{b}, \bar{c}-1}(u, \bar{u}; \hat{K} | x) -\nonumber\\
&-\frac{i(b+1)(\bar{b}+1)}{(a+b+1)(a+b+2)(\bar{a}+\bar{b}+1)(\bar{a}+\bar{b}+2)} e(\dd, \bar{\dd}) C_{a, b+1, c-1}^{\bar{a}, \bar{b}+1, \bar{c}-1}(u, \bar{u}; \hat{K} | x) +\nonumber\\
&+\frac{i(a+1)(\bar{a}+1)(a+b+c+2)(\bar{a}+\bar{b}+\bar{c}+2)}{(a+b+1)(a+b+2)(\bar{a}+\bar{b}+1)(\bar{a}+\bar{b}+2)} e(\dd, \bar{\dd}) C_{a+1, b, c}^{\bar{a}+1, \bar{b}, \bar{c}}(u, \bar{u}; \bar{K} | x) +\nonumber\\
&+\frac{(\bar{b}+1)(\bar{a}+\bar{b}+\bar{c}+2)}{(\bar{a}+\bar{b}+1)(\bar{a}+\bar{b}+2)} e(u, \bar{\dd}) C_{a, b-1, c}^{\bar{a}, \bar{b}+1, \bar{c}}(u, \bar{u}; \hat{K} | x)- \nonumber\\
&-\frac{(c+1)(\bar{a}+1)}{(\bar{a}+\bar{b}+1)(\bar{a}+\bar{b}+2)} e(u, \bar{\dd}) C_{a-1, b, c+1}^{\bar{a}+1, \bar{b}, \bar{c}-1}(u, \bar{u}; \hat{K} | x)-\nonumber\\
&-\frac{i(\bar{b}+1)}{(\bar{a}+\bar{b}+1)(\bar{a}+\bar{b}+2)} e(u, \bar{\dd}) C_{a-1, b, c}^{\bar{a}, \bar{b}+1, \bar{c}-1}(u, \bar{u}; \hat{K} | x) -\nonumber\\
&-\frac{i(\bar{a}+1)(c+1)(\bar{a}+\bar{b}+\bar{c}+2)}{(\bar{a}+\bar{b}+1)(\bar{a}+\bar{b}+2)} e(u, \bar{\dd}) C_{a, b-1, c+1}^{\bar{a}+1, \bar{b}, \bar{c}}(u, \bar{u}; \hat{K} | x) - \nonumber\\
&-\frac{(a+1)(\bar{c}+1)}{(a+b+2)(a+b+1)} e(\dd, \bar{u}) C_{a+1, b, c-1}^{\bar{a}-1, \bar{b}, \bar{c}+1}(u, \bar{u}; \hat{K} | x)+ \nonumber\\
&+\frac{(b+1)(a+b+c+2)}{(a+b+2)(a+b+1)} e(\dd, \bar{u}) C_{a, b+1, c}^{\bar{a}, \bar{b}-1, \bar{c}}(u, \bar{u} ; \hat{K} | x) -\nonumber\\
&-\frac{i(b+1)}{(a+b+2)(a+b+1)} e(\dd, \bar{u}) C_{a, b+1, c-1}^{\bar{a}-1, \bar{b}, \bar{c}}(u, \bar{u} ; \hat{K} | x) -\nonumber\\
&-\frac{i(a+1)(\bar{c}+1)(a+b+c+2)}{(a+b+2)(a+b+1)} e(\dd, \bar{u}) C_{a+1, b, c}^{\bar{a}, \bar{b}-1, \bar{c}+1}(u, \bar{u}; \hat{K} | x)=0\,.
\end{align}

\newpage



\begin{thebibliography}{99}
\parindent=0pt
\parskip=0pt

\bibitem{Coleman:1967ad}
S.~R.~Coleman and J.~Mandula,
Phys. Rev. \textbf{159} (1967), 1251-1256.


\bibitem{Haag:1974qh}
R.~Haag, J.~T.~Lopuszanski and M.~Sohnius,
Nucl. Phys. B \textbf{88} (1975), 257.


\bibitem{Fradkin:1987ks}
E.~S.~Fradkin and M.~A.~Vasiliev,
Phys. Lett. B \textbf{189} (1987), 89-95


\bibitem{Fradkin:1986qy}
E.~S.~Fradkin and M.~A.~Vasiliev,
Nucl. Phys. B \textbf{291} (1987), 141-171


\bibitem{Vasiliev:1988sa}
M.~A.~Vasiliev,
Annals Phys. \textbf{190} (1989), 59-106


\bibitem{Konshtein:1988yg}
S.~E.~Konshtein and M.~A.~Vasiliev,
Nucl. Phys. B \textbf{312} (1989), 402-418


\bibitem{Vasiliev:2004cm}
M.~A.~Vasiliev,
JHEP \textbf{12} (2004), 046
[arXiv:hep-th/0404124 [hep-th]].


\bibitem{Gross:1987ar}
D.~J.~Gross and P.~F.~Mende,
Nucl. Phys. B \textbf{303} (1988), 407-454


\bibitem{Gross:1988ue}
D.~J.~Gross,
Phys. Rev. Lett. \textbf{60} (1988), 1229


\bibitem{Lindstrom:2003mg}
U.~Lindstrom and M.~Zabzine,
Phys. Lett. B \textbf{584} (2004), 178-185
[arXiv:hep-th/0305098 [hep-th]].


\bibitem{Bonelli:2003kh}
G.~Bonelli,
Nucl. Phys. B \textbf{669} (2003), 159-172
[arXiv:hep-th/0305155 [hep-th]].


\bibitem{Sagnotti:2003qa}
A.~Sagnotti and M.~Tsulaia,
Nucl. Phys. B \textbf{682} (2004), 83-116
[arXiv:hep-th/0311257 [hep-th]].


\bibitem{Vasiliev:2018zer}
M.~A.~Vasiliev,
JHEP \textbf{08} (2018), 051
[arXiv:1804.06520 [hep-th]].


\bibitem{Vasiliev:1988xc}
M.~A.~Vasiliev,
Phys. Lett. B \textbf{209} (1988), 491-497.


\bibitem{Joung:2021bhf}
E.~Joung, M.~g.~Kim and Y.~Kim,
JHEP \textbf{12} (2021), 092
[arXiv:2108.05535 [hep-th]].


\bibitem{Iazeolla:2025btr}
C.~Iazeolla, P.~Sundell and B.~C.~Vallilo,
J. Phys. A \textbf{58} (2025) no.36, 365402
[arXiv:2503.14673 [hep-th]].


\bibitem{Misuna:2024ccj}
N.~Misuna,
JHEP \textbf{12} (2024), 090
[arXiv:2402.14164 [hep-th]].


\bibitem{Misuna:2024dlx}
N.~Misuna,
Phys. Lett. B \textbf{870} (2025), 139882
[arXiv:2408.13212 [hep-th]].


\bibitem{Shaynkman:2000ts}
O.~V.~Shaynkman and M.~A.~Vasiliev,
Theor. Math. Phys. \textbf{123} (2000), 683-700.
[arXiv:hep-th/0003123 [hep-th]].


\bibitem{Vasiliev:2001zy}
M.~A.~Vasiliev,
Phys. Rev. D \textbf{66} (2002), 066006
[arXiv:hep-th/0106149 [hep-th]].


\bibitem{Gelfond:2003vh}
O.~A.~Gelfond and M.~A.~Vasiliev,
Theor. Math. Phys. \textbf{145} (2005), 1400-1424
[arXiv:hep-th/0304020 [hep-th]].


\bibitem{Bekaert:2004qos}
X.~Bekaert, S.~Cnockaert, C.~Iazeolla and M.~A.~Vasiliev,
in \textit{1st Solvay Workshop on Higher Spin Gauge Theories}, pp.~132--197 (2004),
[arXiv:hep-th/0503128].

\bibitem{Boulanger:2008up}
N.~Boulanger, C.~Iazeolla and P.~Sundell,
JHEP \textbf{07} (2009), 013
[arXiv:0812.3615 [hep-th]].


\bibitem{Boulanger:2008kw}
N.~Boulanger, C.~Iazeolla and P.~Sundell,
JHEP \textbf{07} (2009), 014
[arXiv:0812.4438 [hep-th]].


\bibitem{Vasiliev:2009ck}
M.~A.~Vasiliev,
Nucl. Phys. B \textbf{829} (2010), 176-224
[arXiv:0909.5226 [hep-th]].


\bibitem{Spirin:2024zgy}
E.~O.~Spirin and M.~A.~Vasiliev,
Phys. Lett. B \textbf{852} (2024), 138625
[arXiv:2401.06933 [hep-th]].


\bibitem{Tatarenko:2025krq}
Y.~A.~Tatarenko,
Nucl. Phys. B \textbf{1021} (2025), 117205
[arXiv:2509.02364 [hep-th]].


\bibitem{Bychkov:2021zvd}
A.~S.~Bychkov, K.~A.~Ushakov and M.~A.~Vasiliev,
Symmetry \textbf{13} (2021) no.8, 1498
[arXiv:2107.01736 [hep-th]].


\bibitem{Skvortsov:2009nv}
E.~D.~Skvortsov,
JHEP \textbf{01} (2010), 106
[arXiv:0910.3334 [hep-th]].


\bibitem{Gelfond:2013lba}
O.~A.~Gelfond and M.~A.~Vasiliev,
JHEP \textbf{10} (2016), 067
[arXiv:1312.6673 [hep-th]].


\bibitem{Tarusov:2025sre}
A.~A.~Tarusov, K.~A.~Ushakov and M.~A.~Vasiliev,
JHEP \textbf{08} (2025), 052
[arXiv:2503.05948 [hep-th]].


\bibitem{Deser:2001pe}
S.~Deser and A.~Waldron,
Phys. Rev. Lett. \textbf{87} (2001), 031601
[arXiv:hep-th/0102166 [hep-th]].


\bibitem{Deser:2001us}
S.~Deser and A.~Waldron,
Nucl. Phys. B \textbf{607} (2001), 577-604
[arXiv:hep-th/0103198 [hep-th]].


\bibitem{Zinoviev:2001dt}
Y.~M.~Zinoviev,
[arXiv:hep-th/0108192 [hep-th]].


\bibitem{Zinoviev:2002ye}
Y.~M.~Zinoviev,
[arXiv:hep-th/0211233 [hep-th]].


\bibitem{Skvortsov:2006at}
E.~D.~Skvortsov and M.~A.~Vasiliev,
Nucl. Phys. B \textbf{756} (2006), 117-147
[arXiv:hep-th/0601095 [hep-th]].


\bibitem{Grigoriev:2020lzu}
M.~Grigoriev, K.~Mkrtchyan and E.~Skvortsov,
Phys. Rev. D \textbf{102} (2020) no.6, 066003
[arXiv:2005.05931 [hep-th]].


\bibitem{Basile:2022mif}
T.~Basile, S.~Dhasmana and E.~Skvortsov,
JHEP \textbf{05} (2023), 136
[arXiv:2212.06226 [hep-th]].


\bibitem{Basile:2024dcs}
T.~Basile and S.~Dhasmana,
JHEP \textbf{12} (2024), 152
[arXiv:2407.11884 [hep-th]].


\bibitem{Zinoviev:2024xta}
Y.~M.~Zinoviev,
JHEP \textbf{04} (2025), 019
[arXiv:2412.04982 [hep-th]].


\bibitem{Zinoviev:2025nlp}
Y.~M.~Zinoviev,
JHEP \textbf{06} (2025), 046
[arXiv:2503.16004 [hep-th]].


\bibitem{Zinoviev:2025tpx}
Y.~M.~Zinoviev,
JHEP \textbf{10} (2025), 132
[arXiv:2507.05744 [hep-th]].


\bibitem{Zinoviev:2025yqa}
Y.~M.~Zinoviev,
JHEP \textbf{10} (2025), 231
[arXiv:2508.06166 [hep-th]].


\bibitem{Zinoviev:2025jff}
Y.~M.~Zinoviev,
JHEP \textbf{01} (2026), 085
[arXiv:2509.18884 [hep-th]].


\bibitem{Fronsdal:1978rb}
C.~Fronsdal,
Phys. Rev. D \textbf{18} (1978), 3624.


\bibitem{Howe:1989}
R.~Howe,
Trans. Amer. Math. Soc. \textbf{313} (1989), 539-570
[Erratum-ibid.\textbf{318}(1990) 823].


\bibitem{Humphreys2008}
J.~E.~Humphreys,
Graduate Studies in Mathematics, Vol.~94, American Mathematical Society (2008),
ISBNs: 978-0-8218-4678-0 (print); 978-1-4704-2120-5 (online)

\bibitem{Shaynkman:2004vu}
O.~V.~Shaynkman, I.~Y.~Tipunin and M.~A.~Vasiliev,
Rev. Math. Phys. \textbf{18} (2006), 823-886
[arXiv:hep-th/0401086 [hep-th]].


\bibitem{Vasiliev:1999ba}
M.~A.~Vasiliev,
[arXiv:hep-th/9910096 [hep-th]].


\bibitem{Vasiliev:2007yc}
M.~A.~Vasiliev,
Nucl. Phys. B \textbf{793} (2008), 469--526,
[arXiv:0707.1085 [hep-th]].

\bibitem{Fierz:1939ix}
M.~Fierz and W.~Pauli,
Proc. Roy. Soc. Lond. A \textbf{173} (1939), 211-232


\bibitem{Singh:1974qz}
L.~P.~S.~Singh and C.~R.~Hagen,
Phys. Rev. D \textbf{9} (1974), 898-909


\bibitem{humphreys2012introduction}
J.~E.~Humphreys,
Vol.~9, Springer Science \& Business Media (2012),
eBook ISBN: 978-1-4612-6398-21

\bibitem{Misuna:2019ijn}
N.~Misuna,
Phys. Lett. B \textbf{798} (2019), 134956
[arXiv:1905.06925 [hep-th]].


\bibitem{Misuna:2020fck}
N.~G.~Misuna,
JHEP \textbf{12} (2021), 172
[arXiv:2012.06570 [hep-th]].


\bibitem{etingof2011introductionrepresentationtheory}
Pavel Etingof, Oleg Golberg, Sebastian Hensel, Tiankai Liu, Alex Schwendner, Dmitry Vaintrob and Elena Yudovina,
[arXiv:0901.0827 [math.RT]].


\bibitem{Benson_1991}
Benson,~D.~J.,
Cambridge Studies in Advanced Mathematics 1991
ISBN:9780521361347


\bibitem{Cartan}
Henri Cartan, Samuel Eilenberg,
Princeton University Press, 2016
ISBN 1400883849, 9781400883844

\end{thebibliography}
\end{document}